%% file: main.tex
\documentclass[sigconf, nonacm]{acmart}
\usepackage[linesnumbered,ruled,vlined]{algorithm2e}
\usepackage{xspace}

\usepackage{subfigure}
\usepackage{textcomp}

\usepackage{stackengine}

\usepackage{array}

\newtheorem{theorem}{Theorem}[section]
\newtheorem{definition}{Definition}[section]
\newtheorem{example}{Example}[section]
\newtheorem{lemma}{Lemma}[section]
\newtheoremstyle{case}{}{}{}{}{}{.}{ }{}
\theoremstyle{case}
\newtheorem{case}{\textbf{Observation}}

\usepackage{algpseudocode}
\usepackage{lipsum} % for the example

\begin{document}
%\title{Reconnecting the estranged relationships: Community Influence Maximization in Evolving Networks}

\title{Reconnecting the Estranged Relationships: Optimizing the Influence Propagation in Evolving Networks} 
% alienated
%%
%% The "author" command and its associated commands are used to define the authors and their affiliations.
\author{Taotao Cai}
\affiliation{%
  \institution{Macquarie University}
  \streetaddress{4 Research Park Dr, Macquarie Park}
  \city{Sydney}
  \country{Australia}
}
\email{taotao.cai@mq.edu.au}

\author{Qi Lei}
\affiliation{
   \institution{Chang'an University}
   \city{Xi'an}
   \country{China}
}
\email{2020024009@chd.edu.cn}

\author{Quan Z. Sheng}
%\orcid{}
\affiliation{%
  \institution{Macquarie University}
  \streetaddress{4 Research Park Dr, Macquarie Park}
  \city{Sydney}
  \country{Australia}
}
\email{michael.sheng@mq.edu.au}

\iffalse
\author{Nur Al Hasan Haldar}
%\orcid{0000-0001-5109-3700}
\affiliation{%
  \institution{The University of Western Australia}
  \city{Perth}
  \country{Australia}
}
\email{nur.haldar@uwa.edu.au}
\fi

\author{Shuiqiao Yang}
\affiliation{%
  \institution{University of New South Wales}
  \city{Sydney}
  \country{Australia}
}
\email{shuiqiao.yang@unsw.edu.au}

\author{Jian Yang}
%\orcid{}
\affiliation{%
  \institution{Macquarie University}
  \streetaddress{4 Research Park Dr, Macquarie Park}
  \city{Sydney}
  \country{Australia}
}
\email{jian.yang@mq.edu.au}

\author{Wei Emma Zhang}
%\orcid{0000-0001-5109-3700}
\affiliation{%
  \institution{The University of Adelaide}
  \city{Adelaide}
  \country{Australia}
}
\email{wei.e.zhang@adelaide.edu.au}

\begin{abstract}
\textit{Influence Maximization} (IM), which aims to select a set of users from a social network to maximize the expected number of influenced users, has recently received significant attention for mass communication and commercial marketing. 
%However, 
Existing research efforts dedicated to the IM problem depend on a strong assumption: the selected seed users are willing to spread the information after receiving benefits from a company or organization. 
In reality, however, some seed users may be reluctant to spread the information, or need to be paid %a 
higher 
%cost 
to be motivated. %we need to pay a high cost to active them. 
Furthermore, the existing IM works pay little attention to 
%catch 
capture 
user's influence propagation in the future period.
In this paper, we target a new research problem, named \textit{\underline{R}econnecting \underline{T}op-\underline{$l$} \underline{R}elationships} (RT$l$R) query, which aims to find $l$ number of previous existing relationships but 
being estranged later, such that reconnecting these relationships will maximize the expected number of influenced users by the given group in a future period. 
We prove that the RT$l$R problem is NP-hard. 
An efficient greedy algorithm is proposed to answer the RT$l$R queries with the influence estimation technique and the well-chosen link prediction method to predict the near future network structure.
We also design a pruning method to reduce unnecessary probing from candidate edges. 
Further, a carefully designed order-based algorithm is 
%also 
proposed to accelerate the RT$l$R queries.
Finally, we conduct extensive experiments on real-world datasets to demonstrate the effectiveness and efficiency of our proposed methods.
\end{abstract}

\maketitle

\iffalse
%%% do not modify the following VLDB block %%
%%% VLDB block start %%%
\pagestyle{\vldbpagestyle}
\begingroup\small\noindent\raggedright\textbf{PVLDB Reference Format:}\\
\vldbauthors. \vldbtitle. PVLDB, \vldbvolume(\vldbissue): \vldbpages, \vldbyear.
\href{https://doi.org/\vldbdoi}{doi:\vldbdoi}
\endgroup
\begingroup
\renewcommand\thefootnote{}\footnote{\noindent
This work is licensed under the Creative Commons BY-NC-ND 4.0 International License. Visit \url{https://creativecommons.org/licenses/by-nc-nd/4.0/} to view a copy of this license. For any use beyond those covered by this license, obtain permission by emailing \href{mailto:info@vldb.org}{info@vldb.org}. Copyright is held by the owner/author(s). Publication rights licensed to the VLDB Endowment. \\
\raggedright Proceedings of the VLDB Endowment, Vol. \vldbvolume, No. \vldbissue\ %
ISSN 2150-8097. \\
\href{https://doi.org/\vldbdoi}{doi:\vldbdoi} \\
}\addtocounter{footnote}{-1}\endgroup
%%% VLDB block end %%%

%%% do not modify the following VLDB block %%
%%% VLDB block start %%%

\ifdefempty{\vldbavailabilityurl}{}{
\vspace{.3cm}
\begingroup\small\noindent\raggedright\textbf{PVLDB Artifact Availability:}\\
The source code, data, and/or other artifacts have been made available at \url{\vldbavailabilityurl}.
\endgroup
}
\fi

%%% VLDB block end %%%
\input{Introduction.tex}

\input{Preliminary.tex}

\input{Problem_definition.tex}

\input{Basic_Method.tex}

\input{Reducing_candidate.tex}

\input{Experiments.tex}

\input{Related_work.tex}
\input{Conclusion.tex}

\bibliographystyle{ACM-Reference-Format}
\bibliography{TaotaoCai2021v4,IJCAI2022}

\iffalse
\appendix
\section{Appendix}
\input{Appendix.tex}
\fi

\end{document}

%% file: Introduction.tex
\section{Introduction} \label{sec:intro}
%Influence Maximization (IM) is the problem of finding a small set of highly influential users such that they will cause the maximum influence spread in a social network. The influence spreads according to an explicit influence propagation model w.r.t the \textit{independent cascade} (IC) model and the \textit{linear threshold} (LT) model~\cite{IM2003}. 

Over the past few decades, the rise of online social networks has brought a transformative effect on the communication and information spread among human beings. Through social media platforms (\textit{e.g., Twitter}), business companies can spread their products information and brand stories to their customers, politicians can deliver their administrative ideas and policies to the public, and researchers can post their upcoming academic seminars information to attract their peers around the world to attend. Motivated by real substantial applications of online social networks, researchers start to keep a watchful eye on \textit{information diffusion}~\cite{brown1987social,IM2003}, as the information could quickly become pervasive through the "\textit{word-of-mouth}" propagation among friends in social networks.  

\textit{Influence Maximization} (IM) is the key algorithmic problem in information diffusion research, which has been extensively studied in recent years. IM aims to find a small set of highly influential users such that they will cause the maximum influence spread in a social network~\cite{IM2003,RIS14,tang2014influence,DBLP:journals/corr/abs-2202-03893}. To fit with different real application scenarios, many variants of the IM problem have been investigated recently, such as
\textit{Topic-aware} IM~\cite{guo2013personalized,yuchenvldb2015,10.1145/3035918.3035952,cai2020target}, \textit{Time-aware} IM~\cite{DBLP:conf/aaai/FengCCZCX14,xie2015dynadiffuse,huang2019finding,singh2021link}, 
\textit{Community-aware} IM~\cite{DBLP:conf/kdd/WangCSX10,yadav2018bridging,DBLP:conf/ijcai/TsangWRTZ19,li2020community}, \textit{Competitive} IM~\cite{lu2015competition,ou2016influence,tsaras2021collective,becker2020balancing}, \textit{Multi-strategies} IM~\cite{DBLP:journals/toc/KempeKT15,DBLP:conf/aaai/ChenZZ20}, and \textit{Out-of-Home IM}~\cite{zhang2020towards,zhang2021minimizing}. However, some critical characteristics of the IM study fail to be fully discussed in existing IM works. We explain these characteristics using the two observations below.
% Community-aware  8844973,yadav2018bridging,LI20181601,

\begin{case}
%Most business companies have a budget for selecting several optimal individual users in social networks, to appeal as much customers as possible to purchase their products in a period of time. 
Some business companies wish their product information would be spread to most of their customers in the period after they spent their budgets on their selected seed users (\textit{e.g., Apple releases its new iPhone every September. They want to find optimal influencers in social networks to appeal to as many users as possible to purchase the new iPhone in the year ahead}). %Existing works mainly focus on the IM study in a static network. 
However, most of the existing IM works modelled the social networks as static graphs, while the topology of social networks often evolves over time in the real world~\cite{chen2015influential, 10.1145/1401890.1401948}. Therefore, the seed users selected currently may not give good performance for influence spread in the following time period due to the evolution of the network. To satisfy Apple's requirement, we would better predict the topology evolution of social networks in the following period and select seed users from the predicted network.
\end{case}

\begin{case}
Existing IM studies dedicated to the influence maximization problem depend on a strong assumption -- the selected seed users will spread the information. However, some of the chosen individual seed users may be unwilling to promote the product information for various reasons. Moreover, most startups and academic groups may not have the budget to motivate the seed users to spread their product or academic activities information. 
\end{case}

\iffalse
\begin{case}
\textit{The mainstream IM researches evaluate the influence diffusion using the Probabilistic models~\cite{IM2003} (\textit{e.g., IC model and LT model}). The core idea of probabilistic model is that the user is influenced with some probability from each influenced neighbor in the network. As mentioned in~\cite{econometric_theory_2010,Weng2013}, a user normally observe the decisions of its neighbours and makes its own decision. i.e., you attend a tutorial if $k$ of your friends do so too. Therefore, we need to further consider how cluster structure in social networks affects the influence spread.} 
\end{case}
\fi

\noindent
\textbf{Our Problem. }
The aforementioned observations motivate us to propose and study a novel research problem, namely \underline{R}econnecting \underline{T}op-\underline{$l$} \underline{R}elationships (RT$l$R). Given a directed evolving graph $\mathcal{G} = \{G_i\}_0^{t-1}$, a parameter $l$, and an institute $\mathcal{U}$ contains a group of users, 
RT$l$R asks for reconnecting a set of $l$ estranged relationships (\textit{e.g., edges that have ever existed in $\mathcal{G}$ while disappearing in the near future snapshot graph $G_t$}). 
Reconnecting the selected edges in RT$l$R query to $G_t$ will maximize the number of influenced users in $G_t$ that are influenced by the members of $\mathcal{U}$.

\begin{figure}[t]
	\centering
		\includegraphics[scale=0.6]{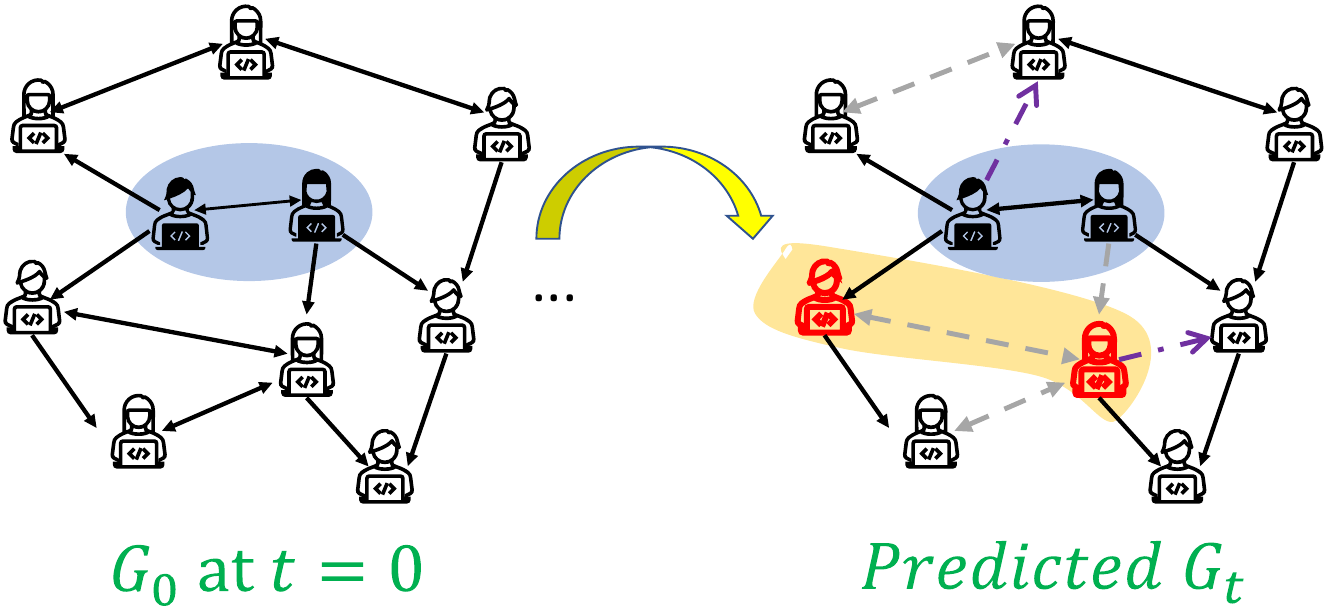}
		\caption{An example of RT$l$R query.}
		\vspace*{-3mm}
	\label{fig:Motivation}

\caption*{\small Note: the given users' group $\mathcal{U}$ are marked as black icons and covered by blue color, $G_0$ is the snapshot of the directed evolving graph $\mathcal{G} = \{G\}_0^{t-1}$ at time 0, and $G_t$ is the predicted graph snapshot of $\mathcal{G}$ at time $t$; the greyish dotted edges in $G_t$ represent the relationship between users exists in $\mathcal{G}$ while disappearing in $G_t$; the purple dotted lines represent the new adding edges in $G_t$; the edge of two red icons which covered by yellow color is the query result of RT$l$R problem.}
%\vspace*{-3mm}
\end{figure}

\begin{example}[Motivation]
%\textit{
LinkedIn\footnote{https://www.linkedin.com/} is a business and employment oriented online social network. It provides a social network platform to allow members to create their profiles and "connect" to each other, representing real-world professional relationships. Members can also post their activity information (e.g., employment Ads) on LinkedIn. The study of RT$l$R can significantly enhance the stickiness of members in LinkedIn without any budgets paid by members or LinkedIn itself.%} 

%\textit{
Figure~\ref{fig:Motivation} presents an evolving social network with ten members and their relationships. Suppose a research group (e.g., black icons)  will host an online virtual academic seminar next month. They post the seminar information on LinkedIn because they wish to attract as many researchers as possible to join their seminar in the month ahead (e.g., $G_t$).
By answering the RT$l$R query, LinkedIn can find out the optimal estranged relationships (e.g., among the greyish dotted edges), in which reconnecting them (e.g., red icons) will maximize the spread of seminar information in the coming month.
To reconnect the estranged relationships, a possible way is to send an email to the related users' platform Inbox and notify them of the recent news of their old friends.
Therefore, the study of RT$l$R query will benefit both users and the social media platform. The members will be more willing to keep active in the network platforms, which provide them a free and efficient information post service.%}
%LinkedIn allows members (both workers and employers) to create profiles and "connect" to each other in an online social network which may represent real-world professional relationships. 
%Suppose a research group in Macquarie University will host an online virtual academic seminar next month, and they wish to attract as many researchers as possible to join their seminar. Obviously, they do not have a budget to be paid to influential users for publicity. Alternatively, it can be wise to spread the seminar information start from the group members, and predict the topology change of social networks from now to next month. Then, they can find a number of optimal edges that disappearing in the predicted social networks but has even existed in social networks. }  (*revise* LinkedIn)
\end{example}

To the best of our knowledge, this is the first IM study that draws the inspiration from the intersection of (1) topology evolving prediction of social networks, and (2) no additional cost. %, and (3) new influence diffusion model that further considering the decision-marking characters of individual users. 
As a result, the following challenges are important to be addressed. 

\vspace{1mm}
\noindent
\textbf{Challenges. }
The first challenge is how to predict the topology of social networks in a specified future period. To deal with this challenge, we adopt the link prediction method~\cite{zhang2021labeling} to predict the network structure evolution in evolving networks. %Moreover, to find the existing TLP algorithm that best matches our RT$l$L problem, we comprehensively review and synthesize a wide spectrum of existing studies on TLP.  \cite{zhang2018link}
\iffalse
The second challenge is how to appropriately evaluate the influence diffusion from a user while considering the decision-making characters of individual users. To manage this challenge, we develop a decision-based influence spread model by considering two critical characteristics -- individual decision-making, and probabilistic models together, by which the practical influence of users can be analyzed in various applications.
\fi
The other challenge is the complexity of RT$l$R query problem.  %posed by the property of the RT$l$R query problem, which is also the core challenge of this paper. 
Unlike traditional IM studies that aim to find Top-$k$ influential users, our RT$l$R focuses on the edges discovery. The existing IM algorithms are not applicable to address the RT$l$R query, and a more detailed analysis is presented in Section~\ref{subsec:IManlyasis}. %Besides, we show that the RT$l$R problem is NP-hard to approximate within any constant factor. 
Thirdly, our RT$l$R query may return different results for different given user groups, while the IM problem only needs to be queried one time to get the most influential users.

To address these algorithmic challenges, we first propose a sketched-based greedy (SBG) algorithm to answer the RT$l$R query of a given group. Besides, a candidate edges reducing method has been proposed to boost the SBG algorithm's efficiency. Furthermore, we carefully designed a novel order-based SBG algorithm to accelerate the RT$l$R query.

\vspace{1mm}
\noindent
\textbf{Contributions. }
We state our major contributions as follows:
\begin{itemize}
    \item We introduce and formally define the problem of \textit{Reconnecting Top-$l$ Relationships} (RT$l$R) for the first time, and explain the motivation of solving the problem with real applications. We also prove that the RT$l$R query problem is NP-hard. 
    \item We propose a sketch-based greedy (SBG) approach to answer the RT$l$R queries. Besides, we present the pruning method to boost the efficiency of the SBG algorithm by reducing the number of candidate edges' probing.
    \item To further accelerate the RT$l$R query, we elaborately design a novel order-based algorithm to answer the RT$l$R query more efficiently.
    \item We conduct extensive experiments to demonstrate the efficiency and effectiveness of our proposed algorithms using real-world datasets.
\end{itemize}

%\vspace{1mm}
\noindent
\textbf{Organization.} %The rest of this paper is organized as follows. 
The remainder of this paper is organized as follows. First, we present the preliminaries in Section~\ref{sec:preliminaries} and formally define the RT$l$R problem in Section~\ref{sec:prob_def}. Then, we propose the sketch-based greedy approach and the  accelerate method in Section~\ref{sec:basic}. We further present a new order-based algorithm to efficiently answer the RT$l$R query in Section~\ref{sec:improved}. After that, the experimental evaluation and results are reported in Section~\ref{sec:exp}. Finally, we review the related works in Section~\ref{sec:related} and conclude this work in Section~\ref{sec:conclusion}.

%% file: Preliminary.tex
\section{Preliminary} \label{sec:preliminaries}
%In this section, we first introduce the concept of link prediction, the new influence diffusion model, and the formulation of the RT$l$L problem. Table~\ref{tab:symbol} summarizes the mathematical notations frequently used throughout this paper.

\begin{table}[t!]
       \caption{Frequently used notations} 
       \vspace{-3mm}
  \begin{center}
    \label{tab:symbol}
    \scalebox{0.9}{
%    \begin{tabular}{|p{1.5cm}<{}|p{5.0cm}<{}|}
    \begin{tabular}{|c|m{6.2cm}|}
    \hline
      \textbf{Notation}         & \textbf{Definition and Description} \\
      \hline
       $\mathcal{G} = \{G\}^{t-1}_0$    & a directed evolving graph \\\hline
       $G_i$             & the snapshot graph of $\mathcal{G}$ at time point $i$ \\\hline
       $V$; $E_i$                 & the vertex set and edge set of $G_i$ \\\hline 
       $G_t$          & the predict snapshot graph of $\mathcal{G}$ at time point $t$ \\\hline
%       $n$; $m_t$                   & number of nodes $(|V|)$ and edges $|E_t|$ in $G_t$  \\\hline
%       $nbr(u,G_i)$               & the set of adjacent vertices of $u$ in $G_i$ \\\hline
%       $d(u)$                 & the degree of $u$ in $G_t$ \\\hline
       $\mathcal{U}$          & the given users group \\\hline
       $I(\mathcal{U},G)$           & the number of activated users in graph $G$ by users in $\mathcal{U}$ \\\hline
       $E(I(\mathcal{U},G))$        & the expected number of users in graph $G$ that influenced by users set $\mathcal{U}$ \\\hline
       $\theta_1$                   & the number of generated RR sets \\\hline
       $S$ ($S_e$)                        & Candidate seed users (edges) set of IM (RT$l$R) query problem \\\hline
%       $S_e$                      & Candidate seed edges set of RT$l$R query problem \\\hline
       $G_t \oplus S_e$           & Reconnecting the edges in $S_e$ of graph $G_t$  \\\hline
       $OPT$ ($OPT^*$)          & the maximum expected spread of any size-$k$ seed users (edges) set of IM (RT$l$R) query problem \\\hline
       $\theta_2$                 & the number of generated sketch subgraphs \\\hline
       $G_{sg} = \{G^j_{sg}\}^{\theta_2}_1$ & the sketch subgraph set \\\hline
 %      $d(u,v)$                   & the distance between vertices $u$ and $v$ \\\hline
%       $L(v)$                     & the label of vertex $v$ \\\hline
%       $L_{in}(v)$                & the vertex pair set with the end point $v$ w.r.t. $(u, d(u,v))$ \\\hline
%       $L_{out}(v)$               & the vertex pair set with the start point $v$ w.r.t. $(u, d(v,u))$ \\\hline
%       $(h, v)$                   & one hop; $h$ represent one path in graph $G$ and $v$ is the point of $h$ \\\hline 
       $\theta_3$                 & the number of generated sketch subgraphs in the SBG method \\\hline
    \end{tabular}}
  \end{center}
\end{table}

We define a directed evolving network as a sequence of graph snapshots $\mathcal{G} = \{G_{i}\}_0^{t-1}$, and $\{0,,1,..,t-1 \}$ is a set of time points. We assume that the network snapshots in $\mathcal{G}$ share the same vertex set. Let $G_i$ represent the network snapshot at timestamp $i\in [0,t-1]$, where each vertex $u$ in $V$ is a social user in $G_i$, each edge $e = (u, v)$ in $E_i$ represents a cyber link or a social relationship between users $u$ and $v$ in $G_i$. Similar to~\cite{DBLP:journals/tkde/JiaLDZGXZ21,das2019incremental}, we can create ``dummy" vertices at each time step $i$ to represent the case of vertices joining or leaving the network at time $i$ (\textit{e.g., $V = \cup_{i=1}^{t-1} V^i$ where $V^i$ is the set of vertices truly exist at $i$}). Besides, each edge $(u,v)\in E$ in $G$ is associated with a \textit{propagation probability} $p(u,v)\in [0,1]$.
% are the vertex set and edge set of $G_t$, respectively. %Besides, we set $nbr(u,G_t)$ as the set of vertices adjacent to vertex $u\in V$ in $G_t$, and the degree $d(u,G_t)$ represents the number of neighbors for $u$ in $G_t$, $i.e., |nbr(u,G_t)|$. 
Table~\ref{tab:symbol} summarizes the mathematical notations frequently used throughout this paper.

\subsection{Link Prediction}
Link prediction is an important network-related problem firstly proposed by
Liben-Nowell et al.~\cite{DBLP:conf/cikm/Liben-NowellK03}, which aims to infer the existence of new links or still unknown interactions between pairs of nodes based on their properties and the currently observed links. 

Given a directed evolving graph $\mathcal{G} = {G_i}^{t-1}_0$ with the time points set $\{0,1,..,t-1\}$, in this paper, we use the recent link prediction method~\cite{zhang2018link,zhang2021labeling}, named learning from \textbf{S}ubgraphs, \textbf{E}mbeddings, and \textbf{A}ttributes for \textbf{L}ink prediction (SEAL) method, to predict the graph structure of snapshot graph $G_t$ of $\mathcal{G}$ at the future time point $t$.
Specifically, SEAL is a graph neural network (GNN) based link prediction method that transforms the traditional link prediction problem into the subgraph classification problem. It first extracts the $h$-hop enclosing subgraph for each target link, and then applies a labeling trick, called Double Radius Node Labeling (DRNL), to add an integer label for each node relevant to the target link as its additional feature.  
Next, the above-labeled enclosing subgraphs are fed to GNN to classify the existence of links. Finally, it returns the predicted graph $G_t$ of evolving graph $\mathcal{G}$ at time point $t$. 

%Code for SEAL (learning from Subgraphs, Embeddings, and Attributes for Link prediction). SEAL is a novel framework for link prediction which systematically transforms link prediction to a subgraph classification problem. For each target link, SEAL extracts its h-hop enclosing subgraph A and builds its node information matrix X (containing structural node labels, latent embeddings, and explicit attributes of nodes). Then, SEAL feeds (A, X) into a graph neural network (GNN) to classify the link existence, so that it can learn from both graph structure features (from A) and latent/explicit features (from X) simultaneously for link prediction.

%Given link data for times range from $0$ to $t - 1$, Temporal link prediction (TLP) problem~\cite{articlehasan, pmlr-v2-sarkar07a} aims to predict the relationships at time $t$. The TLP problem is different from link prediction problem which ignores the temporal aspect and aims to predict missing connections in order to describe a more complete picture of the overall link structure in the data~\cite{clauset_hierarchical_2008}.
%Zhu et al.~\cite{7511675} proposed a scalable temporal latent space model to predict links over time based on a sequence of graph snapshots. This model was under the assumption that each user lies in an unobserved latent space and interactions are more likely to occur between similar users lies in an unobserved latent space and interactions are more likely to occur between similar users in the latent space.

\subsection{Influence Maximization (IM) Problem}
%Give a directed graph $G = (V, E)$ represent a social network, where $V$ and $E$ are the set of nodes and edges in $G$, respectively. Each edge $(u,v)\in E$ in $G$ is associated with a \textit{propagation probability} $p(u,v)\in [0,1]$.   

%is the set of nodes in $G$ (i.e., \textit{users}) and $E$ is the set of edges in $G$, and each edge $(u,v)\in E$ in $G$ is associated with a \textit{propagation probability} $p(u,v)\in [0,1]$. %One of the traditional probabilistic model -- IC model~\cite{IM2003}, has been introduced as follow to model the influence propagation process in $G$.  
To better understand the IM problem, we first introduce the influence diffusion evaluation of given users.

The independent cascade (\textit{IC}) model~\cite{IM2003} is the widely adopted stochastic model which is used for modeling the influence propagation in social networks. In the IC model, for each graph snapshot $G_i$, the \textit{propagation probability} $p(u,v)$ of an edge $(u,v)$ is used to measure the social impact from user $u$ to $v$. This probability is generally set as $p(u,v)=\frac{1}{d(v)}$, where $d(v)$ is the degree of $v$. Every user is either in an \emph{activated} state or \emph{inactive} state. $S_0$ be a set of initial \textit{activated} users, and generates the active set $S_t$ for all time step $t \geq 1$ according to the following randomized rule. At every time step $t \geq 1$, we first set $S_t$ to be $S_{t-1}$; Each user $u$ activated in time step $t$ has one chance to activate his or her neighbours $v$ with success probability $p(u,v)$. If successful, we then add $v$ into $S_t$ and change the status of $v$ to \textit{activated}. %Besides, if $v$ is failed to be activated by $u$ but $v$ has at least $k$ \textit{activated} neighbours in time step $t$, we also add $v$ into $S_t$ and change the status of $v$ to \textit{activated}. 
This process continues until no more possible user activation. Finally, $S_t$ is returned as the \textit{activated} user set of $S_0$.

%Initially, we \textit{activate} the seed users in a selected seed set $S$, while setting all other vertices as \textit{inactive}. In each timestamp $t>0$, if a vertex $u$ is \textit{activate} in $t$, then each of $u$'s neighbours $v$ has one changed to be activated by $u$ with $p(u,v)$ probability in timestamp $t+1$. Each vertex has one chance to activate its neighbors. This procedure terminates when no more users can be activated.

Let $I(S, G_i)$ be the number of vertices that are activated by $S$ in graph snapshot $G_i$ on the above influence propagation process under the IC model. The IM problem aims to find a size-$k$ seed set $S$ with the maximum expected spread $E(I(S,G_i))$. We define the IM problem as follows:

\begin{definition}[IM problem~\cite{IM2003}]
Given a directed graph snapshot $G_i = (V, E_i)$, an integer $k$, the IM problem aims to find an optimal seed set $S^*$ satisfying,
\begin{equation}
    S^* = \mathop{\arg\max}_{S\subseteq V, |S|=k} E(I(S,G_i))
\end{equation}
\end{definition}

Let $OPT$ be the maximum expected spread of any size-$k$ seed set, then we have $OPT = E(I(S^*,G_i))$. 

%\subsection{Sketch-based algorithms}

\subsection{Reverse Reachable Sketch}
The \textit{Reverse Influence Set} (RIS) \cite{RIS14} sampling technique is a \textit{Reverse Reachable Sketch-based} method to solve the \textit{IM} problem. By reversing the influence diffusion direction and conducting reverse \textit{Monte Carlo} sampling~\cite{kroese2014monte}, \textit{RIS} can significantly improve the theoretical run time bound. %For better understanding, we first introduce the concept of \textit{Reverse Reachable} (\textit{RR}) set. 

%***explain the definition clearly?
\begin{definition}[Reverse Reachable Set~\cite{RIS14}] \label{def:RIS}
Suppose a user $v$ is randomly selected from $V$. The reverse reachable (RR) set of $v$ is generated by first sampling a graph $g$ from $G_i$, and then taking the set of users that can reach to $v$ in $g$. %Then, a random RR set is an RR set for a node selected uniformly at random from $V$. 
\end{definition}

%The \textit{RIS} technique contains two phases: %\begin{enumerate}
%(1) Generate $\theta$ random RR sets from $G$; (2) Select a set $S$ of users to cover the maximum number of RR sets generated above by transforming it to the \textit{maximum coverage} problem~\cite{DBLP:journals/combinatorica/Wolsey82}.

By generating $\theta_1$ RR sets on random users, we can transform the IM problem to find the optimal seed set $S$, while $S$ can cover most RR sets. This is because if a user has a significant influence on other users, this user will have a higher probability of appearing in the RR sets. Besides, Tang et al.~\cite{DBLP:conf/sigmod/TangXS14} proved that when $\theta_1$ is sufficiently large, RIS returns near-optimal results with at least $1 - |V|^{-1}$ probability. Therefore, the process of using the RIS method to solve the IM query contains the following steps: 
\begin{itemize}
    \item [1] Generate $\theta_1$ random RR sets from $G_i$.
    \item [2] Find the optimal user set $S$ which can cover the maximum number of above generated RR sets.
    \item [3] Return the user set $S$ as the query result of IM query problem.
\end{itemize}

\begin{theorem}[Complexity of RIS~\cite{tang2014influence}] \label{theo:RIS}
If $\theta_1 \geq (8 + 2\varepsilon)\cdot |V| \cdot \dfrac{ln|V| + ln{|V|\choose k} + ln 2}{OPT\cdot e^2}$, RIS returns an $(1 - \frac{1}{e} - \varepsilon)$ approximate solution to the IM problem with at least $1 - |V|^{-1}$ probability. 
\end{theorem}

\subsection{Forward Influence Sketch} \label{subsection:FI-SKETCH} 
The Forward Influence Sketch (FI-SKETCH) method~\cite{10.1145/2661829.2662077,10.1145/2505515.2505541,ohsaka2014fast} 
%is to
constructs a sketch by extracting the subgraph induced by an instance of the influence process (\textit{e.g., the IC model}). Then, it can estimate the influence spread of a seed set $S$ using these subgraphs accurately with theoretical guarantee. The process of using the FI-SKETCH method to solve the IM query contains the following steps: 
\begin{itemize}
    \item [1] Generate $\theta_2$ sketch subgraph $G^j_{sg}$ by removing each edge $e=(u,v)$ from $G_i$ with probability $1 - P_{u,v}$.
    \item [2] Find the optimal user set $S$, while the average number of users reached by $S$ within $\theta_2$ constructed sketches graphs is maximum.
    \item [3] Return the user set $S$ as the query result of IM query problem.
\end{itemize}

\begin{theorem}[Complexity of FI-SKETCH~\cite{10.1145/2505515.2505541}] \label{theo:FI-Sketch}
If $\theta_2 \geq (8 + 2\varepsilon)\cdot |V| \cdot \dfrac{ln|V| + ln{|V|\choose k} + ln 2}{\varepsilon^2}$, FI-SKETCH returns an $(1 - \frac{1}{e} - \varepsilon)$ approximate solution to the IM problem with at least $1 - |V|^{-1}$ probability. 
\end{theorem}

%% file: Problem_definition.tex
%\section{Reconnecting Top-$l$ Links query}
\section{Problem Definition}
\label{sec:prob_def}
In this section, we formulate the \textit{Reconnecting Top-$l$ Relationships} (RT$l$R) query problem and analyze 
%the 
its complexity.
%of the RT$l$R problem as well.

%We are now ready to formally introduce the RT$l$L problem.

%\vspace{2mm}\noindent
%\textbf{Problem formulation: }
\begin{definition}[RT$l$R Problem] \label{def:prob}
Given a directed evolving graph $\mathcal{G} = \{G_i\}^{t-1}_0$, the parameter $l$, and a group of users $\mathcal{U}$, the problem of \textit{Reconnecting Top-$l$ Relationships} (RT$l$R) asks for finding an optimal edge set $S$ with size $l$ in predicted graph snapshot $G_t$ of $\mathcal{G}$ at time $t$, where the expected spread of $\mathcal{U}$ will be maximized while reconnecting edges of $S_e$ in $G_t$ (\textit{e.g., $ \widehat{G_t} = G_t \oplus S_e$}). Formally, 

\begin{equation}
    \widehat{S_e} = \mathop{\arg\max}_{S_e\subseteq \mathcal{G}\setminus G_t} E(I(\mathcal{U},\widehat{G_t}) \\ \\
\end{equation}
\end{definition}

\iffalse
Let $OPT^*$ represent the maximum expected additional spread of given users group $\mathcal{U}$ by reconnecting any size-$l$ seed edges set $S_e$ in $G_t$, satisfying with $OPT^* = E(I(\mathcal{U}, G_t\oplus \widehat{S_e}) - I(\mathcal{U},G_t))$. 
\fi

In the following, we conduct a theoretical analysis on the hardness of the RT$l$R problem.

\begin{theorem}[Complexity] \label{theo:problem_complexity} \label{theo:NPhard}
%Given an evolving general graph $\mathcal{G} = \{G_i\}^{t}_1$, t
The RT$l$R problem is NP-hard. 
\end{theorem}

\begin{proof}
%We prove the problem of RT$l$L is NP-hard, by reducing RT$l$L problem to the maximum coverage problem~\cite{karp1972reducibility}. ... Since the \textit{Maximum Coverage} problem is NP-hard, and so is the RT$l$L problem.
We prove the hardness of RT$l$R problem by a reduction from the decision version of the maximum coverage (MC) problem~\cite{karp1972reducibility}. Given an integer $l$ and several sets where the sets may have some elements in common, the maximum coverage problem aims to select at most $l$ of these sets to cover the maximum number of elements. Furthermore, we need to discuss the existence of a solution that the MC problem is reducible to the RT$l$L problem in polynomial time.

Given a directed evolving graph $\mathcal{G}$, a group of users $\mathcal{U}$, and the predicted snapshot graph $G_t$ from $\mathcal{G}$, we reduce the MC problem to RT$l$L with the following process: (1) For a given group $\mathcal{U}$, we compute the influence users set of $\mathcal{U}$ as $I(\mathcal{U},G_t)$; (2) $\forall e\in \mathcal{G} \setminus G_t$, we create a set $S_e$ with the elements 
%that 
collected from the influenced users $I(\mathcal{U}, \widehat{G_t}) - I(\mathcal{U},G_t)$ while $\widehat{G_t} = G_t\oplus e$; (3) We set the reconnecting edges of RT$l$L as $l$, which is the same as the input of $MC$.
The above reduction can be done in polynomial time. Since the \textit{Maximum Coverage} problem is NP-hard, so is the RT$l$L problem.
\end{proof}

\begin{theorem}[Influence Spread] \label{lemma:greedy}
The influence spread function $I(.)$ under the RT$l$R problem is \textit{monotone} and \textit{submodular}. 
\end{theorem}

\begin{proof}
Given a snapshot graph $G_t$, and a group $\mathcal{U}\in V(G_t)$, $I(\mathcal{U}, G_t)$ represents the influenced user set of $\mathcal{U}$. For two edge sets $S_e \subseteq T_e$, we have $I(\mathcal{U}, G_t \oplus S_e) \leq I(\mathcal{U}, G_t \oplus T_e)$. Then, we have verified that $I(.)$ is \textit{monotone}. Besides, for a new reconnecting edge $e$, the marginal contribution when added to set $S_e$ and $T_e$ respectively satisfies $I(\mathcal{U},G_t \oplus (S_e\cup e)) - I(\mathcal{U}, G_t \oplus S_e) \geq I(\mathcal{U},G_t \oplus (T_e\cup e)) - I(\mathcal{U}, G_t \oplus T_e)$. Therefore, we have proved that $I(.)$ is \textit{submodular}. Thus, we can conclude that the influence spread function $I(.)$ of RT$l$L problem is \textit{monotone} and \textit{submodular}.   
\end{proof}

%% file: Basic_Method.tex
\section{Sketch based Greedy Algorithm} \label{sec:basic}
To answer the RT$l$R query problem, we first predict the graph structure of the given evolving graph $\mathcal{G}$ 
%in time point 
at 
$t$ by using the link prediction method~\cite{zhang2021labeling}. According to Theorem~\ref{lemma:greedy}, the influence spread function of 
%the 
RT$l$R 
%problem 
is \textit{submodularity} and \textit{monotonicity}. Therefore, one possible solution of the RT$l$L problem is to use the greedy approach to iteratively find out the most influential edge $e$, in which reconnecting $e$ in predicted snapshot graph $G_t$ will maximize the influence spread of given users group $\mathcal{U}$ in $\widehat{G_t}$ (\textit{e.g., $\widehat{G_t} = G_t \oplus e$}). 
So far, the remaining challenge of RT$l$R query is to evaluate the effect of a reconnected edge $e$ on the influence spread of $\mathcal{U}$ in $G_t$. 

\subsection{Existing IM Approaches Analysis}\label{subsec:IManlyasis}
As mentioned in~\cite{IM2003}, we can estimate the influence spread of given users by using the \textit{Monte Carlo} simulation. Specifically, given users group $\mathcal{U}$, we simulate the randomized diffusion process with $\mathcal{U}$ in $G_t$ for $\mathcal{R}$ times. Each time we count the number of active users after the diffusion ends, and then we take the average of these counts over the $\mathcal{R}$ times as the estimated number of influenced users of $\mathcal{U}$. However, the \textit{Monte Carlo} simulation method is much time-consuming and cannot be used in the large graph. Later on, Borgs et al.~\cite{RIS14} proposed a \textit{Reverse Reachable Sketch-based} method to the \textit{IM} problem, named \textit{Reverse Influence Set} (RIS) sampling, and the extended versions of the RIS method~\cite{tang2014influence,nguyen2016stop,nguyen2017importance} were widely used to answer the IM problem as the state-of-the-art IM query methods.  
The \textit{Reverse Influence Set} (RIS) sampling technique is a \textit{Reverse Reachable Sketch-based} method to the IM problem. By reversing the influence diffusion direction and conducting reverse \textit{Monte Carlo} sampling, \textit{RIS} can significantly improve the theoretical run time bound of the IM problem. 

Unfortunately, the RIS sampling method is not suitable for answering our RT$l$R query. That is because the RIS sampling is designed to find the Top-$k$ most influential users in a graph, but our RT$l$R query focuses on reconnecting several optimal edges to enhance a given user group's influence spread. In particular, the RIS sampling method transforms the IM problem to find the optimal seed set $S$ by generating $\theta_1$ RR sets, while $S$ can cover most RR sets. The RR sets only contain the user's information while discarding the graph sketch (\textit{e.g., the edge's information}). Therefore, if we use the RIS sampling to answer the RT$l$R query, we have to recompute the RR sets for each edge insertion during the RT$l$R query process, which is time-consuming and unrealistic in large graphs.

\subsection{FI-Sketch based Greedy Algorithm} \label{subsec:sbg}
Facing the challenges mentioned above, we propose a sketch-based greedy (SBG) method to answer the RT$l$R query. Precisely, we first set $\theta_3$ as a sufficient number of generated sketch subgraphs in our SBG method to theoretically ensure the quality of the returned results for the RT$l$R query (\textit{i.e., the details of how $\theta_3$ should be set will further discuss in Section~\ref{subsec:theta_3}}). Then, we use the FI-SKETCH to evaluate the effect of a new adding edge $e$ on the influence spread of a given users group $\mathcal{U}$ based on the $\theta_3$ generated sketch subgraphs. Compared with the RIS approach, the graph structure information was contained in the generated $\theta_3$ sketch subgraphs during the process of the FI-SKETCH approach (refer to Section~\ref{subsection:FI-SKETCH}), so that we do not need to recompute the sketches while the edges update.% inserting.
\setlength{\textfloatsep}{0.2cm}
\begin{algorithm}[t!] %[hbt!] 
%	\caption{\textit{The Greedy Algorithm for RT$l$L Query}}
	\caption{\textbf{RT$l$R:} SBG} % $(\mathcal{G}, l, \mathcal{U})$}}
	\label{alg:greedy_basic}
	\KwIn{$\mathcal{G}=\{G_i\}^{t-1}_0:$ an evolving graph, $l$: the number of selected edges, and $\mathcal{U}$: a group of users}
	\KwOut{ $\widehat{S_e}:$ the optimal reconnecting edge set} 
	\BlankLine
	Predict the snapshot graph $G_t$ from $\mathcal{G}$~\cite{zhang2021labeling}; \\  \label{alg1:1}
	Generate $\theta_3$ sketch subgraph $G_{sg} = \{G^j_{sg}\}_1^{\theta_3}$; \\
	Initialize $\widehat{S_e}\leftarrow \emptyset$,
	Candidate edges set $CE\in \{\mathcal{G}\setminus G_t\}$; \\ \label{alg1:2}
	\For {$i = 1$ to $l$}{  \label{alg1:greedy1}
	    $\widehat{e} \leftarrow \mathop{\arg\max}_{e\in CE}$ FI-SKETCH$(\mathcal{U},e)$; \\
        $\widehat{S_e}\leftarrow \widehat{S_e}\cup \widehat{e}$; \\ \label{alg1:greedy2} 
    }
\Return{$\widehat{S_e}$}	 \label{alg1:result}

\SetKwFunction{FMain}{FI-SKETCH} 
    \SetKwProg{Fn}{Function}{:}{}
    \Fn{\FMain{$\mathcal{U},e$}}{  \label{alg1:FI1}
    $count \leftarrow 0$; \\ % $G'\leftarrow \widehat{G} \oplus e$; \\
    \For{$j = 1$ to $\theta_3$}{
       $\widehat{G^j_{sg}} \leftarrow G^j_{sg} \oplus \{\widehat{S_e} \cup e\}$; \\
       $n_a\leftarrow$ the number of vertexes reached by $\mathcal{U}$ in $\widehat{G^j_{sg}}$; \\ % after the diffusion ends; \\
       $count \leftarrow count + n_a$; \\
    }
   \Return{$\dfrac{count}{\theta_3}$}\label{alg1:FI2}
}
\textbf{End Function}
\end{algorithm}
\setlength{\floatsep}{0.2cm}

The details of the SBG method are described in Algorithm~\ref{alg:greedy_basic}. In the pre-computing phase (Lines~\ref{alg1:1}-\ref{alg1:2}), we predict the snapshot graph $G_t$ using the link prediction method~\cite{zhang2021labeling}, and then generate $\theta_3$ random sketch graphs by removing each edge $e = (u,v)$ from $G_t$ with probability $1 - P_{u,v}$. Besides, based on Definition~\ref{def:prob}, we initialize $CE\in \{\mathcal{G} \setminus G_t\}$ as the candidate edges set of the RT$l$R query. %followed the RT$l$R problem definition that the selected reconnecting edges disappear in $G_t$ while have existed in $\mathcal{G}$. 
In the main body of SBG (Lines~\ref{alg1:greedy1}-\ref{alg1:greedy2}), we use the greedy method to iteratively find the $l$ number of optimal reconnecting edges. Specifically, in each iterative, we call the \textit{FI-SKETCH Function} to find out the optimal edge $\widehat{e}$ from the candidate edge set $CE$ and add $\widehat{e}$ into set $\widehat{S_e}$ , while reconnecting the selected edge can maximize the influence diffusion of given users group $\mathcal{U}$. Meanwhile, given an edge $e$, the \textit{FI-SKETCH Function} returns back the influenced users evaluation results by using the \textit{Forward Influence Sketch} method mentioned in Section~\ref{subsection:FI-SKETCH} (Lines~\ref{alg1:FI1}-\ref{alg1:FI2}).     
Finally, we return edges set $\widehat{S_e}$ as the result of RT$l$R query (Line~\ref{alg1:result}).

\vspace{1mm}
\noindent
\textbf{Complexity.} 
%From the above process of Algorithm~\ref{alg:greedy_basic}, we can see that 
The time complexity of calling the FI-SKETCH function for each candidate edges is $\mathcal{O}(\theta_3 \cdot |E_t|)$, while the space complexity is $\mathcal{O}(\theta_3 \cdot (|V| + |E_t|))$. Hence, the time complexity and space complexity of SBG algorithm are $\mathcal{O}(l \cdot |CE| \cdot \theta_3 \cdot |E_t|$) and $\mathcal{O}(\theta_3 \cdot (|V| + |E_t|))$, respectively. 

\subsection{Theoretical Analysis of SBG} \label{subsec:theta_3}
In this part, we will establish our theoretical claims for SBG. Specifically, we analyze how $\theta_3$ should be set to ensure our SBG method returns near-optimal results to RT$l$R query with high probability.
%Let the total number of additional reached users in each sketch subgraph $G_{sg}^j\in G_{sg}$ as $F(\mathcal{U},S_e)$, which is reached by reconnecting the edges in $S_e$ of graph $G_t$. 
Our analysis highly relies on the \textit{Chernoff bounds}~\cite{Hoeffdinginequality}. 

%Let $F(\mathcal{U},S)$ denote the average number of additional reached users in each graph sketch $G_{sg}^j$, which is reached by reconnecting the edges in $S$ to graph $G_t$.
\begin{lemma} \label{lem:Hoeff} %~\cite{Hoeffdinginequality}
Let $X_1$,...,$X_r$ be $r$ number of independent random variables in $[0,1]$ and $X = $ with a mean $\mu$. For any $\sigma >0$, we have
\begin{equation}
\begin{split}
    & Pr[X - r\mu \geq \sigma\cdot r\mu] \leq exp(-\frac{\sigma^2}{2 + \sigma}r\mu), \\
    & Pr[X - r\mu \leq -\sigma\cdot r\mu] \leq exp(-\frac{\sigma^2}{2}r\mu).
\end{split}
\end{equation}
\end{lemma}

%Let the total number of additional reached users in each sketch subgraph $G_{sg}^j\in G_{sg}$ as $F(\mathcal{U},S_e)$, which is reached by reconnecting the edges in $S_e$ of graph $G_t$. 
Let $\mathcal{U}$ be a group of users, $S_e$ be the selected reconnecting edges, $\mathcal{R}_2$ be the number of generated sketch subgraphs in the SBG algorithm (\textit{Algorithm~\ref{alg:greedy_basic}}), and $F_R(\mathcal{U},S_e)$ be the total number of additional reached users by $\mathcal{U}$ in each sketch subgraph after reconnecting edges in $S_e$. From~\cite{ohsaka2014fast}, the expected value of $\frac{F_R(\mathcal{U}, S_e)}{\mathcal{R}_2}$ equals the expected influence diffusion enhance by reconnecting edges of $S_e$ in $G_t$.  
Then, we have the following lemma.

\begin{lemma} \label{lem:expected}
$E[\frac{F_R(\mathcal{U},S_e)}{\mathcal{R}_2}] = E[I(\mathcal{U}, G_t\oplus S_e) - I(\mathcal{U}, G_t)]$
\end{lemma}

\begin{proof}
%Let $\mathcal{R}$ be the number of sketch subgraphs generated in SBG algorithm, $\mathcal{U}$ be the given user's group, and $S_e$ be the reconnecting edges set in $G_t$. 
Each sketch subgraph in the SBG algorithm is generated by removing each edge $e$ with $1 - p(e)$ probability. From~\cite{ohsaka2014fast}, we can observe that the expected value of the average number of reached users to $\mathcal{U}$ in all sketch subgraphs is equal to the expected spread of $\mathcal{U}$ in $G_t$.   
From the above relation of equality, we can easily deduce that $E[\frac{F_R(\mathcal{U},S_e)}{\mathcal{R}_2}] = E[I(\mathcal{U}, G_t\oplus S_e) - I(\mathcal{U}, G_t)]$.
\end{proof}

\begin{theorem}[Approximate ratio] \label{theo: FIproof}
By generating $\theta_3$ sketch subgraphs with $\theta_3 \geq (8 + 2\varepsilon)\cdot |V| \cdot \dfrac{ln|V| + ln{|V|\choose l} + ln 2}{\varepsilon^2}$, we have $|\frac{F(\mathcal{U}, S_e)}{\theta_3} - (E[I(\mathcal{U}, G_t\oplus S_e) - I(\mathcal{U}, G_t)])| <\frac{\varepsilon}{2}$ holds with probability $1- |V|^{-l}$ simultaneously for all selected edges set $S$ (i.e., $|S| = l$). 
\end{theorem}

\begin{proof}
We can prove 
%the 
Theorem~\ref{theo: FIproof} by tweaking the proof in Theorem~\ref{theo:RIS} of~\cite{tang2014influence}. Let $\rho$ be the probability of $\mathcal{U}$ can activate a fixed user $v$ after reconnecting edges in $S_e$ in $G_t$. Based on Lemma~\ref{lem:expected}, 
\begin{equation}
    \rho = E[\frac{F_R(\mathcal{U},S_e)}{\mathcal{R}_2}] / |V| = (E[I(\mathcal{U}, G_t\oplus S_e) - I(\mathcal{U}, G_t)]) / |V| 
\end{equation}
Then, we have
\begin{equation} \label{eq:prop}
    \begin{split}
        Pr & [|\frac{F_R(\mathcal{U},S_e)}{\theta_3} - (E[I(\mathcal{U}, G_t\oplus S_e) - I(\mathcal{U}, G_t)]) |] \geq \frac{\varepsilon}{2} \\
        & = Pr[|\frac{F_R(\mathcal{U},S_e)}{|V|} - \rho \theta_3|] \geq \frac{\varepsilon \theta_3}{2|V|}
    \end{split}
\end{equation}
Let $\sigma = \frac{\varepsilon}{2|V|\rho}$. Based on Lemma~\ref{lem:Hoeff}, we have
\begin{equation}
    \begin{split}
        Equation~(\ref{eq:prop})  < & 2 \cdot exp(-\frac{\sigma^2}{2 + \sigma}\cdot \rho\cdot \theta_3) \\
                                  & = 2 \cdot exp(\frac{\varepsilon^2}{8|V|^2\rho + 2|V|\varepsilon}\cdot \theta_3) \\
                                  & \leq 2 \cdot exp(-\frac{\varepsilon^2}{8|V| + 2\varepsilon|V|}\cdot \theta_3) \\
                                  & \leq \frac{1}{|V|^l}.
    \end{split}
\end{equation}
Thus, Theorem~\ref{theo: FIproof} is proved.
\end{proof}

\begin{theorem}[Complexity of SBG] \label{theo:sbgproof}
With a probability of $1 - |V|^{-l}$, the SBG method for solving the RT$l$R query problem requires $\theta_3 \geq (8 + 2\varepsilon)\cdot |V| \cdot \dfrac{ln|V| + ln{|V|\choose k} + ln 2}{\varepsilon^2}$ number of sampling sketch subgraphs so that an $(1 - \frac{1}{e} -\varepsilon)$ approximation ration is achieved. 
\end{theorem}

\begin{proof}
The proof of Theorem~\ref{theo:sbgproof} is summarized as following three steps. 
Firstly, based on the property in Theorem~\ref{theo: FIproof}, if the number of generated sampling sketch subgraphs $\theta_3 \geq (8 + 2\varepsilon)\cdot |V| \cdot \dfrac{ln|V| + ln{|V|\choose k} + ln 2}{\varepsilon^2}$, then we have $|\frac{F(\mathcal{U},S_e)}{\theta_3} - (E[I(\mathcal{U}, G_t\oplus S_e) - I(\mathcal{U}, G_t)])| < \frac{\varepsilon}{2}$ holds with probability $1 - |V|^{-l}$. Secondly, the SBG method we proposed in this paper to solve the RT$l$R problem by utilizing the greedy algorithm of \textit{maximum coverage} problem~\cite{karp1972reducibility}, which produces a $(1 - \frac{1}{e})$ approximation solution (\textit{mentioned in Theorem~\ref{theo:NPhard}}). 
Finally, by combining the above two approximation ration $\frac{\varepsilon}{2}$ and $(1 - \frac{1}{e})$, we can conclude the final approximation ration of our SBG method for solving RT$l$R query problem is $(1 - \frac{1}{e} - \varepsilon)$ with at least $1 - |V|^{-l}$ probability.
\end{proof}

\subsection{Reducing \# Candidate Edges} \label{subsec:optim1} 
Since the SBG algorithm's time complexity is cost-prohibitive, which would hardly be used for dealing with the sizeable evolving graph. %In this Section, we accelerate the SBG algorithm from two aspects: (i) Reducing the candidate edges (\textit{w.r.t decrease $|CE|$}); 
In this subsection, we present our optimization method by pruning the unnecessary potential edges in candidate edge set $CE$. The core idea behind this optimization strategy is to eliminate the edges in $CE$ which will not have any benefit to expend the influence spread of given users group $\mathcal{U}$ while reconnecting it.

We use the symbol $u\leftsquigarrow \mathcal{U}$ to denote that $u$ can be reached by $\mathcal{U}$. %in $G_t$. 
In order to reduce the size of $CE$, we present the below theorem to identify the quality reconnecting edge candidates (denote as $\widehat{CE}$) from $CE$.  

\begin{theorem}[Reachability] \label{theo_edges}
Given a directed snapshot graph $G_t$ and a users group $\mathcal{U}$, if an edge $e = (u,v)$ is selected to reconnect, one of its related users (i.e., $u$ or $v$) requires to be reached by $\mathcal{U}$ in $G_t$; that is $e\in\widehat{CE}$ implies $u\leftsquigarrow \mathcal{U}$ or $v \leftsquigarrow \mathcal{U}$ in $G_t$.
%$u \stackrqarrow{G_t}\mathcal{U}$ or $v\stackrqarrow{G_t} \mathcal{U}$.
\end{theorem}

\begin{proof}
We prove the correctness of this theorem by contradiction. The intuition is that at least one pathway exists from a user to all of its influenced users in social networks. For the selected edge $e = (u,v)$, if both 
%of 
the user $u$ and $v$ are not reached by the users group $\mathcal{U}$, then 
%no exists 
the pathway between $\mathcal{U}$ and $e$
does not exist.
%If edge $e = (u, v)$ is selected under both of its related users $u$ and $v$ are not reached by $\mathcal{U}$, then no pathway from group $\mathcal{U}$ to $e$ exists.
Therefore, reconnecting the edge $e$ does not bring any benefits to the expansion of influence spread starting from $\mathcal{U}$, which contradicts with Definition~\ref{def:prob}. Thus, the theorem is proved.   
\end{proof}

\begin{figure}[!h]
    \centering
    \includegraphics[scale=0.55]{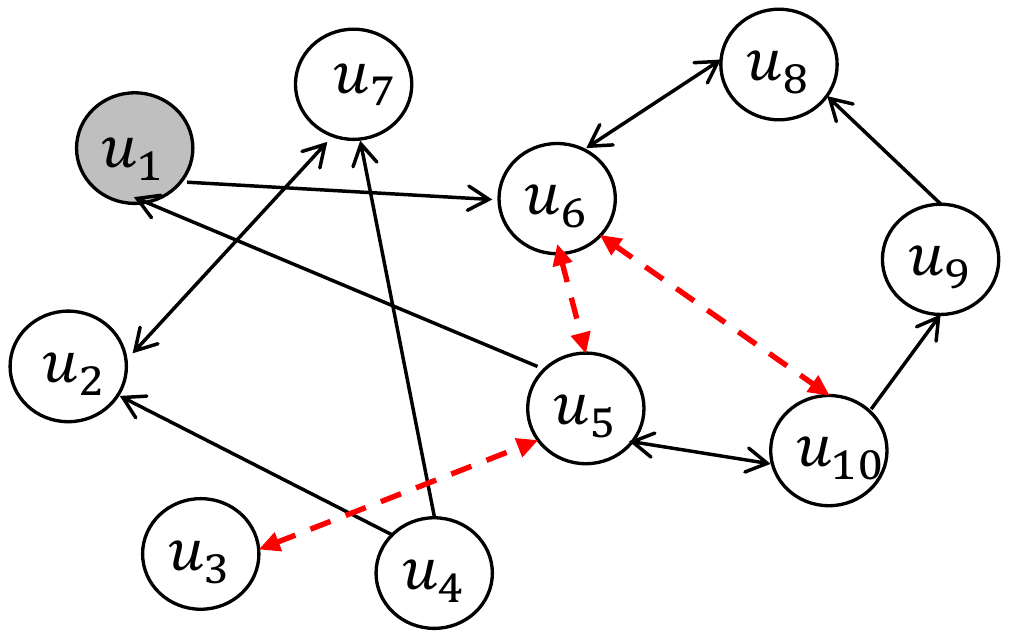}
    \caption{Running Example}
    \label{fig:CE-SBG2}
\end{figure}

\begin{example}
Figure~\ref{fig:CE-SBG2} shows a snapshot graph $G_t$ with $10$ nodes and $9$ edges. The candidate edges set of RT$l$R is $CE = \{(u_3,u_4), (u_5,u_6),\\ (u_6, u_{10})\}$. For a given user group $\mathcal{U} = \{u_1\}$, the pruned candidate edge set would be $\widehat{CE}=\{(u_5,u_6), (u_6, u_{10})\}$ due to $u_6\leftsquigarrow \mathcal{U}$.
\end{example}

Based on Theorem~\ref{theo_edges}, we present a BFS-based method for pruning the candidate edge set $CE$ in graph $G_t$ with a given users group $\mathcal{U}$. The core idea of the BFS-based algorithm is to traverse the graph $G_t$ starting from the nodes in $\mathcal{U}$ by performing breadth-first search (BFS). For edges in $CE$, if both of its related nodes are not visited in the above BFS process, then we directly prune it.

In Algorithm~\ref{alg:BFS_prunning_edges}, we outline the major steps of the BFS-based method for processing the $CE$ pruning. Initially, each user $u$ in graph $G_t$ are marked a visiting status as FALSE (Line~\ref{label:alg2_1}). Then, for the users in a given group $\mathcal{U}$, we update its visiting status as TRUE (Lines~\ref{label:alg2_2}-\ref{label:alg2_3}). Further, we process a BFS search starting from root user $v\in \mathcal{U}$, and update the status of each visited users as TRUE (Lines~\ref{label:alg2_4}-\ref{label:alg2_5}). Next, based on Theorem~\ref{theo_edges}, we reduce all candidate edges $e=(u,v)$ from $CE$ while both $u$ and $v$ have the FALSE visited status (Lines~\ref{label:alg2_6}-\ref{label:alg2_7}), and finally, we return the pruned candidate edges set $\widehat{CE}$ (Line~\ref{label:alg2_8}). % in which can be reached by $\mathcal{U}$.

%To identify the quality reconnecting candidate edge set $\widehat{CE}$, we first propose a BFS-based method. The main idea of our method is to 

\begin{algorithm}[t] 
	\caption{Reducing $CE$ $\#$ \textit{BFS} $(CE, \mathcal{U})$} 
	\label{alg:BFS_prunning_edges}
	\BlankLine
%	Candidate edges set $CE\leftarrow \{\mathcal{G}\setminus G_t\}$; \\
    Initialize set $\widehat{CE}\leftarrow \emptyset$, an empty Queue $Q$; \\
    Initialize visited array $A$ with size $|V|$ as FALSE; \label{label:alg2_1} \\ 

%	Initialize array $A$ with size $n$; \\ %$\widehat{S}\leftarrow \emptyset$, $\widehat{G}\leftarrow G_t$; \\ % $t\leftarrow 1$; \\ 
   \For{each $u\in U$}{ \label{label:alg2_2}
      $A[u] = TRUE$; \\
      Enqueue $u$ into $Q$; \\ \label{label:alg2_3}
    }
    \While {$Q$ is not empty}{ \label{label:alg2_4}
         Dequeue $v$ from $Q$; \\
         \For{each neighbor $v'\in nbr(v,G_t)$ in $G_t$}{
            \If{$A[v'] = FALSE$}{
                $A[v'] = TRUE$, 
                Enqueue $v'$ into $Q$; \\
            }
            \Else{
                Continue; \\ \label{label:alg2_5}
            }
         }
    }
   \For{$e = (u,v)\in CE$}{\label{label:alg2_6}
       \If{$A[u] = TRUE$ or $A[v] = TRUE$}{
        add $e$ into $\widehat{CE}$
       }
       \Else{
          continue; \label{label:alg2_7} 
       }
   }
\Return {$\widehat{CE}$} \label{label:alg2_8}
\end{algorithm}

\vspace{1mm}
\noindent
\textbf{Complexity. } Obviously, for a given group $\mathcal{U}$, the time complexity of Algorithm~\ref{alg:BFS_prunning_edges} is 
$\mathcal{O}(|V| + |E_t| + |CE|)$, and the space complexity is $\mathcal{O}(|V|)$. Furthermore, the occupied space by Algorithm~\ref{alg:BFS_prunning_edges} will be released after the pruned candidate edges $\widehat{CE}$ is returned. For each RT$l$R query with a new given users group $\mathcal{U}$ as input, we need to recall the BFS-based pruning method %(\textit{w.r.t Algorithm~\ref{alg:improve_prunning_edges}}) 
to reduce the size of candidate set $CE$ with time cost $\mathcal{O}(|V| + |E_t| + |CE|)$, which is the main drawback of the BFS-based pruning method.

%% file: Reducing_candidate.tex
\section{The Improvement Algorithm} \label{sec:improved}
%and (ii) Order-based early termination, which focus on early terminating the effect evaluation of potential edges to given users group's influence spread in each iteration of SBG algorithm. The details of the above two aspects will be discussed in Section~\ref{subsec:optim1} and Section~\ref{subsec:order} respectively. 

%\subsection{Order-based SBG} \label{subsec:order}
%This section proposes an efficient algorithm to answer the RT$l$L query, named Order-based %SBG, which is extended from our SBG algorithm.
Although the SBG algorithm and its optimization method can successfully answer the RT$l$R query problem, it is still time-consuming to handle the sizeable social networks. To address this limitation, in this section, we propose an ordered sketch-based greedy algorithm, which can significantly reduce the number of edges influence probing at each iterative of RT$l$R query process, so as to answer the RT$l$R query more efficiently.  

%to solve the RT$l$L problem

\subsection{Algorithm Overview}
Let $\mathcal{G} = \{G_0,G_1,...,G_{t-1}\}$ be an evolving graph. We first use the temporal link prediction method~\cite{7511675} to predict the future snapshot of graph $G_t$, and the potential reconnecting edges will be selected from candidate edges set $CE = \{\mathcal{G} \setminus G_t \}$. 
Before introducing the core idea of our Order-based SBG algorithm, we first briefly review using the SBG algorithm to answer the RT$l$R query and analyze the bottleneck of the SBG algorithm.

For each given users group $\mathcal{U}$, the SBG algorithm aims to find $l$ reconnecting edges by iteratively probing each edge in $CE$ to find out the edge $\widehat{e}$ in which reconnecting $\widehat{e}$ will bring the maximum benefits to the influence spread of $\mathcal{U}$. The time complexity of influence spread by reconnection of an edge is $\mathcal{O}(\theta_3\cdot |E_t|)$, which is the bottleneck of the SBG algorithm. 

To deal with the above limitation of the SBG method, we propose an Order-based SBG algorithm, which focuses on reducing the number of edges probing in each iteration by using our elaboratively designed two-step bounds approach together with the order-based probing strategy. Specifically, we first generate a label index (UBL) to store the first step upper bound of influence spread expansion for each candidate edge $e\in CE$ \textit{w.r.t $UB_1(e)$} (in Section~\ref{subsec:label}). Then, we generate the initial second-step upper bound ($UB_2$) for $e$ (\textit{i.e., $UB_2.e$}) from $UB_1(e)$ of the UBL index. Next, in the influence spread expansion estimation query processing of each given users group $\mathcal{U}$ and probing edge $e$, we narrow the second-step upper bound of $e$ and update the $UB_1(e)$ value of UBL index, while the narrowed second-step upper bound will be served the optimal edge finding in the following iterations (in Section~\ref{subsec:estimation}).   
Finally, we order the candidate edges by their $UB_2$ values. The edge probing at the current iteration will be early terminated while the second upper bound of probing edge $e$ is less than the present influence spread expansion estimation value (in Section~\ref{subsec:order}).

\subsection{Upper Bound Label  (UBL) Construction} \label{subsec:label}
%In this section, we first introduce how to build the upper bound label for each edges in candidate edges set $CE$, and then present our query processing algorithm based on the constructed label.
This section introduces how to build the label index (UBL) for each candidate edge. The UBL index contains two parts, including (1) the $\theta_3$ sketch subgraphs $G_{sg}$; (2) the first-step bound $UB_1(e)$ of each candidate edge $e$ and its updating status. The details of UBL construction procedure is shown in Algorithm~\ref{alg:label}.

\begin{algorithm}[t!] %[hbt!] 
	\caption{Build UBL($\mathcal{L}$, $G_{sg}$)}
	\label{alg:label}
	\BlankLine
	($\mathcal{L}$, $G_{sg}$) $\leftarrow (\emptyset, \emptyset)$; \\ 
	Generate $\theta_3$ sketch subgraph $G_{sg} = \{G^j_{sg}\}_1^{\theta_3}$; \label{index:1} \\ 
	\For {each edge $e=(u,v)\in CE$}{ \label{index:2}
        $UB_1(e)\leftarrow$ the number of vertices that can be reached from $v$ in $G_t$; \\
        $flag(e) \leftarrow 0$; \\
        add $(UB_1(e),flag(e))$ into $\mathcal{L}$; \label{index:3}\\
    }
    Store ($\mathcal{L}$, $G_{sg}$)   \label{index:4}
\end{algorithm}

From Section~\ref{subsection:FI-SKETCH}, we first generate $\theta_3$ sketch subgraphs from the predicted snapshot graph $G_t$ that will be used for the future influence spread estimation (Line~\ref{index:1}). Then, for each candidate edge $e$ in $CE$, we initialize its updating mark (\textit{i.e.,} $flag(e)$) as $0$. Meanwhile, we compute the number of vertices in $G_t$ that can be reached from $e$ as the first step upper bound of $e$, denoted as $UB_1(e)$ (Lines~\ref{index:2} - \ref{index:3}).  
Finally, we store the Labeling Scheme ($\mathcal{L}$, $G_{sg}$) for RT$l$R query processing (Line~\ref{index:4}). 

\vspace{1mm}
\noindent
\textbf{Complexity.} The time complexity of sketch subgraphs generation is $\mathcal{O}(\theta_3\cdot|E_t|)$, and the $UB_1$ labeling construction of all candidate edges in $CE$ is $O(|CE|\cdot |E_t|)$. Therefore, the time complexity of UBL construction is $\mathcal{O}(\theta_3\cdot|E_t| + |CE|\cdot |E_t|)$. 
Besides, the space complexity of UBL index construction is $\mathcal{O}(\theta_3\cdot (|V|+|E_t|) + |CE|)$, while storage sketch subgraphs $G_{sg}$ has space complexity $O(\theta_3\cdot (|V|+|E_t|))$ and generating $UB_1$ labeling of edges in $CE$ has space complexity of $\mathcal{O}(|CE|)$.

\subsection{Influence Spread Expanding Estimation}  \label{subsec:estimation}
Here, we present the influence spread expansion estimation of given users group $\mathcal{U}$ and edge $e$. Further, we also introduce the strategies of narrowing the two-step upper bounds of $e$ (i.e., $UB_1(e)$ and $UB_2(e)$) during the above estimation process. 

%In this subsection, we will introduce the computing process of benefiting estimation for an edge to the influence spread of given group $\mathcal{U}$. Besides, we also present the update strategies of the two-step bounds during the above computing process. 

\begin{algorithm}
	\BlankLine  
	\caption{Sketch-Estimate Function} \label{alg:estimation}
\SetKwFunction{FMain}{Sketch-Estimate} 
    \SetKwProg{Fn}{Function}{:}{}
    \Fn{\FMain{$\mathcal{U},e$}}{
    $count \leftarrow 0$, $count_R \leftarrow 0$, $e = (u,v)$;  \\ \label{alg:fun1}
    \For{$k = 1$ to $\theta_3$}{  \label{alg:fun2}
         \While{$SG[k][u]==1$ $\&\&$ $SG[k][v]==0$}{
%       $n_a\leftarrow$ the number of reached nodes $u'$ from $e$ with $SG[k][u']==0$; \\ 
       $n_a\leftarrow |\{u'\in V|u'\leftsquigarrow e$ in $G^k_{sg}  \ \wedge$  \ $SG[k][u']==0 \}|$; \\
       $count \leftarrow count + n_a$; \\ \label{alg:fun+}
       }
       \If{$\mathcal{L}.flag(e) == 0$}{ \label{alg:fun++}
%            $n_R\leftarrow$ the number of reached nodes from $e$ in $G^k_{sg}$; \\
            $n_R\leftarrow |\{u'\in V|u'\leftsquigarrow e$ in $G^k_{sg}\}|$ ; \\
            $count_R \leftarrow count_R + n_R$; \\
       }
       \Else{
            continue; \\     \label{alg:fun3}
       }
    }
    update $(e, UB_2.e) \leftarrow (e, count / \theta_3)$ of $Q$; \\ \label{alg:fun4}
    \If{$\mathcal{L}.flag(e) == 0$}{ \label{alg:fun5}
       update $(UB_1(e),flag(e))\leftarrow  (count_R / R, 1)$ of $\mathcal{L}$; \\ 
       $\mathcal{L}.flag(e)\leftarrow 1$; \\   \label{alg:fun6} 
    }
   \Return{$count / \theta_3$}     \label{alg:fun7}
}
\textbf{End Function}
\end{algorithm}

The details of the influence spread expansion estimation are described in Algorithm~\ref{alg:estimation}. For a given users group $\mathcal{U}$ and edge $e = (u,v)$, the Sketch-Estimate Function aims to compute the incremental of $\mathcal{U}$'s influence spread while reconnecting edge $e$ in graph $G_t$. It takes sketch subgraphs $G_{sg}$, query edge $e$ and group $\mathcal{U}$, influenced marking array $SG$, two-step bound $UB_1$ and $UB_2$, and returns the influence spread expansion value of $e$ to $\mathcal{U}$. We initialize two variable $count$ and $count_R$ as $0$ (Line~\ref{alg:fun1}). Then, an inner loop fetches the total number of the reached nodes $v$ for $e$ in each sketch subgraph $G^k_{sg}\in G_{sg}$ but not be reached by $\mathcal{U}$ (\textit{i.e.,} $SG[k][v] == 0$), and we use $count$ to record it (Lines~\ref{alg:fun2}-\ref{alg:fun+}). Meanwhile, if the first step upper bound of $e$ is never updated (\textit{i.e.,} $\mathcal{L}.flag(e) == 0$), we further compute the total number of nodes reached by $e$ in each sketch graph of $G_{sg}$, and store the result in $count_R$ (Lines~\ref{alg:fun++} - \ref{alg:fun3}). Next, we update the value of $UB_2$ as $count / \theta_3$, which is also the influence spread expansion value of $e$ (Line~\ref{alg:fun4}); we also update $UB_1(e)$ and remark $flag(e) = 1$ of $\mathcal{L}$ when the original mark $\mathcal{L}.flag(e) = 0$ (Line~\ref{alg:fun5} - \ref{alg:fun6}).      
Finally, the influence spread expansion of $e$ is returned (Line~\ref{alg:fun7}).
%\vspace{1mm}
%\noindent
%\textbf{Remark. } 
%\begin{remark}
It is remarkable that with the increasing number of RT$l$R queries for different given users group $\mathcal{U}$, the more edges' first-step upper bound $UB_1$ will be narrowed, so as to the performance of the later RT$l$R query with new users group will increase with no additional cost. %This is one of the bright spots of our Sketch-Estimate Function in Algorithm~\ref{alg:estimation}.
%\end{remark}

%\vspace{1mm}
\noindent
\textbf{Complexity.} It is easy for us to derive that the time complexity and space complexity of Algorithm~\ref{alg:estimation} are $\mathcal{O}(\theta_3\cdot |E_t|)$ and $\mathcal{O}(\theta_3 \cdot |V|)$, respectively. 

\begin{figure}[!t]
    \centering
    \includegraphics[scale=0.53]{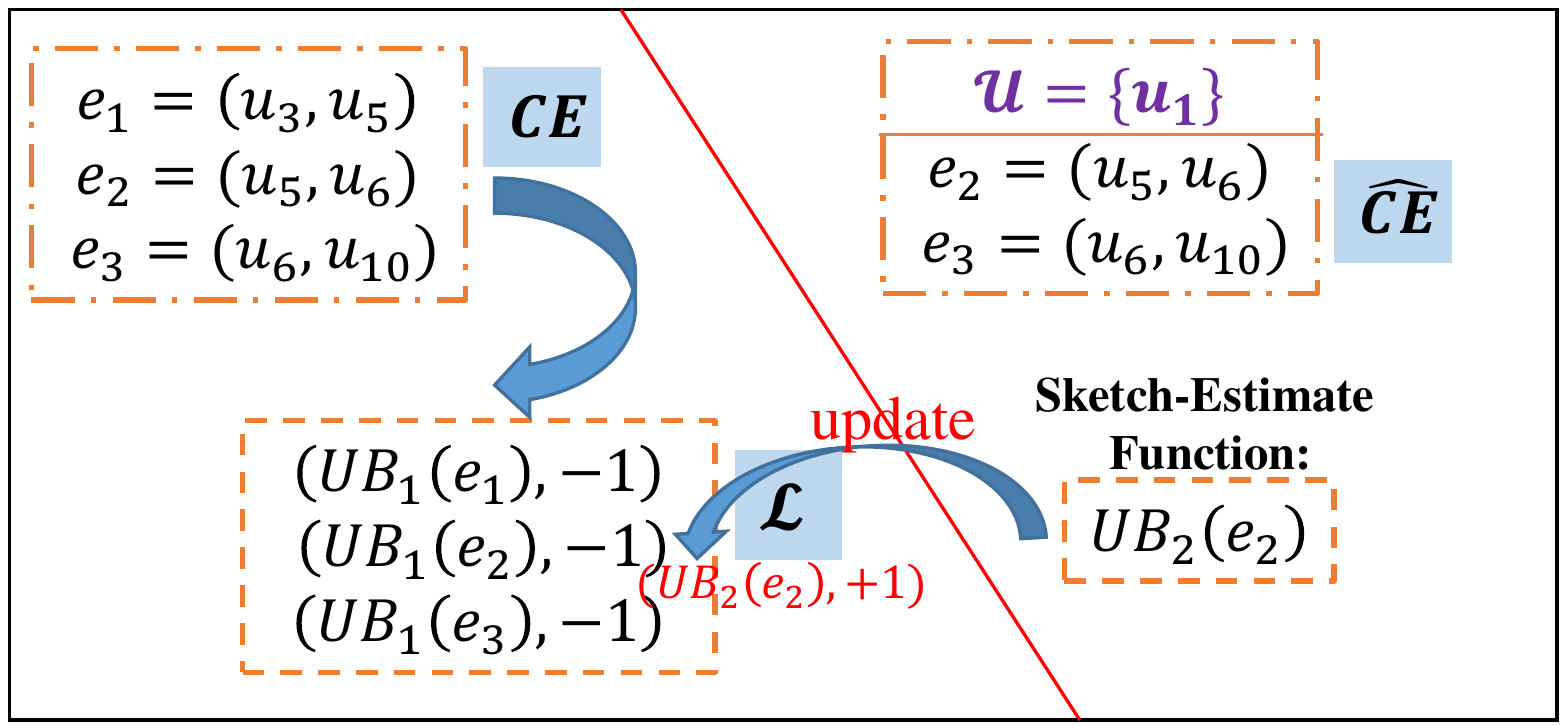}
    \caption{The Two-Step-Bounds Example}
    \label{fig:bounds}
%    \vspace{-3mm}
\end{figure}

\begin{example}
\textit{
Figure~\ref{fig:bounds} shows a running example of our two step bounds generation. For a given graph $G_t$ in Figure~\ref{fig:CE-SBG2}, we first identify the candidate edges set $CE = \{(u_3, u_5), (u_5, u_6), (u_6,u_{10})\}$. Then, we compute the first-step bound of each edge in $CE$ (\textit{e.g., $UB_1(u_3,u_5)$}) and set its initial flag as $-1$. During the process of each RT$l$L query with different given users group $\mathcal{U}$, we will prune the candidate edges set from $CE$ to $\widehat{CE}$, and call the Sketch-Estimate Function (\textit{e.g., Algorithm~\ref{alg:estimation}}) to estimate the influence spread expansion of each probing edge (e.g., $e_2 = (u_5,u_6)$) from $\widehat{CE}$. Meanwhile, during the above process, we get a byproduct of $e_2$, the second step upper bound $UB_2(e_2)$, which can be used to narrow the first upper bound of $e_2$ (\textit{e.g., $UB_1(e_2) \leftarrow UB_2(e_2)$}). Once $UB_1(e_2)$ is updated, $e_2$'s flag also needs to be changed to $+1$.  }
\end{example}

\subsection{Order-based SBG for RT$l$R Query Processing} \label{subsec:order}

\begin{algorithm}[t!] 
	\caption{\textbf{RT$l$R:} Order-based SBG} 
	\label{alg:orderSBG}
	\KwIn{$l$: the number of selected edges, $\mathcal{U}$:  users group, $CE = \mathcal{G} \setminus G_t$: candidate edges, and ($\mathcal{L}$, $G_{sg}$): UBL}
	\KwOut{ $\widehat{S}:$ the optimal Reconnecting edge set} 
	\BlankLine

	Initialize $\widehat{S}\leftarrow \emptyset$, Priority queue $Q$, and Array $SG[\theta_3][|V|]$; \\  \label{alg:order1}
    $\widehat{CE}\leftarrow$ Reducing CE \# BFS $($CE$, \mathcal{U})$; /*using Algorithm~\ref{alg:BFS_prunning_edges} */ \\  \label{alg:order2}
	\For {each edge $e=(u,v)\in \widehat{CE}$}{ \label{alg:order3}
	    $UB_2.e \leftarrow \mathcal{L}.{UB_1(e)}$; \ \ \ 
        push $(e, UB_2.e)$ into $Q$; \\ \label{alg:order4}
    }
    \For {$i = 1$ to $\theta_3$}{  \label{alg:order5}
         \For{each $u \leftsquigarrow \mathcal{U}$ in $G^i_{sg}$}{
               $SG[i][u]\leftarrow 1$; \\ \label{alg:order6}
         }
    }

	\For {$j = 1$ to $l$}{   \label{alg:order7}
	    $(e',UB_2.{e'}) \leftarrow Q.front$ ; \ \ \
	    $I_{max} \leftarrow 0$; \ \ \ $\widehat{e} \leftarrow e'$; \\
	    \While{$I_{max} < UB_2.{e'}$}{
	          $I_{max} \leftarrow$ Sketch-Estimate$(\mathcal{U}, e')$; \\
	          $\widehat{e} \leftarrow e'$; 
	          $(e',UB_2.{e'}) \leftarrow Q.front$ ; \\  
	    }
        $\widehat{S}\leftarrow \widehat{S}\cup \widehat{e}$;\\  \label{alg:order8}
%        \textcolor{red}{
        \For{each edge $e_{ce}=(u,v)\in CE \setminus \widehat{CE}$ $\&\&$ $e_{ce} \leftsquigarrow \widehat{e}$}{
             $\widehat{CE} \leftarrow \widehat{CE} \cup e_{ce}$; 
              $UB_2.e_{ce} \leftarrow \mathcal{L}.{UB_1(e_{ce})}$; \\ 
              push $(e_{ce}, UB_2.e_{ce})$ into $Q$; \\ \label{alg:order4}
        }%}
        \For{$m = 1$ to $\theta_3$}{  \label{alg:order9}
            $\widehat{e}=(\widehat{u},\widehat{v})$; \\
            \If{$SG[m][\widehat{u}]==1$ $\&\&$ $SG[m][\widehat{v}]==0$}{
            $SG[m][\widehat{v}]\leftarrow 1$; \\
            \For{each $u' \leftsquigarrow \widehat{e}$ in $G^m_{sg}$}{
                $SG[m][u']\leftarrow 1$; \\
            }
            }
            \Else{
                 continue; \\ \label{alg:order10}
            }
        }
    }
\Return{$\widehat{S}$}  \label{alg:order11}
\end{algorithm}
In the previous parts of this section, we have overviewed the main idea of our order-based SBG algorithm. We also have introduced the details of two essential blocks of our Order-based SBG algorithm: (i) the UBL construction and (ii) the Sketch-Estimation Function.
%, respectively, which are two essential blocks of order-based SBG algorithm. 
In the rest of this section, we will discuss the details of the Order-based SBG algorithm.

The details of the Order-based SBG algorithm are described in Algorithm~\ref{alg:orderSBG}. It takes an integer $l$, a users group $\mathcal{U}$, the candidate edges $CE$, and UBL index $(\mathcal{L}, G_{sg})$ as inputs, and returns a set $\widehat{S}$ of $l$ optimal reconnecting edges that maximizes the influence spread of $\mathcal{U}$. We initialize a set $\widehat{S}$ as empty, an empty Priority queue $Q$ that will be used to 
%storage 
store
the $UB_2$ information of candidate edges related to $\mathcal{U}$, and an array $SG$ to mark whether a node can be reached by $\mathcal{U}$ or edges in $\widehat{S}$ at each sketch subgraphs $G_{sg}$ (Line~\ref{alg:order1}). Then, we reduce the candidate edges from $CE$ by using Algorithm~\ref{alg:BFS_prunning_edges}, and record the reduced candidate edges into set $\widehat{CE}$ (Line~\ref{alg:order2}). For each edge $e$ in $\widehat{CE}$, we get $e$'s first-step upper bound $UB_1(e)$ from UBL index, and set $UB_1(e)$ as the initial second-step upper bound value of $e$ (\textit{i.e.,} $UB_2(e) = \mathcal{L}.UB_1(e)$), and then push $(e, UB_2(e))$ into priority queue $Q$ (Lines~\ref{alg:order3} - \ref{alg:order4}). Next, we mark the nodes which are reached by $\mathcal{U}$ in each sketch subgraphs of $G_{sg}$ (Lines~\ref{alg:order5} - \ref{alg:order6}). Further, in each iteration, we probe the candidate edges in priority queue $Q$ in order based on their $UB_2$ value, and then call Sketch-Estimation Function to compute the influence spread expansion of the probing edge $e$, the edge probing in this iteration will be early terminated once the front edge from Q is less than the currently maximum influence spread expansion value (Line~\ref{alg:order7}). After finding out the optimal edge $\widehat{e}$, we update the mark of nodes reached by $\widehat{e}$ in each sketch subgraphs (Lines~\ref{alg:order8} -\ref{alg:order9}). Finally, it returns the optimal reconnecting edge set $\widehat{S}$ having maximum influence spread expansion of $\mathcal{U}$ (Line~\ref{alg:order11}).

\vspace{1mm}
\noindent
\textbf{Complexity.} The time complexity of Algorithm~\ref{alg:orderSBG} is $\mathcal{O}(|\widehat{CE}| + \theta_3\cdot |E_t| + l\cdot|\widehat{CE}|\cdot\theta_3\cdot |E_t|)$. Besides, the space complexity is $\mathcal{O}(|CE| + \theta_3\cdot |E_t|)$. Although the time complexity of Algorithm~\ref{alg:orderSBG} is not significantly better than the SBG algorithm, it can greatly reduce the number of candidate edges probing for influence spread estimation, which is the bottleneck of the SBG algorithm.

%% file: Experiments.tex
\section{Experimental Evaluation}\label{sec:exp}
In this section, we present the experimental evaluation of our proposed approaches for the RT$l$L queries: the sketch based greedy algorithm (\textbf{SBG}) in Section~\ref{subsec:sbg}; the candidate edges pruning method to accelerate SBG (\textbf{CE-SBG}) in Section~\ref{subsec:optim1}; and the Order-based SBG solution (\textbf{O-SBG}) in Section~\ref{subsec:order}. %The source codes of this work are available at \url{https://github.com/Jade-Ray/Reconnecting_Top-l_Relationships}.

\subsection{Experimental Setting}
We implement the algorithms using Python 3.6 on Windows environment with 2.90GHz Intel Core i7-10700 CPU and 64GB RAM. 

\vspace{1mm}
\noindent
\textbf{Baseline. } To the best of our knowledge, no existing work investigates the RT$l$R problem. To further validate, we use our SBG algorithm as the baseline algorithm to compare with CE-SBG and O-SBG. This is because the well-known RIS based IM methods~\cite{tang2014influence,nguyen2016stop} are hardly used in the RT$l$R query (\textit{i.e., mentioned in Section~\ref{subsec:IManlyasis}}). Meanwhile, our SBG algorithm is extended from the FI-sketch IM method (\textit{i.e., SG algorithm~\cite{10.1145/2661829.2662077}}), while the SG algorithm performs well within the existing IM efforts, which has been validated in the state-of-the-art IM benchmark study~\cite{arora2017debunking}.    

%well-know SG algorithm to answer the IM problem, and the performance has been validate in the state-of-the-art IM benchmark study~\cite{arora2017debunking}.   

%compared with the baseline adapted from the traditional \textit{Monte Carlo} (\textbf{MC}) method~\cite{IM2003}.
\iffalse
\begin{table*}[t]
	\caption{The Description of Dataset \label{tab:datasets}}
	\begin{center}
		\scalebox{1.0}{
			\begin{tabular}{|c|c|c|c|c|c|c|} \hline
				{\bf Dataset}   & Nodes       & Temporal Edges  & Edges in static graph & $degree_{avg}$     & Time (days) & Type      \\\hline
				\hline 
				eu-core         & 986         & 332,334         & 24,929     & 25.28   & 803    & Directed \\\hline
				CollegeMsg      & 1,899       & 59,835          & 20,296     & 10.69   & 193    & Directed \\\hline
				mathoverflow    & 21,688      & 107,581         & 90,489     & 4.17   & 2,350  & Directed  \\\hline
				ask-ubuntu      & 137,517     & 280,102         & 262,106    & 1.91   & 2,613  & Directed  \\\hline    
%				wiki-talk       & 1,140,149   & 7,833,140 & 2,320  & Directed  \\\hline
				stack-overflow  & 2,464,606   &17,823,525       & 16,266,395   & 6.60   & 2,774  & Directed  \\\hline
			\end{tabular}
		}
	\end{center}
\end{table*}
\fi

\begin{table}[t!]
	\caption{The Description of Dataset \label{tab:datasets}}
	\vspace{-2mm}
	\begin{center}
		\scalebox{0.9}{
			\begin{tabular}{|c|c|c|c|c|c|c|} \hline
				{\bf Dataset}   & Nodes       & Temporal Edges    & $d_{avg}$     & Days    & Type      \\\hline
				\hline 
				eu-core         & 986         & 332,334           & 25.28      & 803     & Directed \\\hline
				CollegeMsg      & 1,899       & 59,835            & 10.69      & 193     & Directed \\\hline
				mathoverflow    & 21,688      & 107,581           & 4.17       & 2,350   & Directed  \\\hline
				ask-ubuntu      & 137,517     & 280,102           & 1.91       & 2,613   & Directed  \\\hline   
%				wiki-talk       & 1,140,149   & 7,833,140 & 2,320  & Directed  \\\hline
				stack-overflow  & 2,464,606   &17,823,525         & 6.60       & 2,774   & Directed  \\\hline
			\end{tabular}
		}
	\end{center}
\end{table}

\vspace{1mm}
\noindent
\textbf{Datasets.} We conduct the experiments using five publicly available datasets from the \textit{Large Network Dataset Collection}~\footnote{http://snap.stanford.edu/data/index.html}: \textit{eu-core}, \textit{CollegeMsg}, \textit{mathoverflow}, \textit{ask-ubuntu}, and \textit{stack-overflow}. The statistics of the datasets are shown in Table~\ref{tab:datasets}. We have averagely divided all datasets into $T$ graph snapshots (\textit{e.g., $G_t = (V, E_t)$, $t\in [1,T]$}), where $V$ is the node and $E_t$ is the edges appearing in the time period of $t$ in each dataset.  

%\textit{wiki-talk}, and \textit{stack-overflow}

\vspace{1mm}
\noindent
\textbf{Parameter Configuration. }
Table~\ref{tab:parameter} presents the parameter settings. We consider four parameters in our experiments: the number of queries $Q$, the size of given users group $|\mathcal{U}|$, reconnecting edges size $l$, and the number of snapshots $T$. Besides, the near future snapshot $G_{T+1}$ is generated by using the recent link prediction method~\cite{zhang2021labeling} %(\textit{i.e., source codes link: \url{https://github.com/facebookresearch/SEAL_OGB}}). 
In each experiment, if one parameter varies, we use the default values for the other parameters. 
Besides, we set $\theta_3 = 200$, which is consistent with~\cite{arora2017debunking}. %the state-of-the-art IM benchmark study~\cite{arora2017debunking}. %All the programs are implemented in Python3. All experiments were conducted on a machine with 24 CPU cores and 64 GB main memory running Windows10. Each CPU core is Intel Core i7-10700 2.90 GHz. and $R = 10,000$

\begin{table}[t!]
	\caption{Parameters and their values \label{tab:parameter}}
	\vspace{-2mm}
	\begin{center}
		\scalebox{0.9}{
			\begin{tabular}{|c|p{4.0cm}<{\centering}|c|} \hline
				{\bf Parameter}    & Values                             & Default \\\hline 
				\hline
				$Q$                & $[20, 40, 60, 80, 100]$            & 80      \\\hline   
				$|\mathcal{U}|$    & $[1, 2, 4, 6, 8]$                  & 6       \\\hline
				$l$                & $[1, 10, 20, 40]$ or $[1,2,3,4]$   & 10 or 2     \\\hline
				$T$                & $[20, 40, 60, 80, 100]$            & 100     \\\hline		
			\end{tabular}
		}
	\end{center}
\end{table}

\subsection{Efficiency Evaluation} 
%*Note: The EU-core dataset is only used in the Case study section, do not use the Eu-core dataset in the efficiency and effectiveness comparison. *
We study the efficiency of the approaches for the RT$l$L problem regarding running time under different parameter settings.

\subsubsection{Varying Reconnecting Edges Set Size $l$ }

\iffalse
\begin{figure}[ht]
    \centering
    \includegraphics[scale=0.32]{figures/varying_re.png}
    \caption{Varying $l$}
    \label{fig:l}
\end{figure}
\fi

\begin{figure}[t!]
	\centering
	\subfigure[mathoverflow]{\label{R3:exp:vary_L1}
		\includegraphics[scale=0.30]{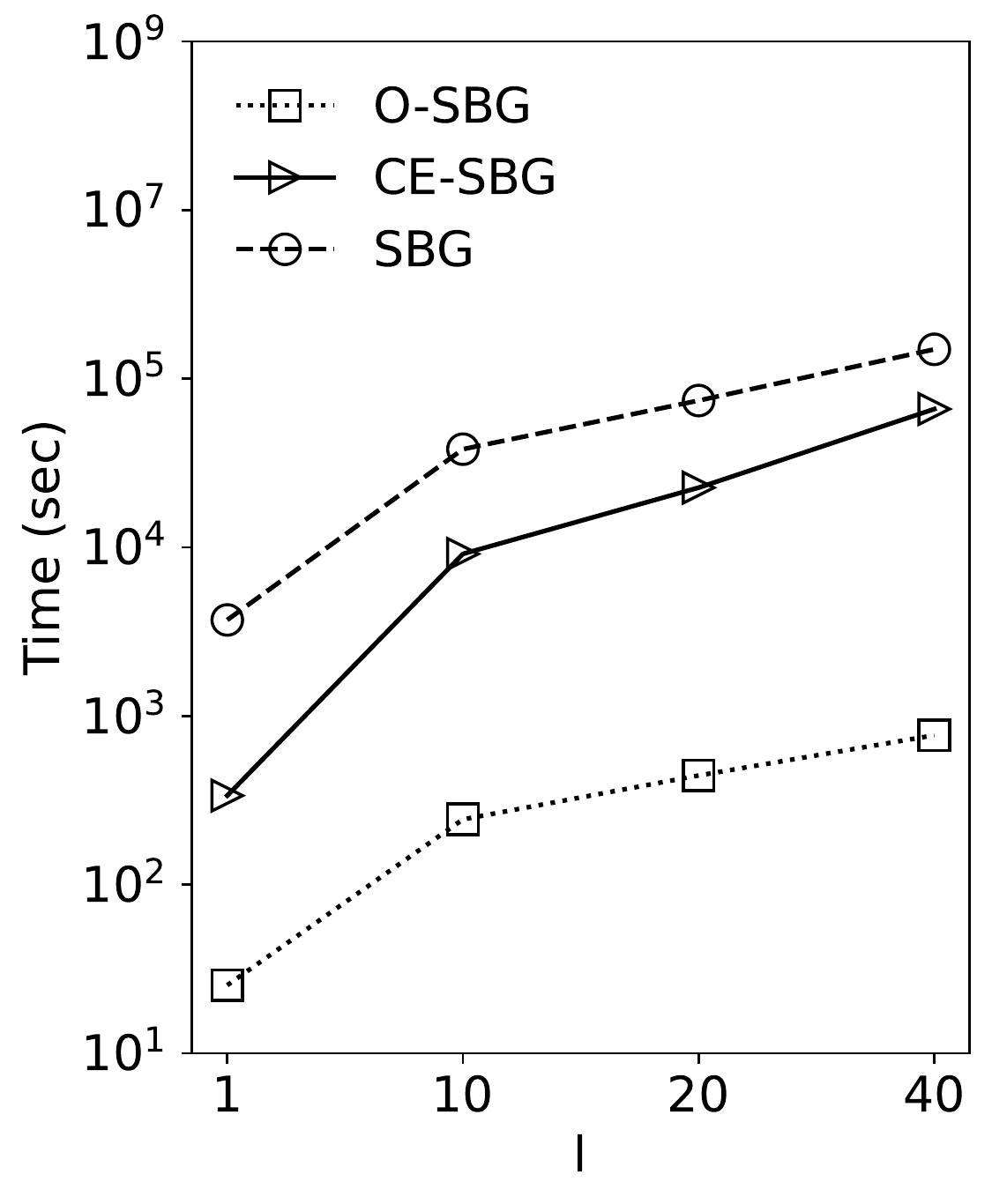}
	}
	\subfigure[ask-ubuntu]{\label{R3:exp:vary_L2}		
		\includegraphics[scale=0.30]{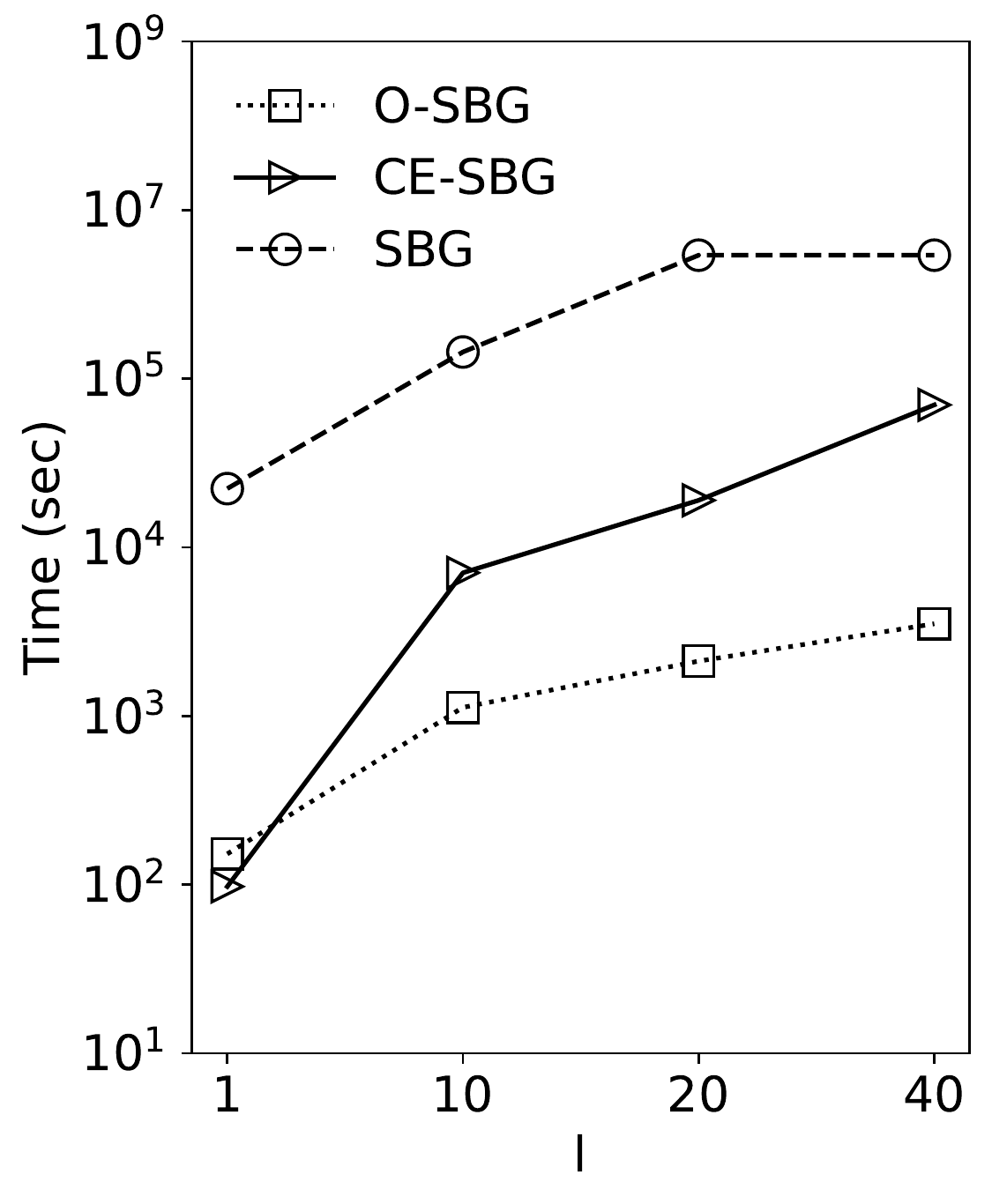}
	}
	\subfigure[CollegeMsg]{\label{R3:exp:vary_L3}		
	\includegraphics[scale=0.30]{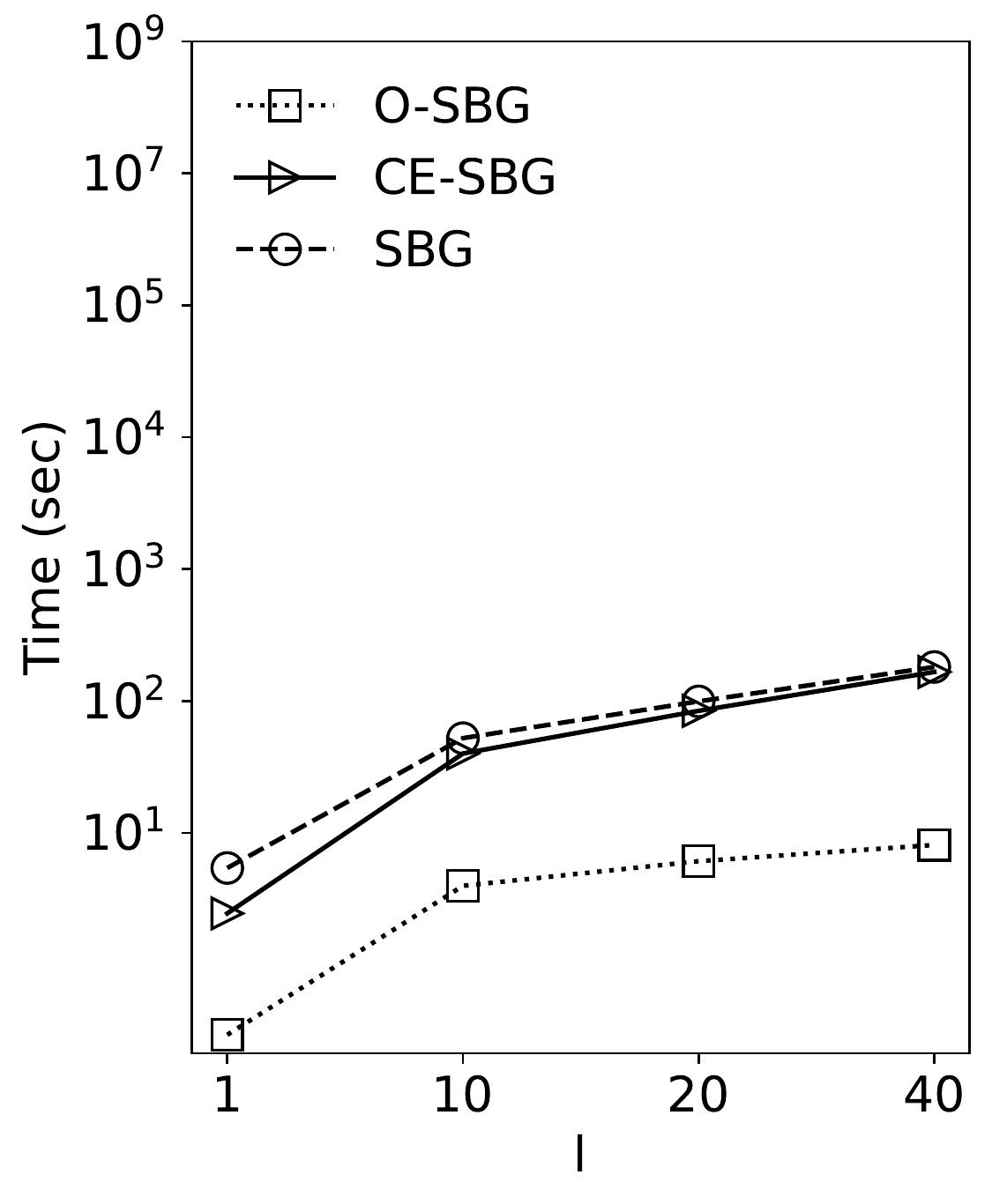}
	}
	\subfigure[eu-core]{\label{R3:exp:vary_L4}		
	\includegraphics[scale=0.30]{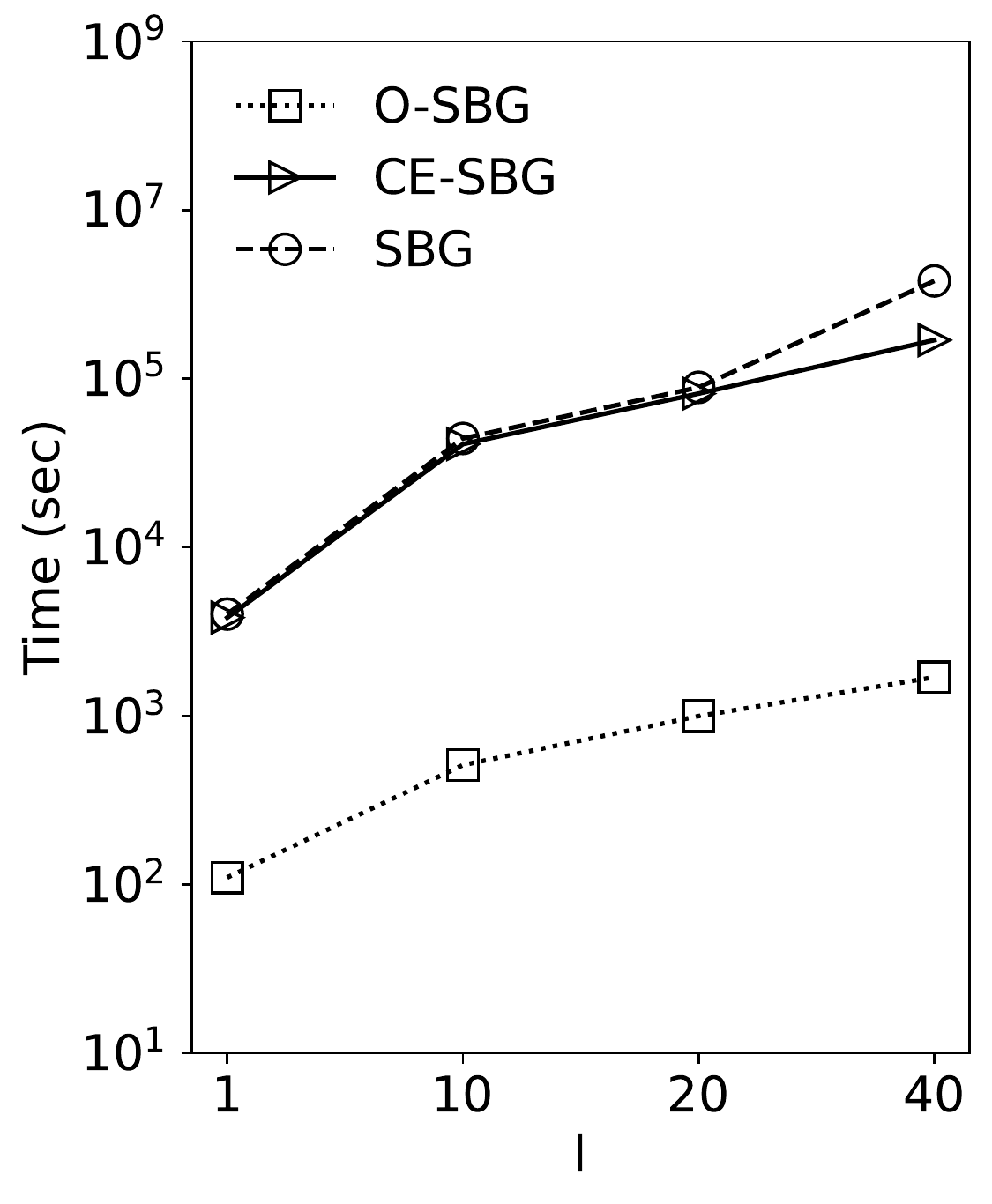}
	}
	\vspace{-3mm}
	\caption{Time cost of algorithms with varying $l$}
%	\vspace{-2mm}
	\label{fig:running_time_L}
\end{figure}

Figure~\ref{fig:running_time_L} shows the average running time of our proposed methods by varying $l$ between $1$ to $40$. The running time of the algorithms follows similar trends, where \textit{SBG} consumes maximum time to process an RT$l$R query. On average, \textit{O-SBG} is $65$ to $99$ times faster than \textit{CE-SBG}, and 90 to 167 times faster than \textit{SBG}. Also, \textit{CE-SBG} is about $3$ to $11$ times faster than \textit{SBG} in different datasets when $l$ varies from $2$ to $40$. Notably, when $l$ is 
%large 
larger than 
$30$, the \textit{SBG} algorithm fails to return the result of the RT$l$R query within one day. As expected, the running time of both three approaches significantly increases when $l$ is varied from $1$ to $40$. Besides, the growth of running time 
%increasing 
in \textit{O-SBG} is much slower than the other two algorithms. This is because the probing candidate edges will increase in all three approaches when $l$ increases, and \textit{O-SBG} has the smallest number of probing candidate edges 
%increasing within 
among 
the three approaches (refer to Figure~\ref{fig:Visited_edges_l}). 

%Figure~\ref{fig:Visited_edges_l} shows the number of probing candidate edges of \textit{SBG}, \textit{CE-SBG}, and \textit{O-SBG} on different datasets by varying $l$. As can be seen, all these three approaches

\iffalse
\begin{figure}
    \centering
    \includegraphics[scale=0.30]{figures/num_prob_edges.png}
    \caption{Number of Probing edges}
    \label{fig:num_edge}
\end{figure}
\fi

%\fi

\begin{figure}[t!]
	\centering
	\subfigure[mathoverflow]{\label{R3:exp:visited_edges_l1}
		\includegraphics[scale=0.28]{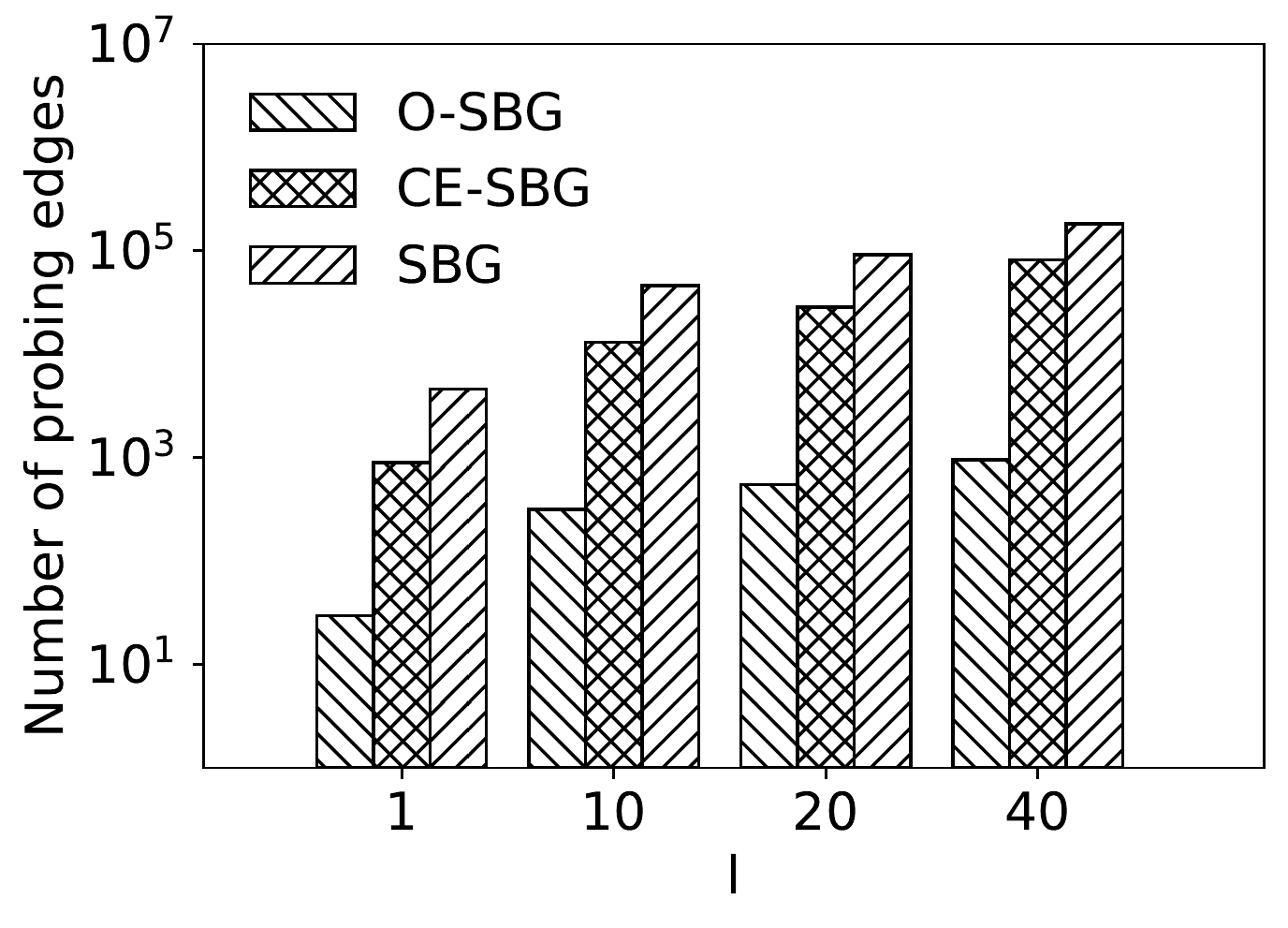}
	}
	\subfigure[ask-ubuntu]{\label{R3:exp:visited_edges_l2}		
		\includegraphics[scale=0.28]{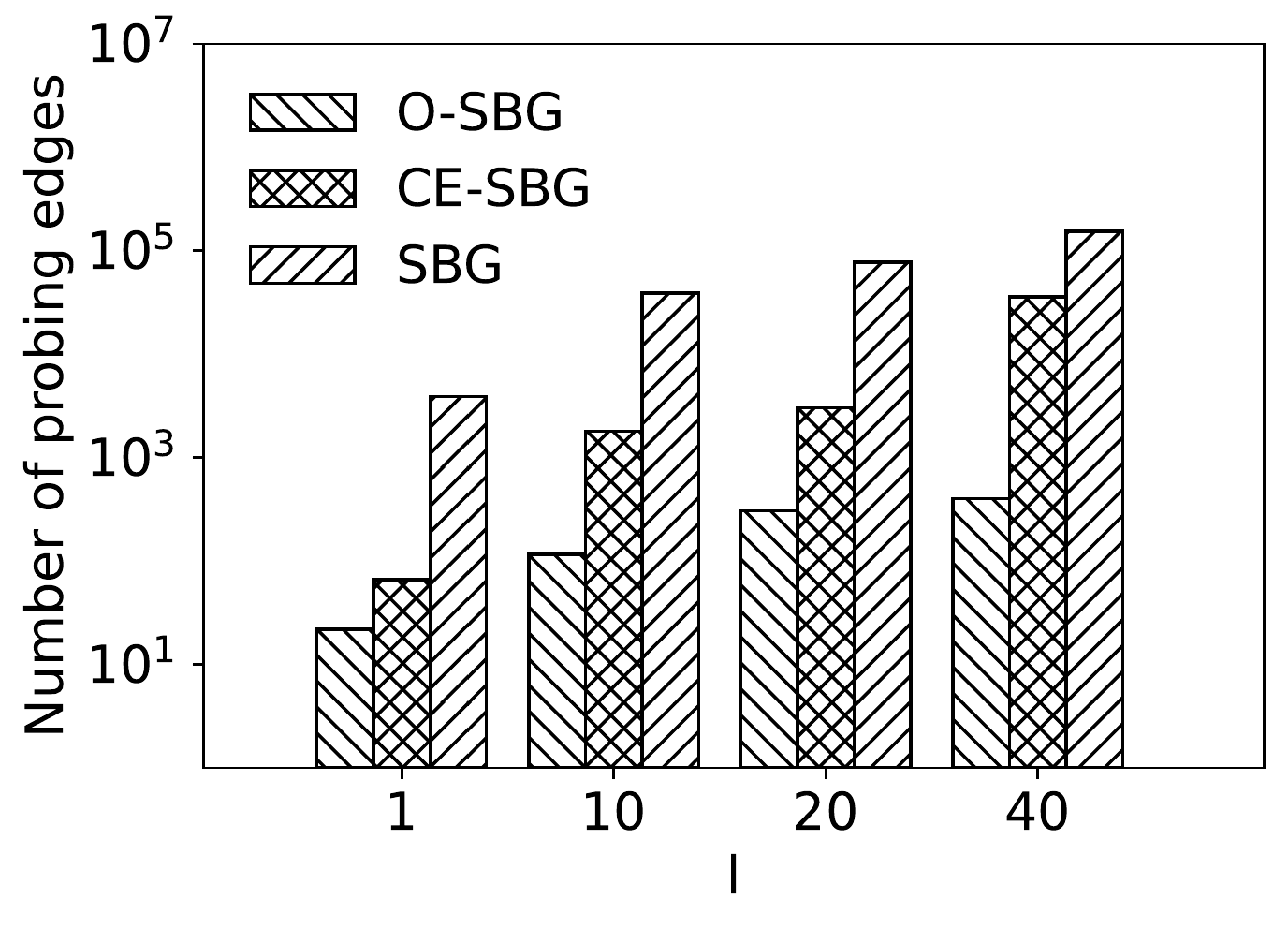}
	}
	\subfigure[CollegeMsg]{\label{R3:exp:visited_edges_l3}		
	\includegraphics[scale=0.28]{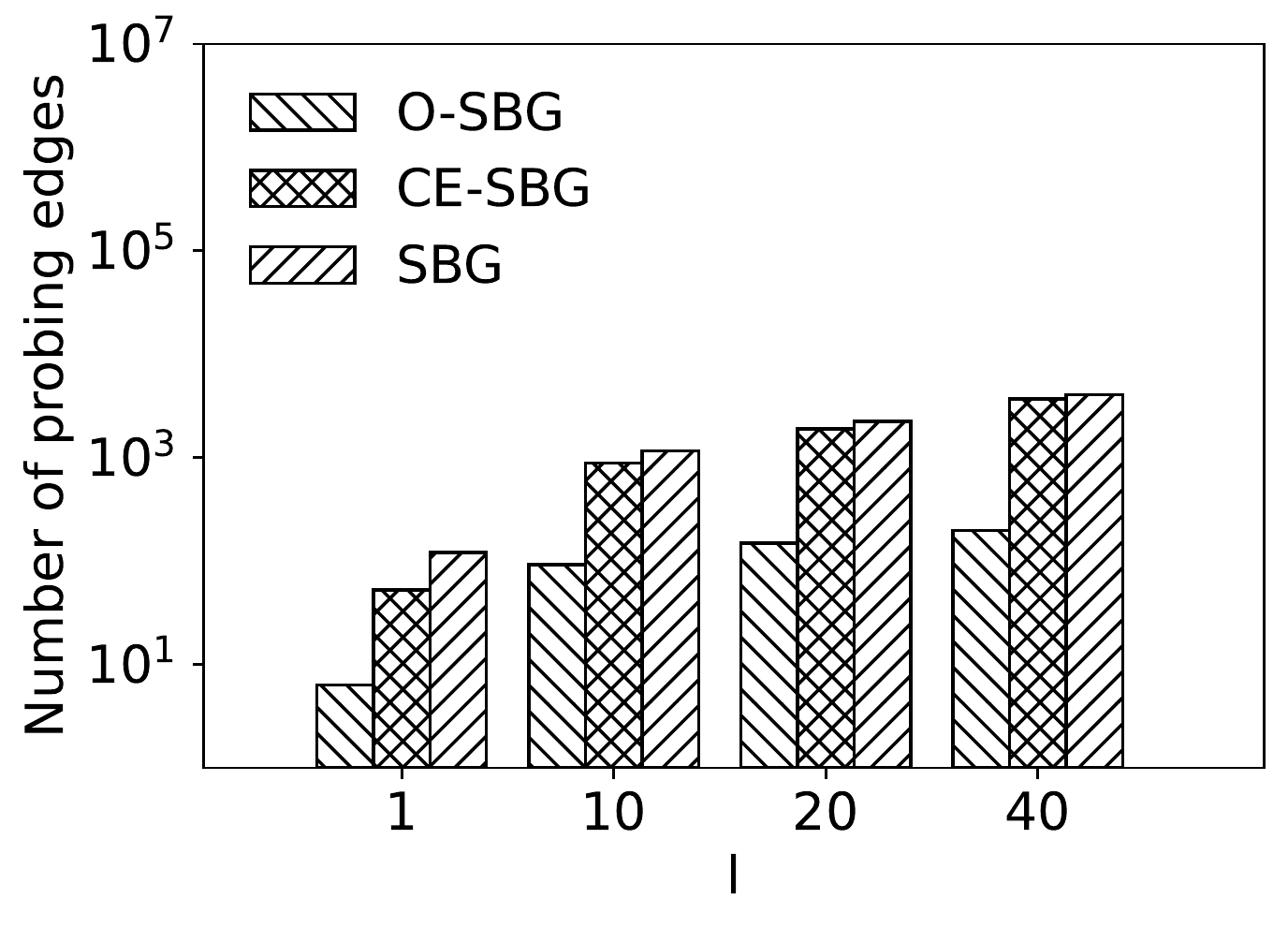}
	}
	\subfigure[eu-core]{\label{R3:exp:visited_edges_l4}		
	\includegraphics[scale=0.28]{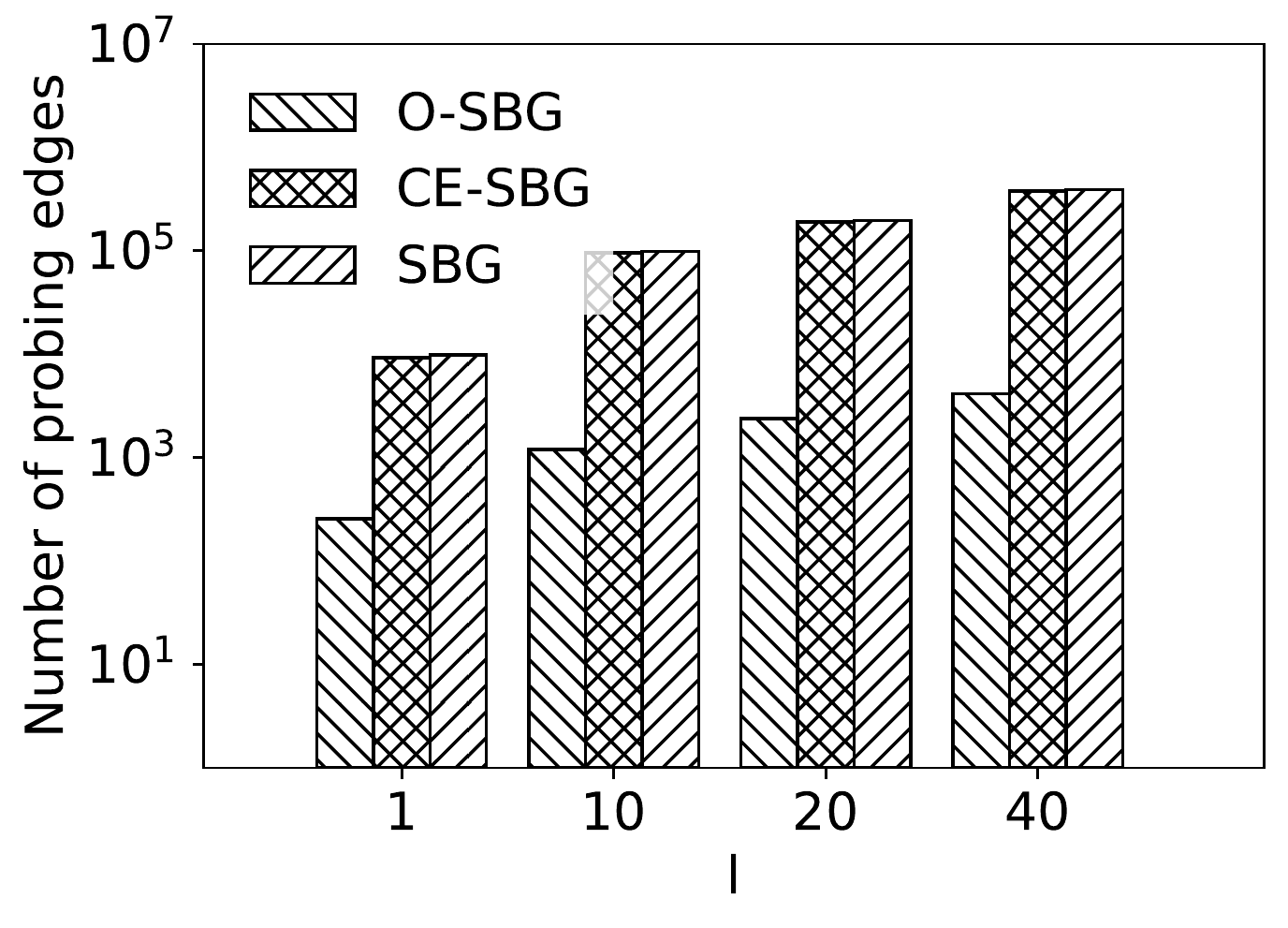}
	}
	\vspace{-3mm}
	\caption{Number of probing edges of algorithms}
%	\vspace{-2mm}
	\label{fig:Visited_edges_l}
\end{figure}

The number of probing candidate edges of \textit{SBG}, \textit{CE-SBG}, and \textit{O-SBG} with varying $l$ are presented in Figure~\ref{R3:exp:visited_edges_l1}-\ref{R3:exp:visited_edges_l4}. As can be seen, the probing candidate edges of \textit{O-SBG} is much less than \textit{SBG} and \textit{CE-SBG} for all values of $l$. For example, when $l=20$, the probing candidate edges of \textit{SBG}, \textit{CE-SBG}, and \textit{O-SBG} in \textit{mathoverflow} are $91,850$, $28,588$, and $547$, respectively.
Besides, the number of probing candidate edges increases in all three approaches with the increase of $l$, and \textit{O-SBG} probing the least number of candidate edges in all three approaches.
This result has verified the above explanation about why \textit{O-SBG} performs better than the other two approaches with varying $l$.

\subsubsection{Varying Number of Queries $Q$}

\iffalse
\begin{figure}[ht]
    \centering
    \includegraphics[scale=0.32]{figures/varying_nq.png}
    \caption{Varying Number of Queries}
    \label{fig:Q}
\end{figure}
\fi

\begin{figure}[t!]
	\centering
	\subfigure[mathoverflow]{\label{R3:exp:vary_Q1}
		\includegraphics[scale=0.30]{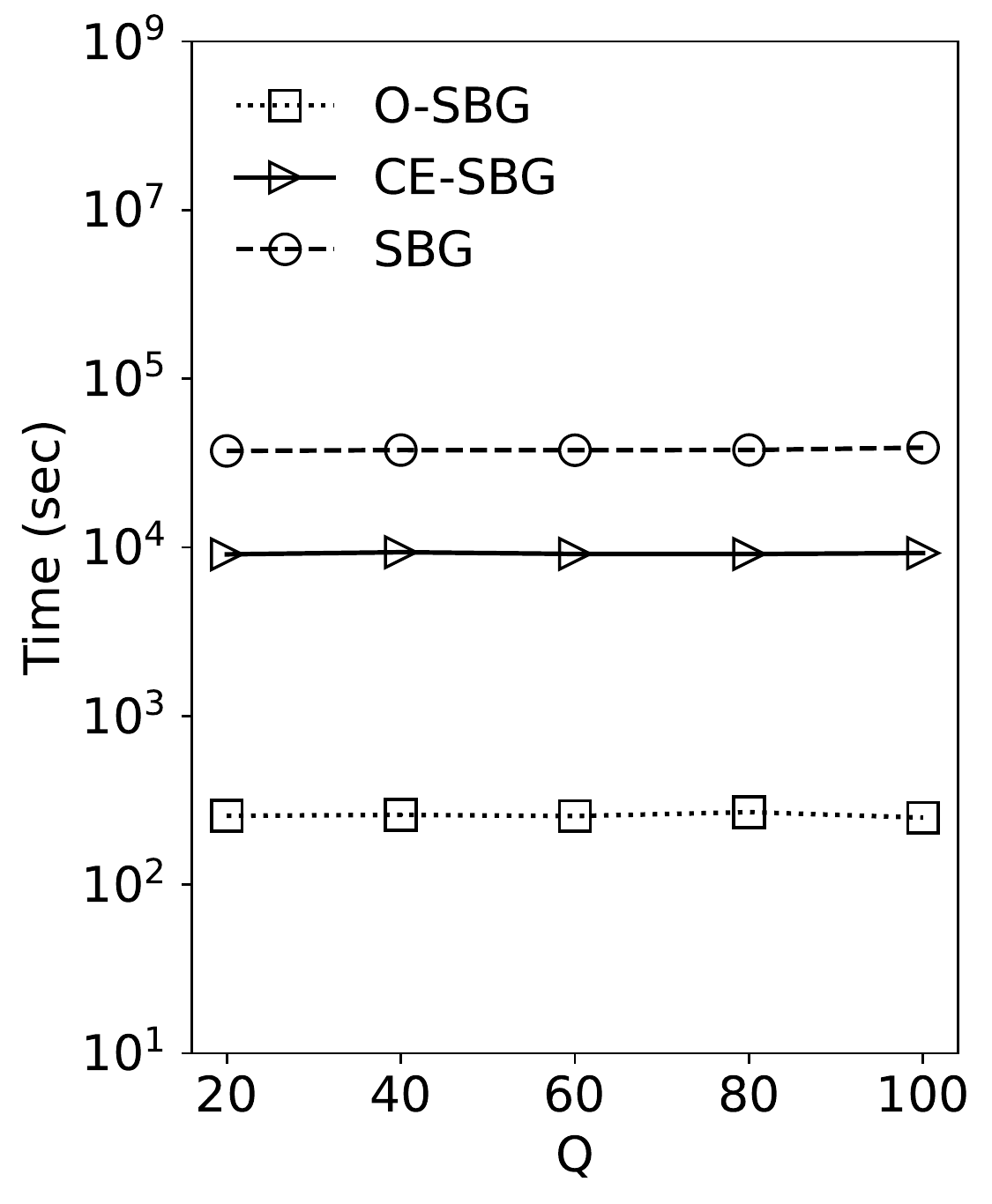}
	}
	\subfigure[ask-ubuntu]{\label{R3:exp:vary_Q2}		
		\includegraphics[scale=0.30]{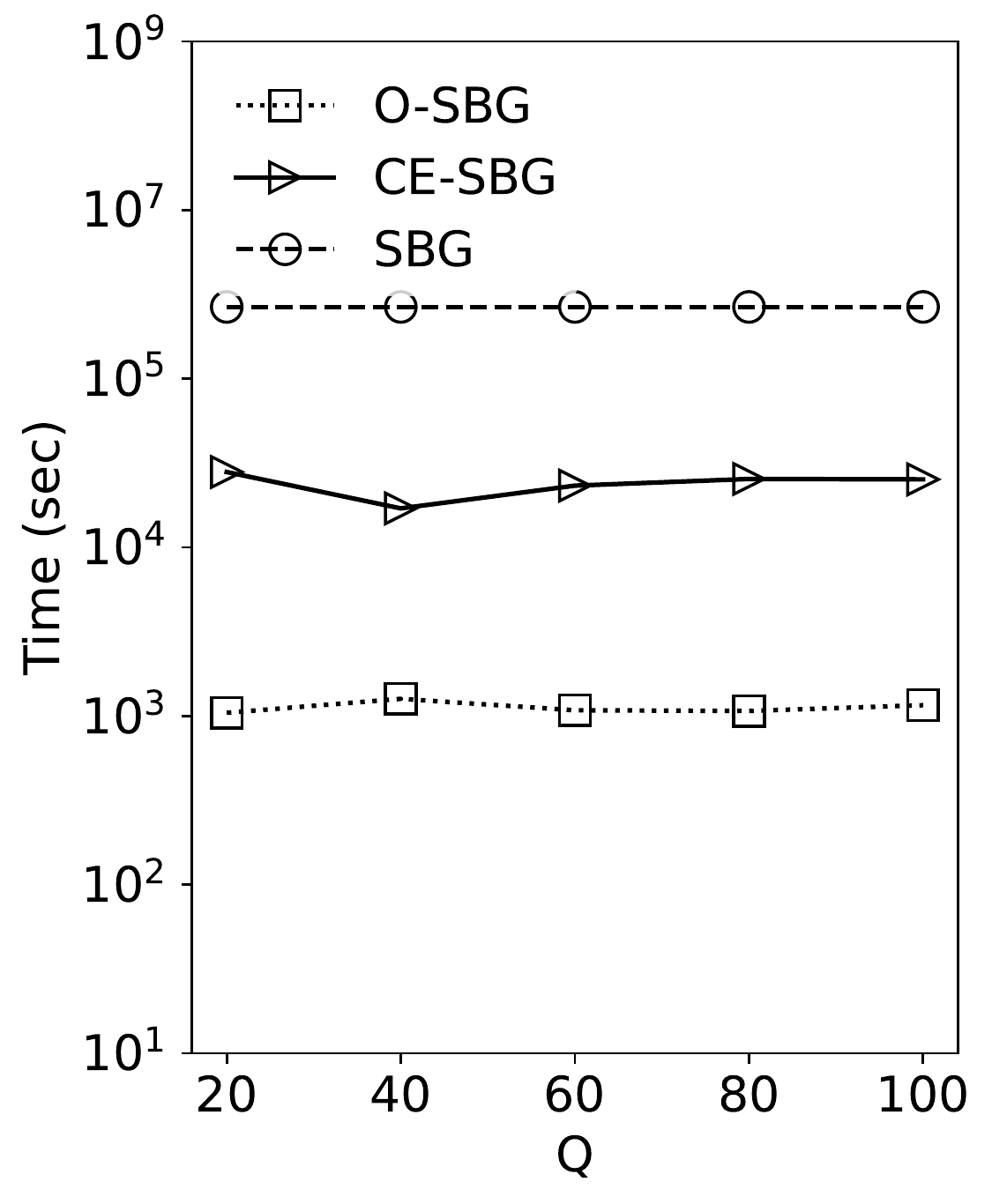}
	}
	\subfigure[CollegeMsg]{\label{R3:exp:vary_Q3}		
	\includegraphics[scale=0.30]{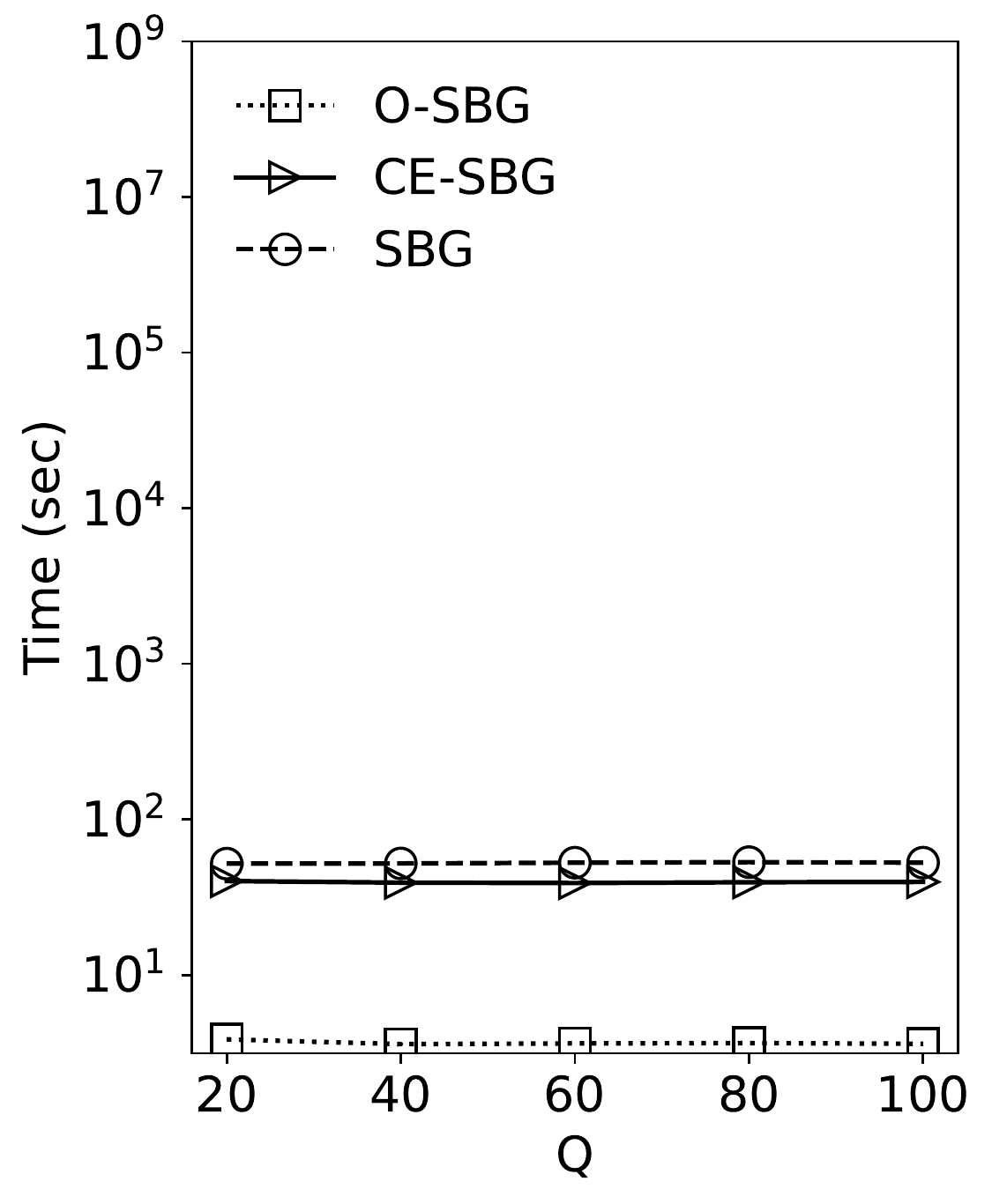}
	}
	\subfigure[eu-core]{\label{R3:exp:vary_Q4}		
	\includegraphics[scale=0.30]{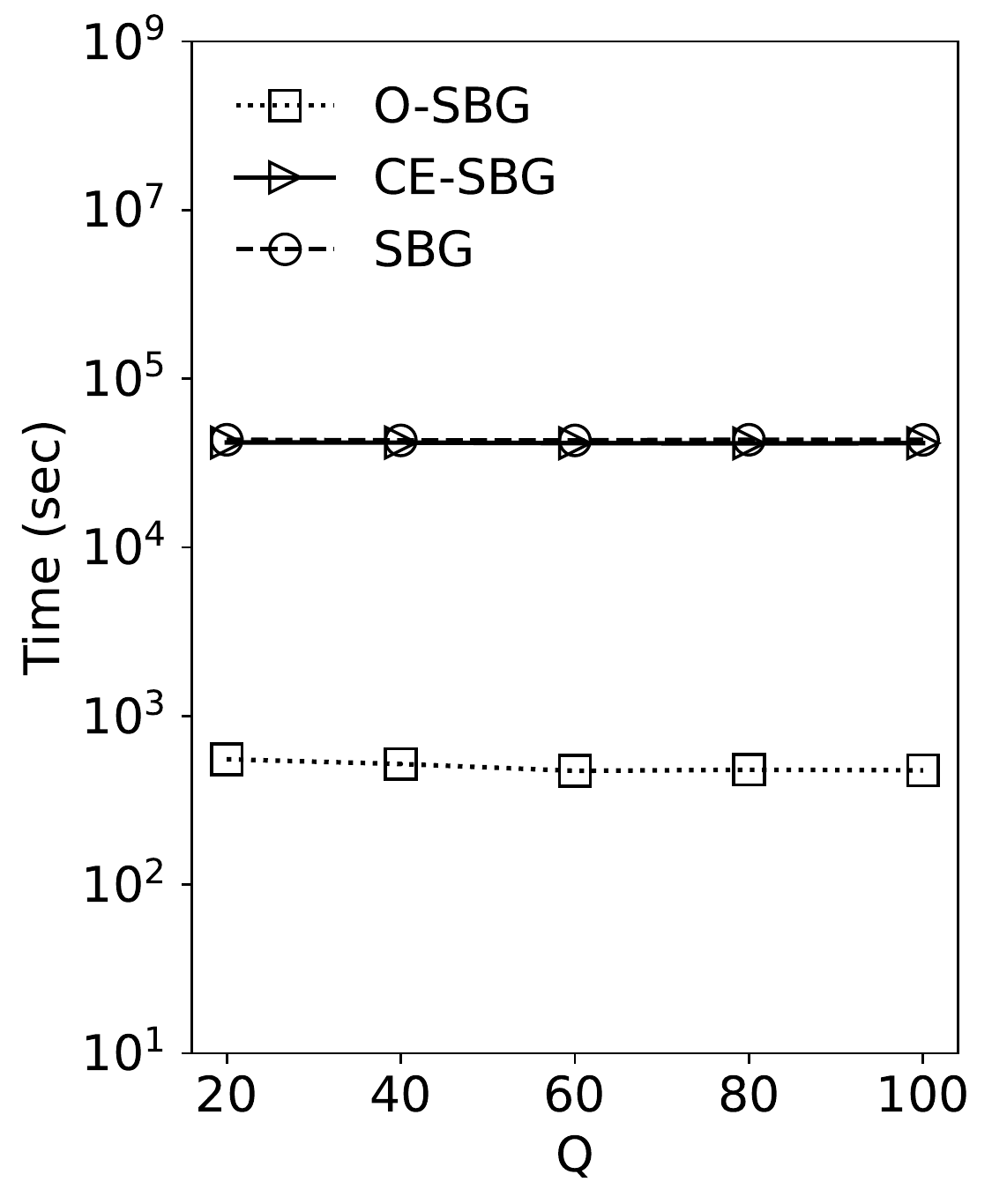}
	}
	\vspace{-3mm}
	\caption{Time cost of algorithms with varying $Q$}
%	\vspace{-2mm}
	\label{fig:running_time_Q}
\end{figure}

We compare the performance of different approaches by varying the number of RT$l$R queries from $20$ to $100$. Figure~\ref{fig:running_time_Q} shows the average running time of \textit{SBG}, \textit{CE-SBG}, and \textit{O-SBG} on the four datasets. As we can see, \textit{O-SBG} is significantly efficient than \textit{SBG} and \textit{CE-SBG}. Specifically, \textit{O-SBG} performs two to three orders of magnitude faster than \textit{SBG} 
%approach 
and one to two orders of magnitude faster than \textit{CE-SBG} in all datasets, respectively. %As expected, we do not observe any noticeable trend from all three approaches when $Q$ is varied. 
%can reduce the running time by around xx times and xx times compared with \textit{SBG} and \textit{CE-SBG} in the \textit{mathoverflow} dataset. %Meanwhile, the running time of \textit{CE-SBG} is less than \textit{SBG}. 

\subsubsection{Varying Users Group Size $|\mathcal{U}|$ }

\iffalse
\begin{figure}[ht]
    \centering
    \includegraphics[scale=0.32]{figures/varying_us.png}
    \caption{Varying $|\mathcal{U}|$}
    \label{fig:U}
\end{figure}
\fi

\begin{figure}[ht]
	\centering
	\subfigure[mathoverflow]{\label{R3:exp:vary_U1}
		\includegraphics[scale=0.30]{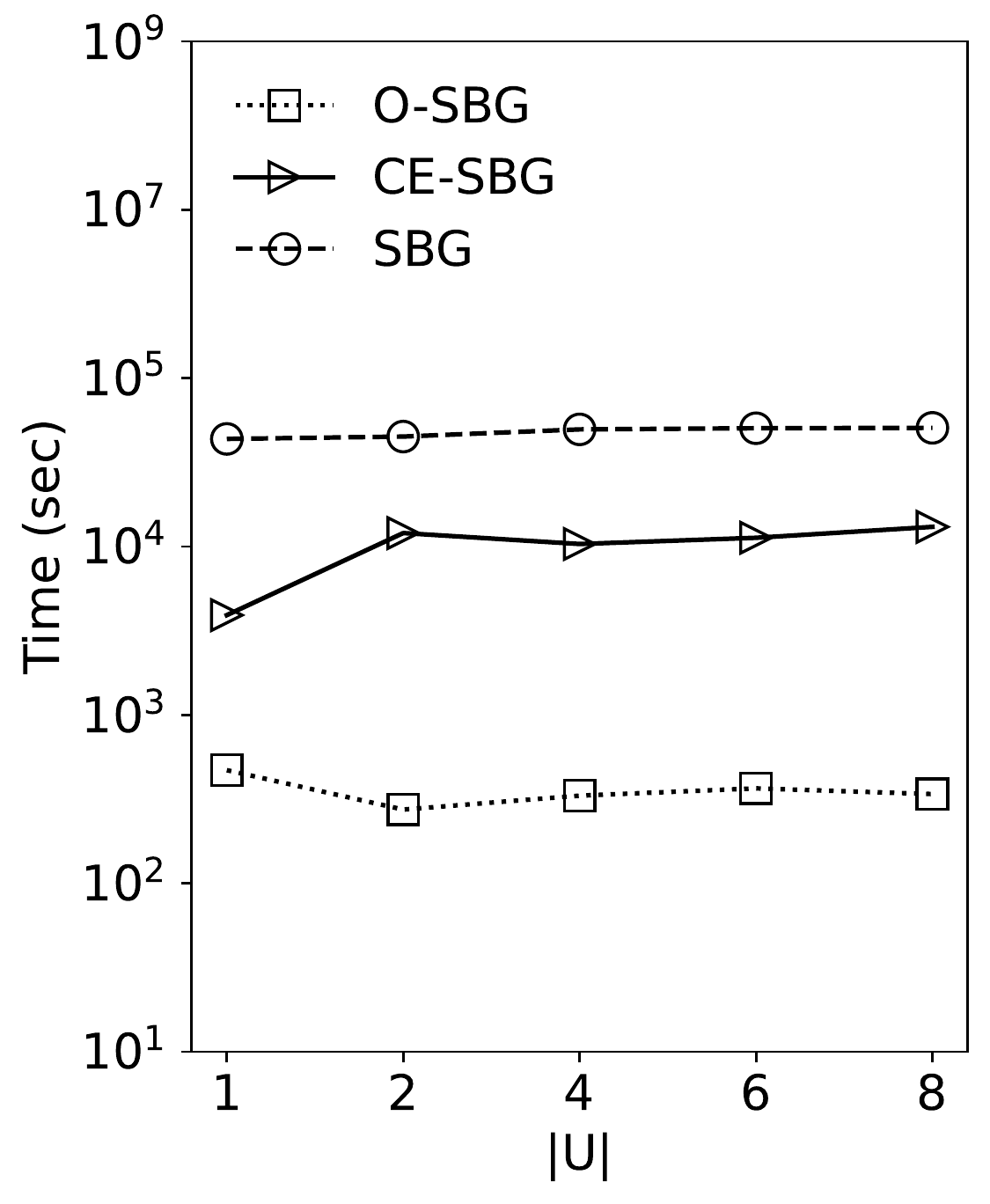}
	}
	\subfigure[ask-ubuntu]{\label{R3:exp:vary_U2}		
		\includegraphics[scale=0.30]{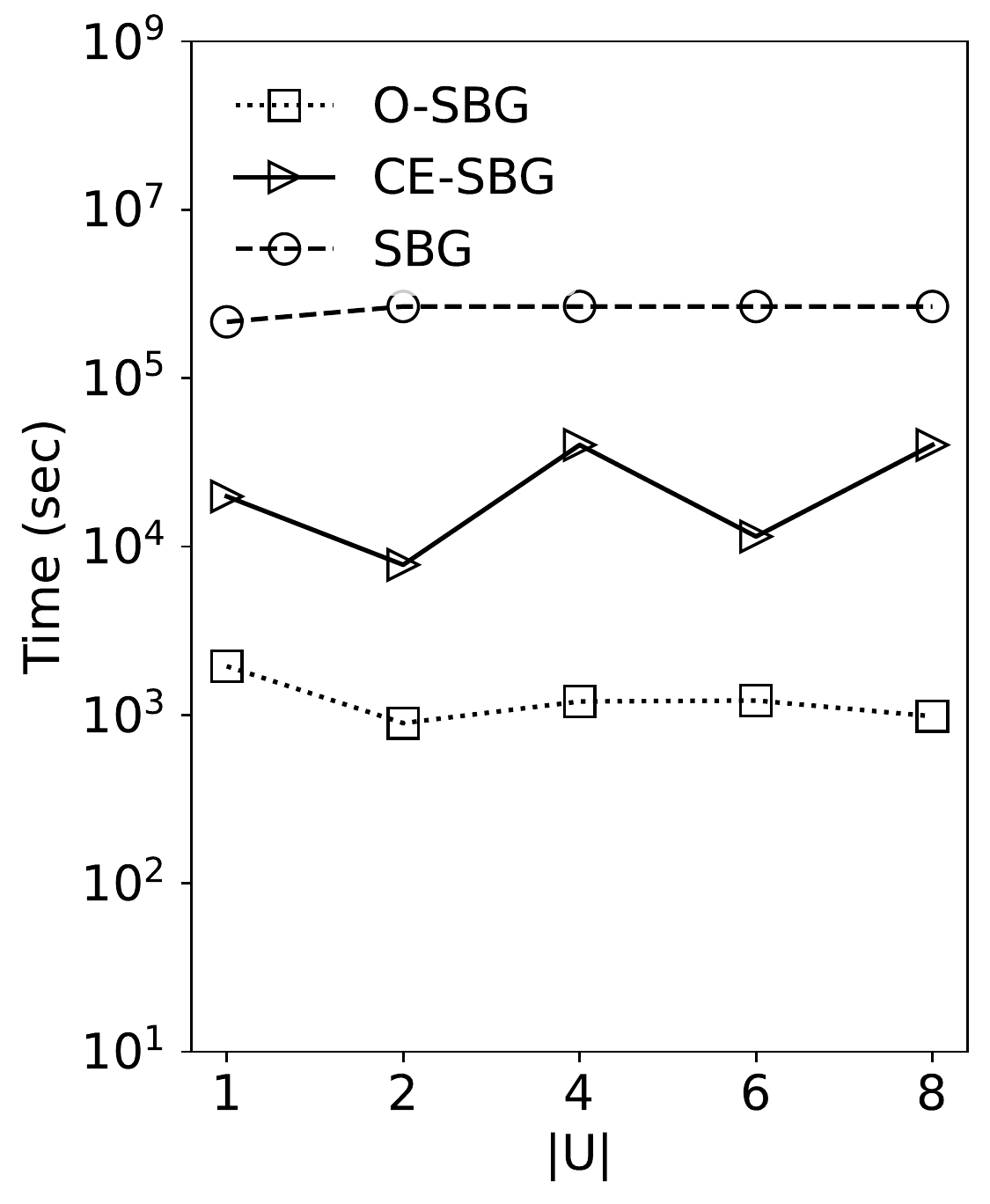}
	}
	\subfigure[CollegeMsg]{\label{R3:exp:vary_U3}		
	\includegraphics[scale=0.30]{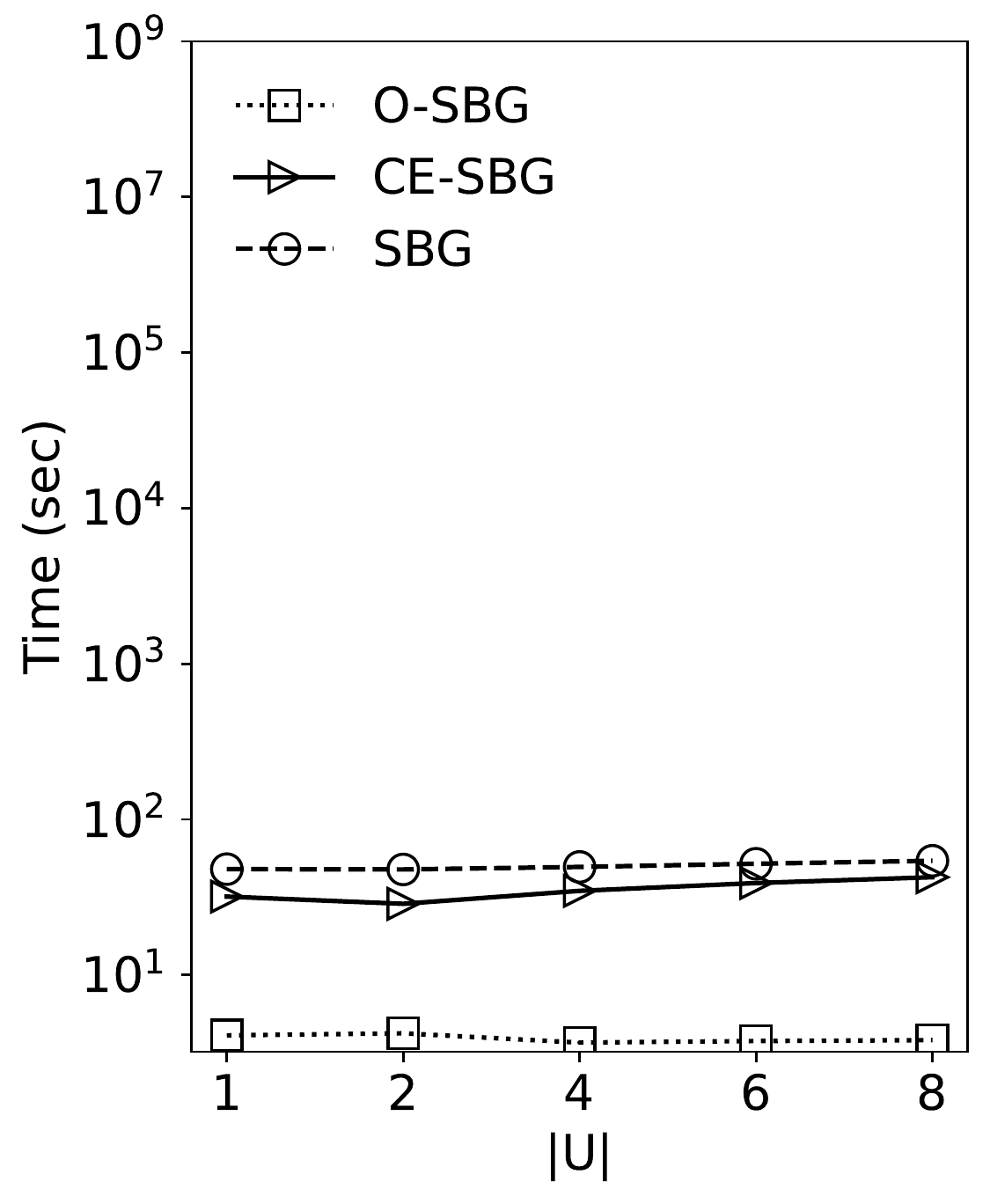}
	}
	\subfigure[eu-core]{\label{R3:exp:vary_U4}		
	\includegraphics[scale=0.30]{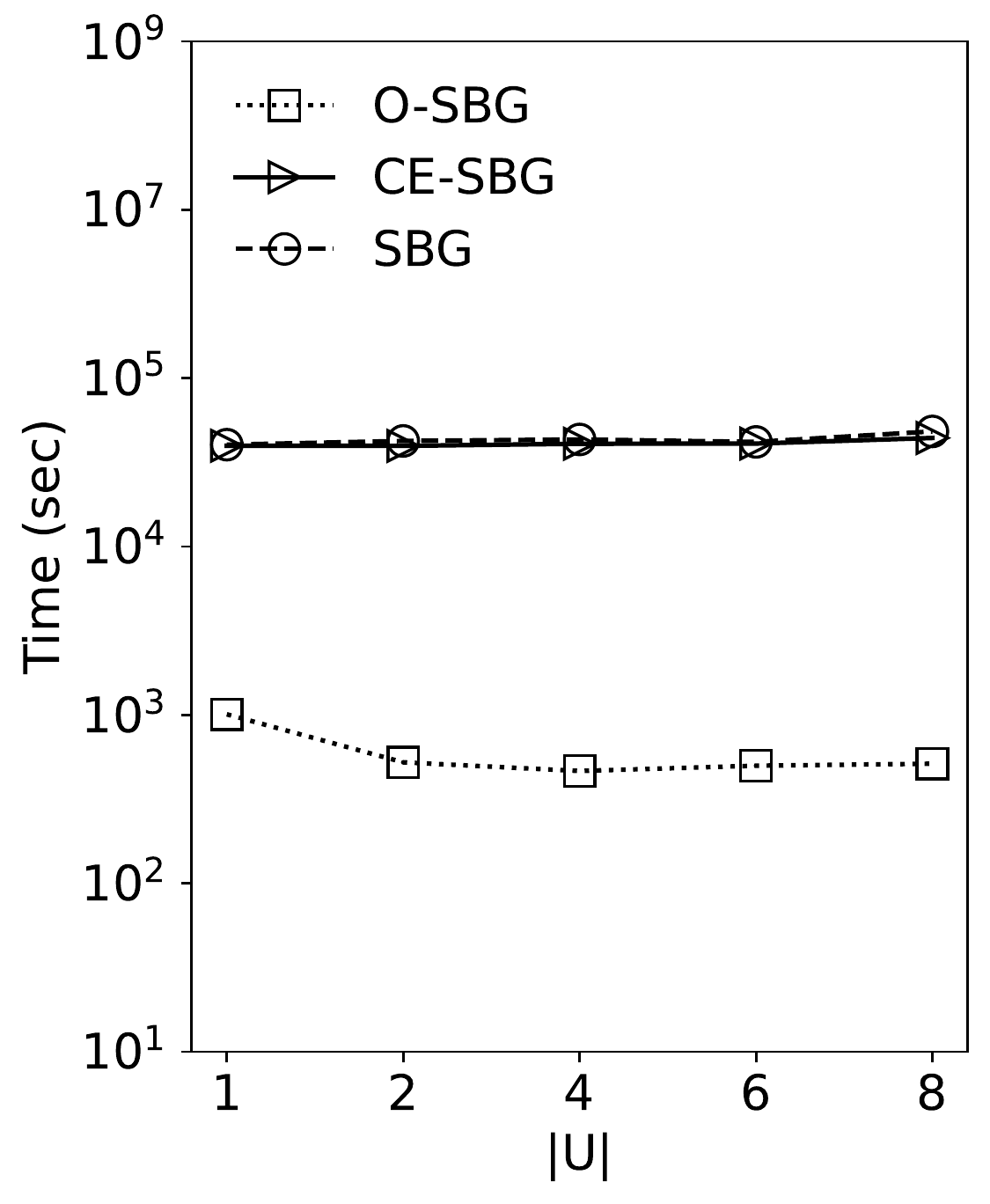}
	}
	\vspace{-3mm}
	\caption{Time cost of algorithms with varying $|\mathcal{U}|$}
%	\vspace{-2mm}
	\label{fig:running_time_U}
\end{figure}

Figure~\ref{fig:running_time_U} shows the running time of the approaches by varying the size of users group $\mathcal{U}$ from $1$ to $8$. The results show similar findings that \textit{O-SBG} outperforms \textit{CE-SBG} and \textit{SBG} as it utilizes the two step bounds to significantly reduce the probing candidate edges. For example, \textit{O-SBG} can reduce the running time by around $150$ times and $31$ times compared with \textit{SBG} and \textit{CE-SBG} respectively under different $|\mathcal{U}|$ settings on the \textit{mathoverflow} dataset.  

%the speed of running time increasing in \textit{O-SBG} is slower than the other two algorithms when $\mathcal{U}$ increases.

\subsubsection{Varying Snapshot Size $T$}

\iffalse
\begin{figure}
    \centering
    \includegraphics[scale=0.30]{figures/varying_ss.png}
    \caption{Varying $T$}
    \label{fig:T}
\end{figure}
\fi

\begin{figure}[!tb]
	\centering
	\subfigure[mathoverflow]{\label{R3:exp:vary_T1}
		\includegraphics[scale=0.30]{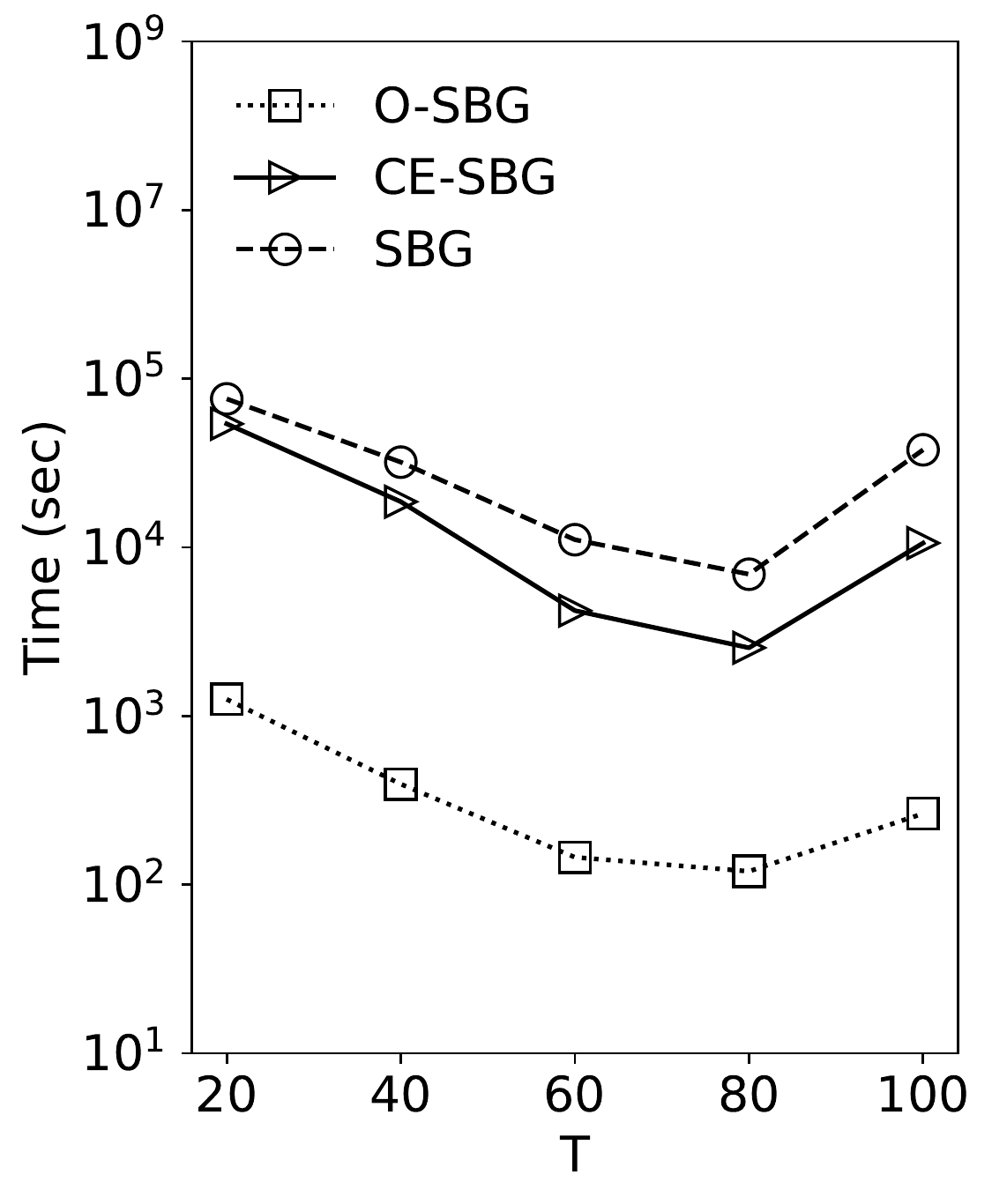}
	}
	\subfigure[ask-ubuntu]{\label{R3:exp:vary_T2}		
		\includegraphics[scale=0.30]{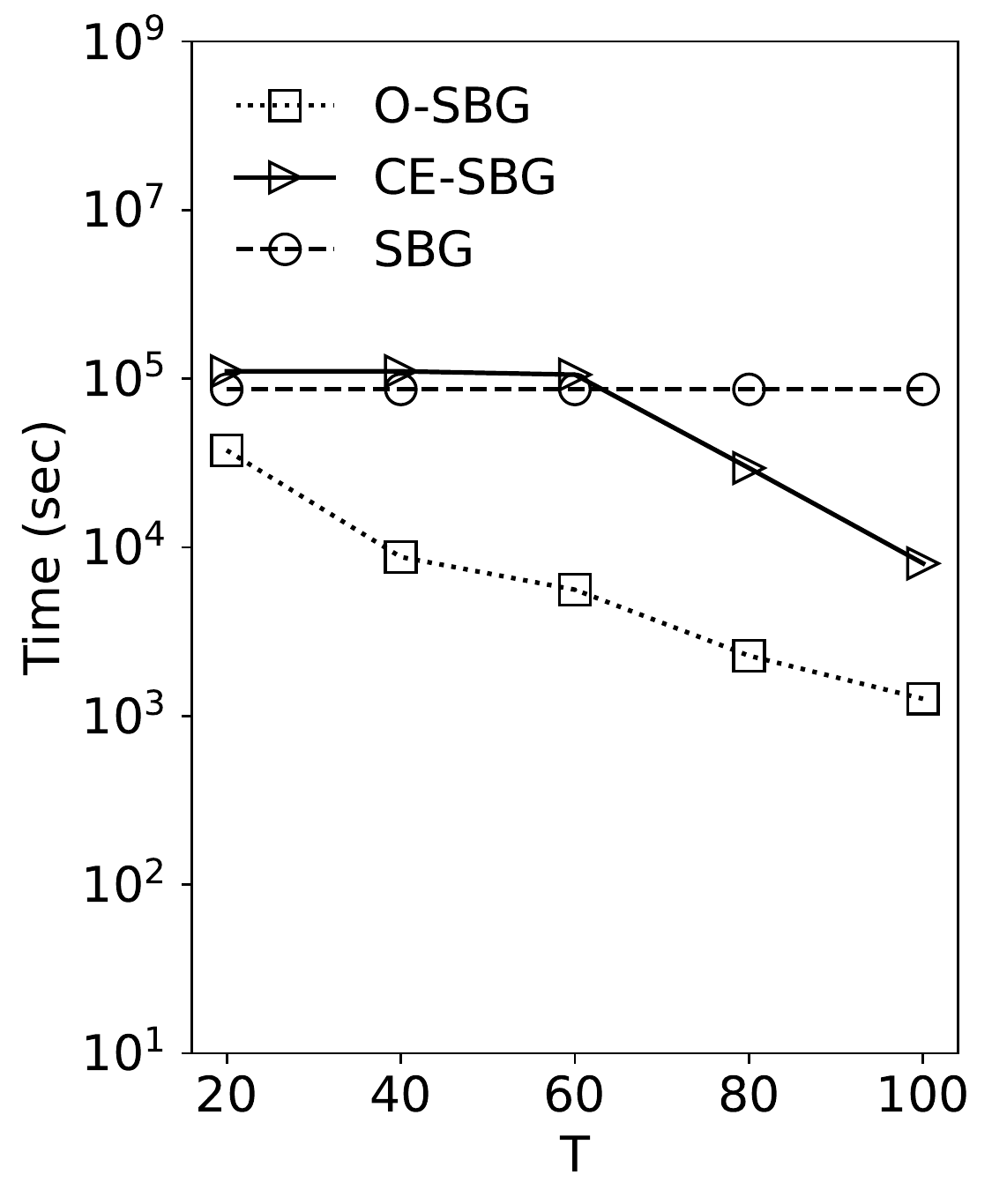}
	}
	\subfigure[CollegeMsg]{\label{R3:exp:vary_T3}		
	\includegraphics[scale=0.30]{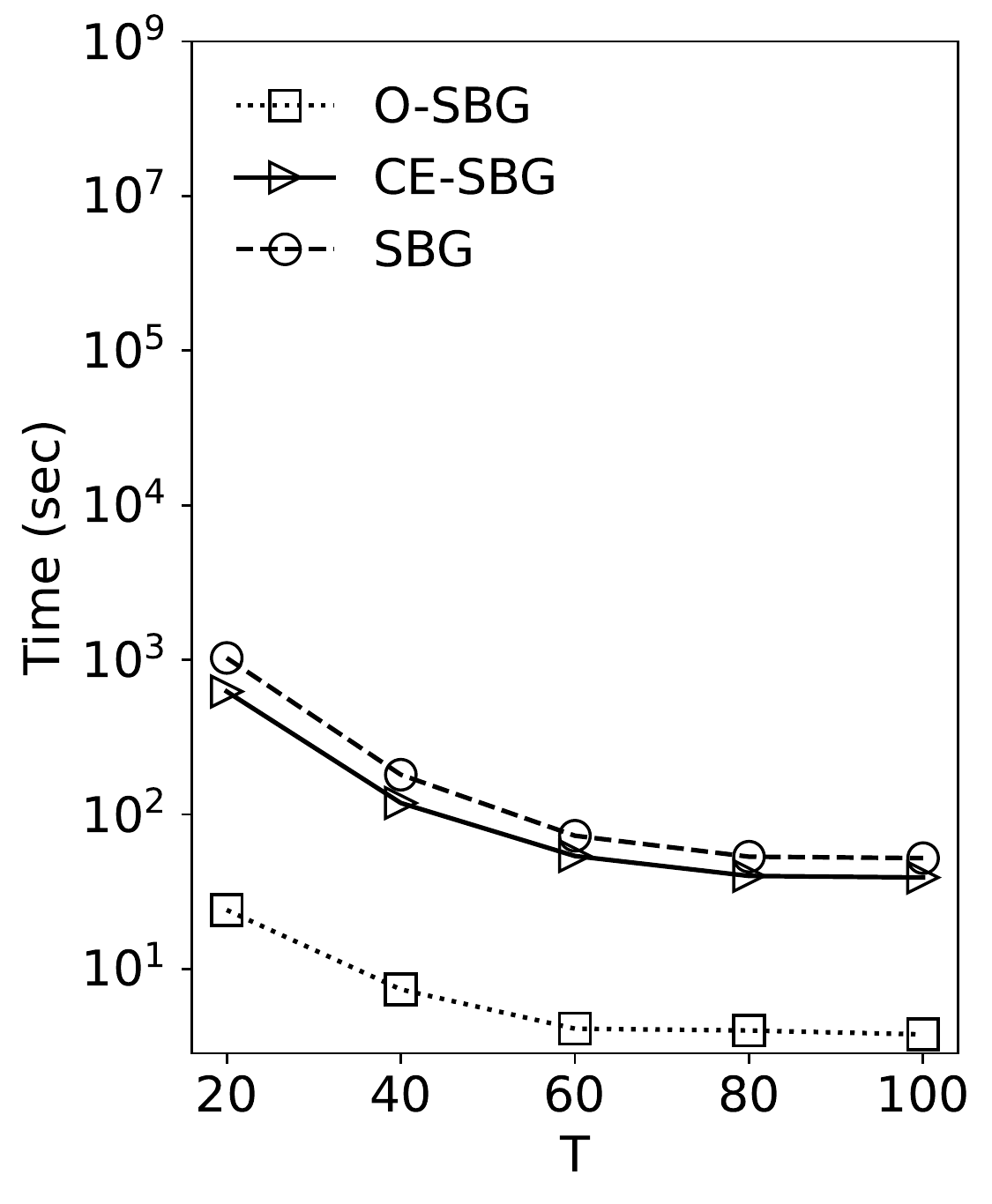}
	}
	\subfigure[eu-core]{\label{R3:exp:vary_T4}		
	\includegraphics[scale=0.30]{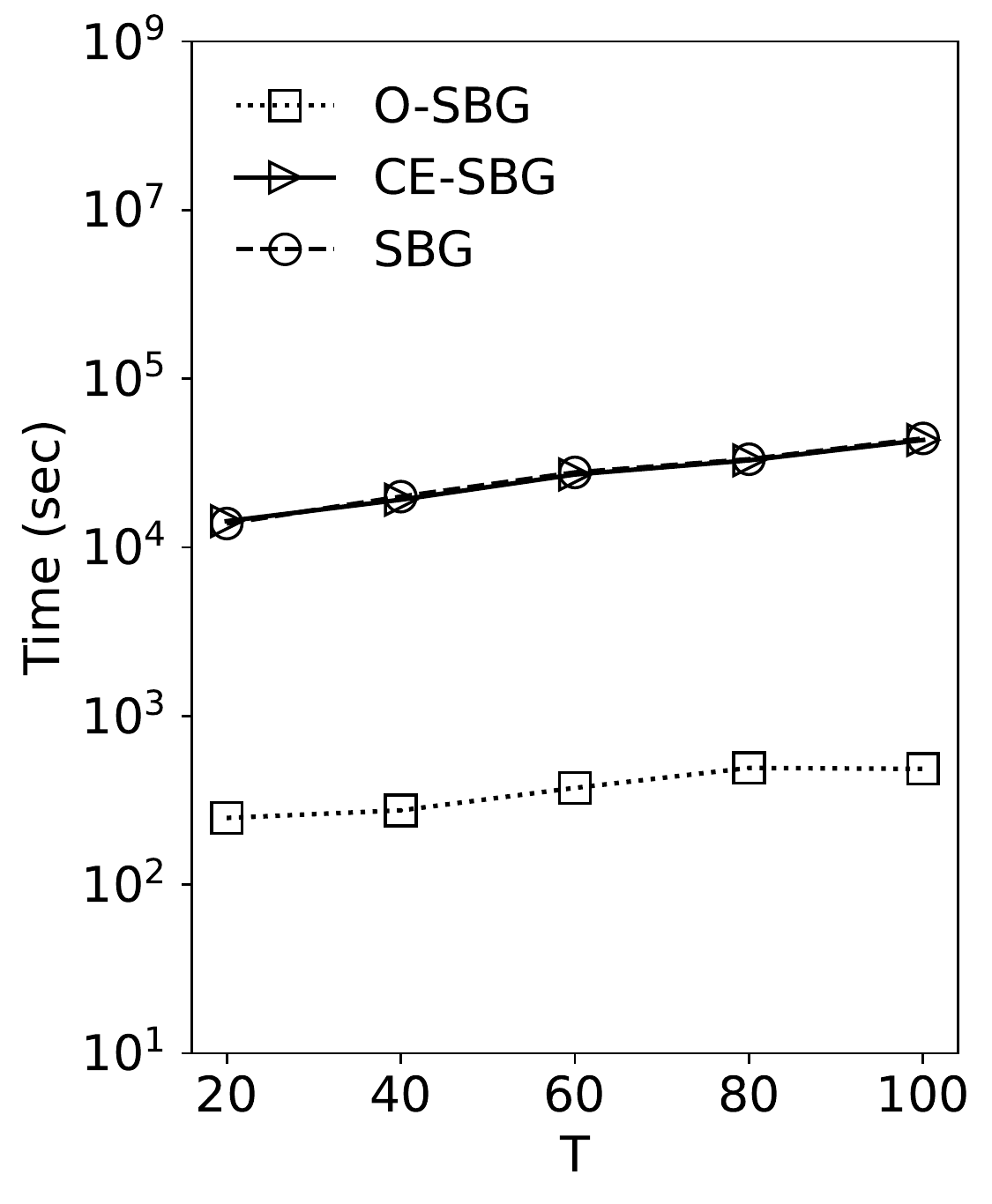}
	}
	\vspace{-3mm}
	\caption{Time cost of algorithms with varying $T$}
%	\vspace{-2mm}
	\label{fig:running_time_T}
\end{figure}

\begin{figure}[!tbp]
    \centering
    \includegraphics[scale=0.65]{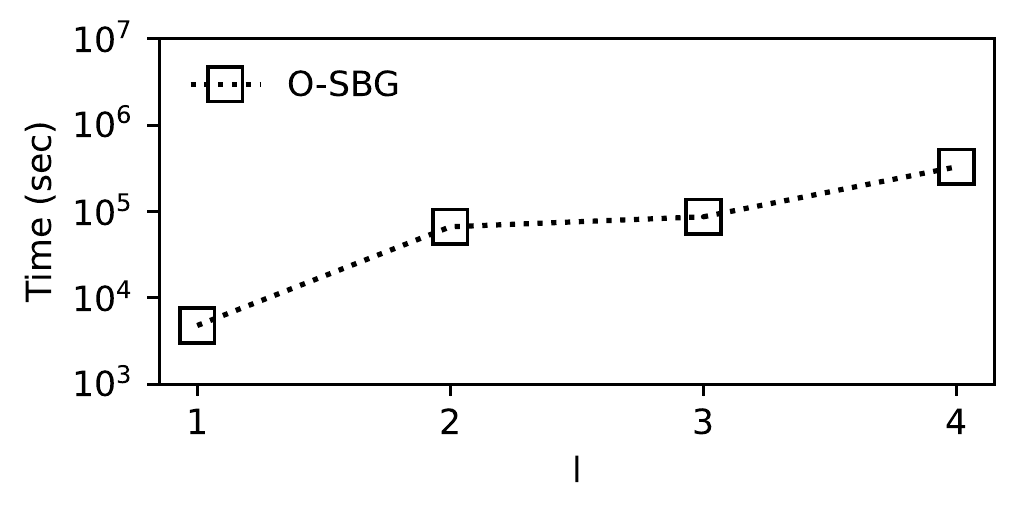}
  \vspace{-3mm}
    \caption{Performance of O-SBG on stack-overflow}
%    \vspace{-3mm}
    \label{fig:Scal}
\end{figure}
We compare the efficiency of our proposed algorithms by varying the graph snapshots size $T$ from $20$ to $100$. Figure~\ref{fig:running_time_T} presents the running time with varied values of $T$. The results show similar finding that \textit{O-SBG} outperforms \textit{SBG} and \textit{CE-SBG} in all datasets. Besides, we notice a similar running time trend in the proposed three methods when $T$ varies. Note that the running time does not always keep the same correlation with the varies of $T$. This is because the performance of all three proposed approaches highly depends on the graph structure, and the number of snapshots does not 
%be observed 
show 
a perceptible effect on the network structure.

\subsubsection{Performance in the Hyper Scale Networks.}

We further study the performance of different approaches on \textit{mathoverflow}, which is a huge dataset with $2,464,606$ nodes and $17,823,525$ edges. It is noticed that \textit{SBG} and \textit{CE-SBG} cannot get results in a valid time period on \textit{mathoverflow}, while \textit{O-SBG} can get the results in a valid period by varying $l$ from $1$ to $4$. Figure~\ref{fig:Scal} reports the average running time of \textit{O-SBG} on \textit{mathoverflow}. As we can see, the running time of \textit{O-SBG} scales linearly 
%increasing 
%while 
with the increase of $l$.
%increases.

%\subsubsection{Number of Probing edges}

\subsection{Effectiveness Evaluation}
In this experiment, we evaluate the number of expanding influence users produced by the RT$l$L problem with different datasets and approaches in Figure~\ref{fig:influenced_users_Q} - Figure~\ref{fig:influenced_users_l} by varying one parameter and setting the others as defaults. As can be seen, the average number of influenced users of RT$l$R queries in dense graphs is significantly larger than in sparse graphs for all three approaches. Figure~\ref{fig:influenced_users_Q} shows the average number of influenced users of all three approaches \textit{O-SBG}, \textit{CE-SBG}, and \textit{SBG} on four datasets 
%and 
with 
varying $Q$. For example, in Figure~\ref{R3:exp:influenced_users_Q1}, \textit{O-SBG}, \textit{CE-SBG}, and \textit{SBG} algorithms return back $39$, $23$, $20$ number of influenced users on average when $Q=20$ in \textit{mathoverflow} (\textit{i.e., $nodes = 21,688$, temporal edges $= 107,581$, average degree $= 4.96$}), respectively. Meanwhile, in Figure~\ref{R3:exp:influenced_users_Q4}, \textit{O-SBG}, \textit{CE-SBG}, and \textit{SBG} algorithms return back $102$, $165$, $164$ number of influenced users on average when $Q=20$ in \textit{eu-core} (\textit{i.e., $nodes = 986$, temporal edges $= 332,334$, average degree $= 25.28$}), respectively.
%by using \textit{O-SBG} is significantly large than using the other two approaches. For example, as we can see in Figure~\ref{fig:influenced_users_Q}, \textit{O-SBG}, \textit{CE-SBG}, and \textit{SBG} algorithms return back $49$, $38$, $35$ number of influenced users on average when $Q=20$ in \text , respectively. 
Similar pattern can also be found in Figure~\ref{fig:influenced_users_U} - Figure~\ref{fig:influenced_users_l} as more influenced users be returned in dense graphs than in sparse graphs.  %by using \texit{O-SBG} than the other two approaches. 
In addition, Figure~\ref{fig:influenced_users_U} reports that the influenced users of all three approaches do not always keep the same correlation with the increases of $\mathcal{U}$.
Figure~\ref{fig:influenced_users_l} shows that the number of influenced users by all three approaches significantly increases when $l$ changes from $1$ to $40$. For example, the numbers of influenced users by \textit{O-SBG}, \textit{CE-SBG}, and \textit{SBG} when setting $l$ as $40$ are $23$ times, $11$ times, and $15$ times larger than setting $l$ as $1$ in the \textit{mathoverflow} dataset. From the above experimental results, we can conclude that reconnecting the top-$l$ relationship query is necessary to maximize the benefits of expanding the influenced users of a given group.    

%\subsubsection{Varying Number of Queries $Q$}
%We compare the performance of different approaches by varying $Q$...

\iffalse
\begin{figure}
    \centering
    \includegraphics[scale=0.30]{figures/quality_nq.png}
    \caption{Quality Number of Queries}
    \label{fig:qu_Q}
\end{figure}
\fi

\begin{figure*}[htbp]
	\centering
	\subfigure[mathoverflow]{\label{R3:exp:influenced_users_Q1}
		\includegraphics[scale=0.29]{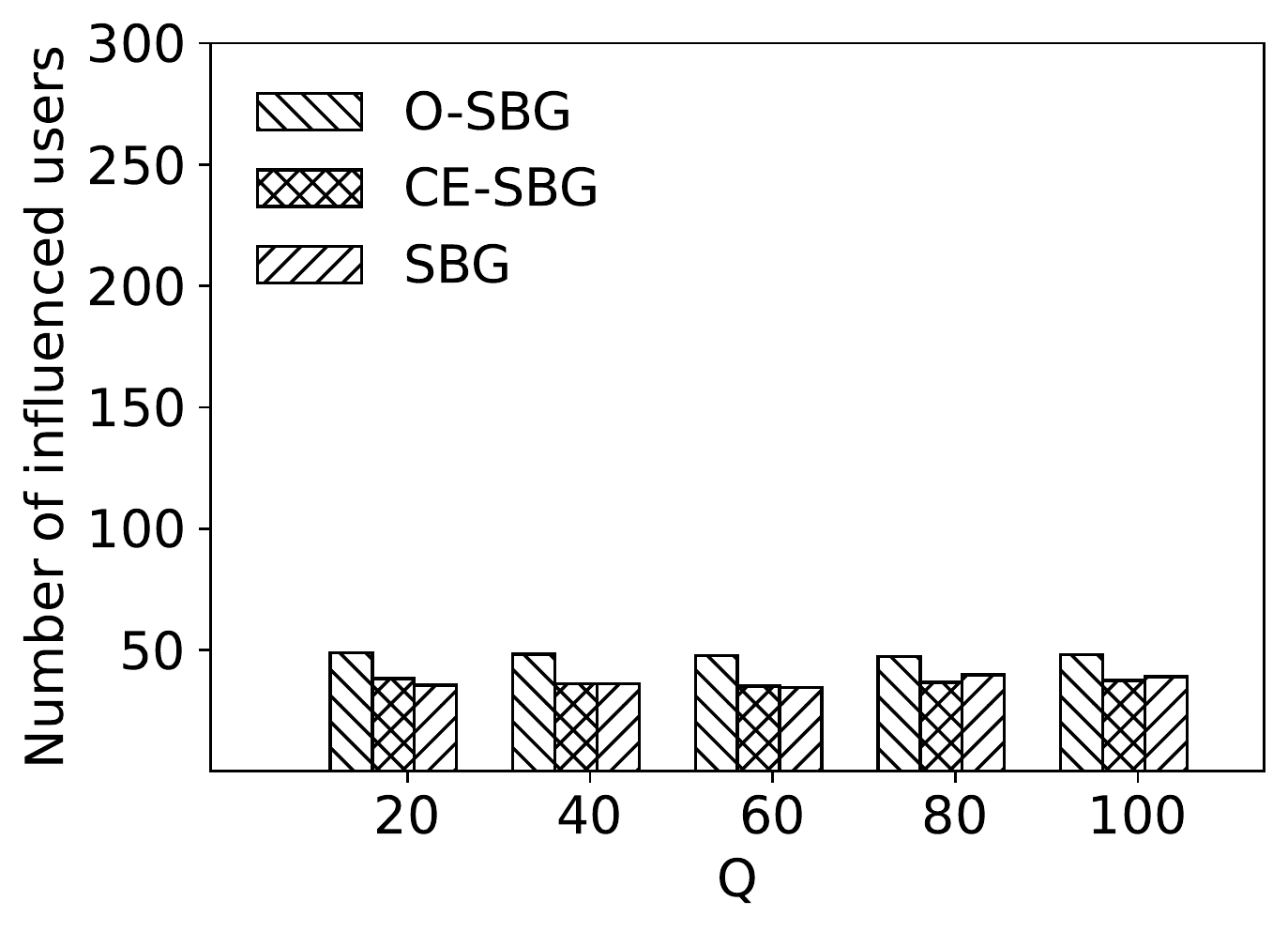}
	}
	\subfigure[ask-ubuntu]{\label{R3:exp:influenced_users_Q2}		
		\includegraphics[scale=0.29]{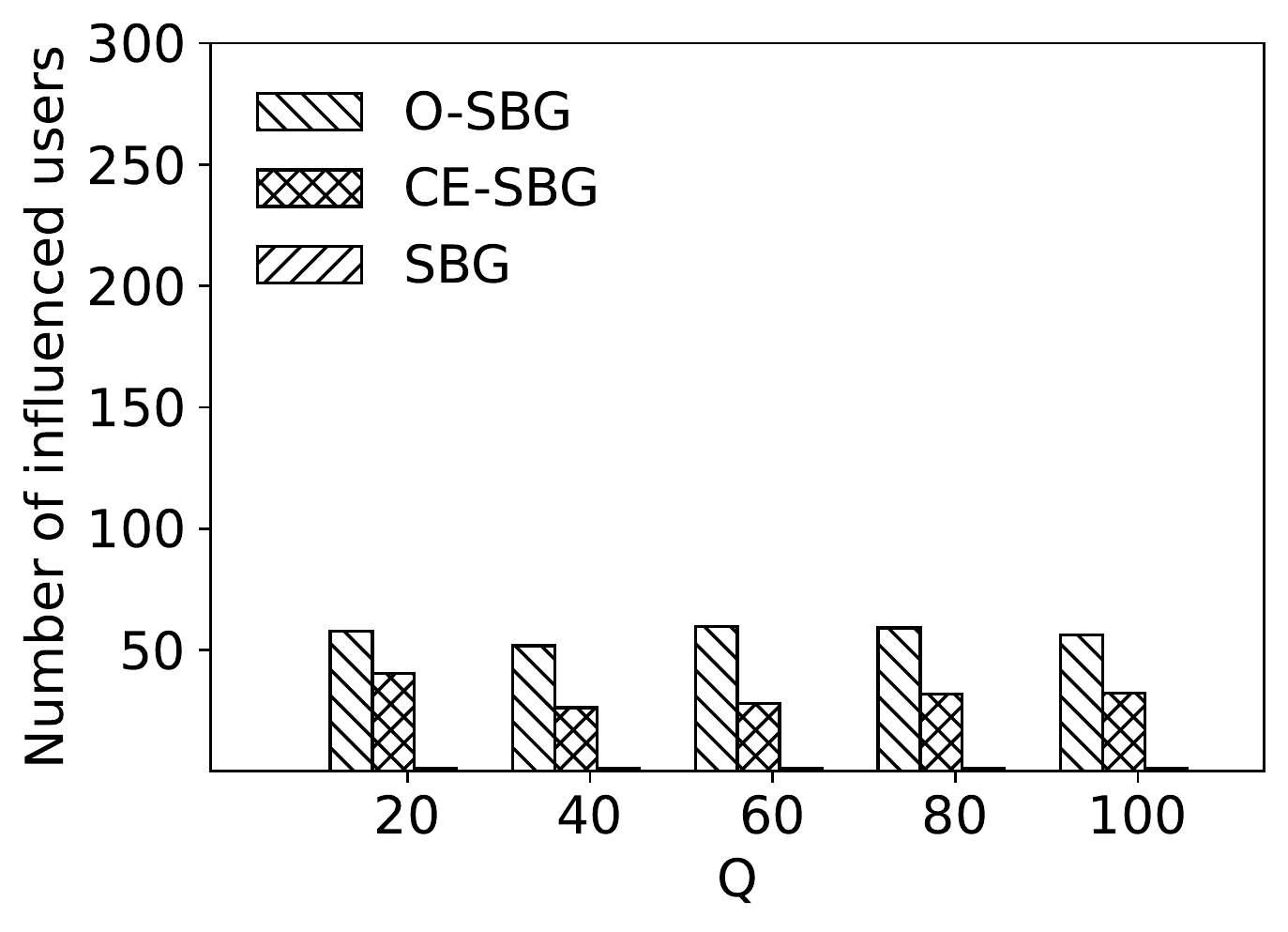}
	}
	\subfigure[CollegeMsg]{\label{R3:exp:influenced_users_Q3}		
	\includegraphics[scale=0.29]{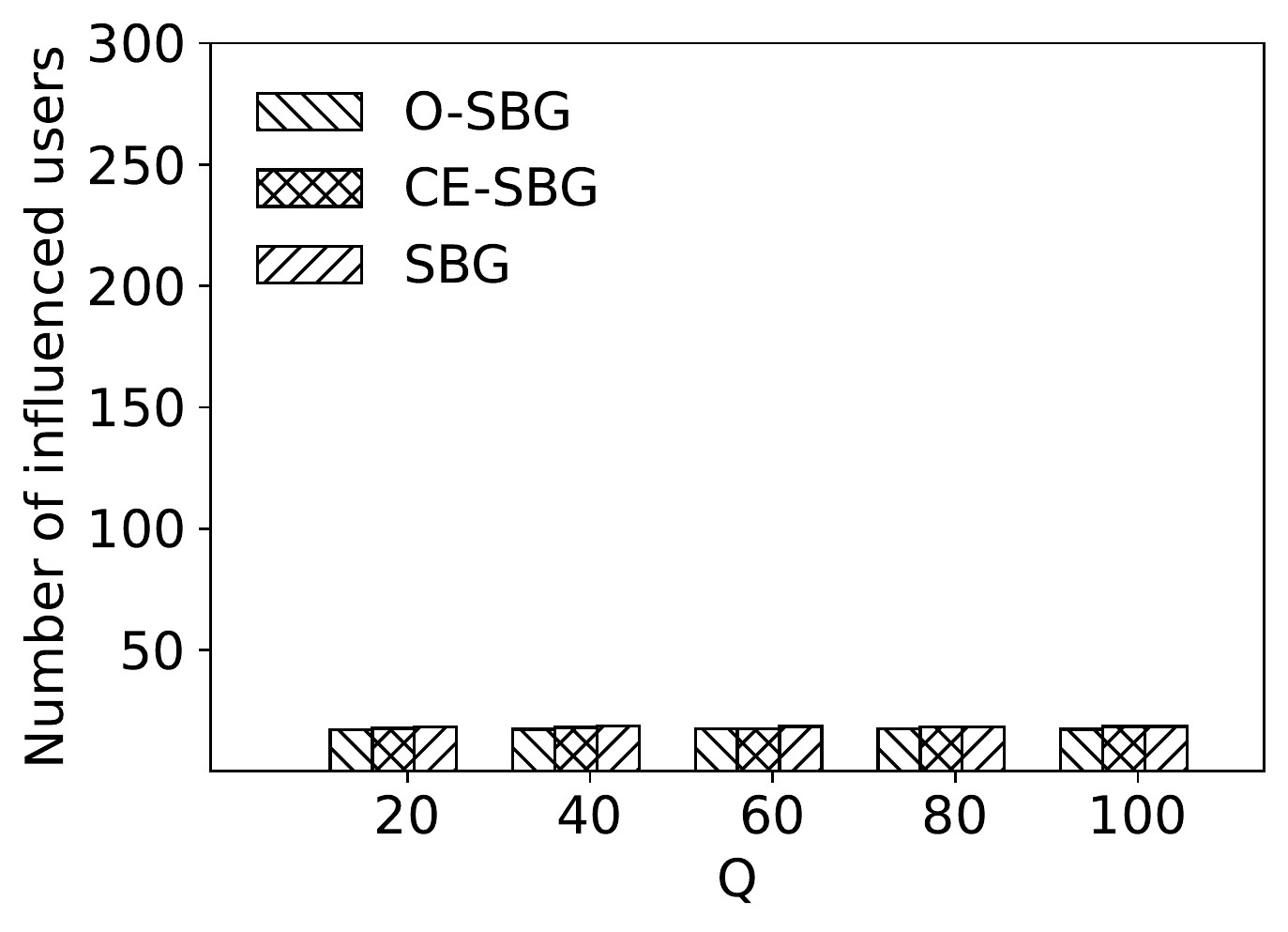}
	}
	\subfigure[eu-core]{\label{R3:exp:influenced_users_Q4}		
	\includegraphics[scale=0.29]{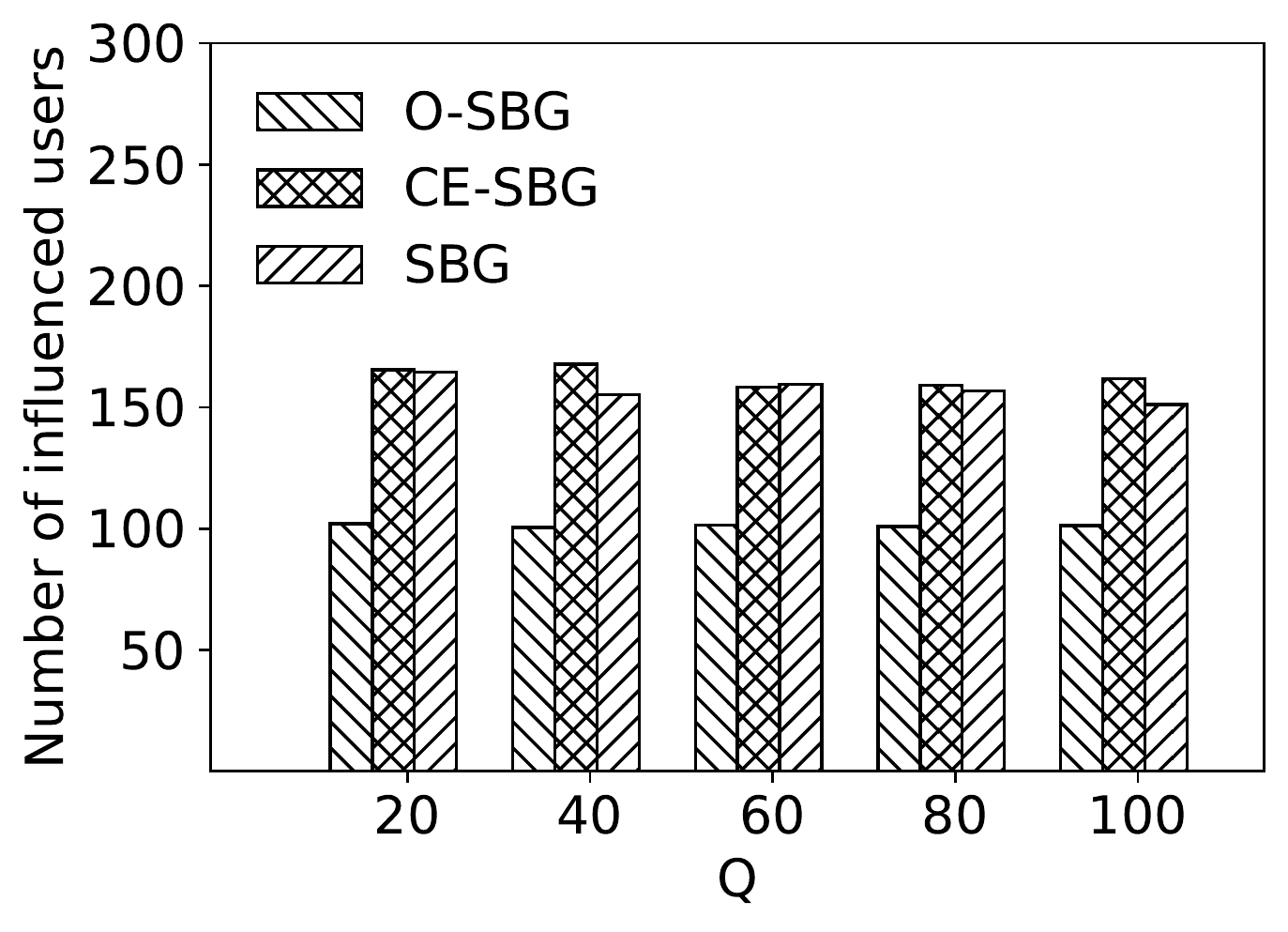}
	}
	\vspace{-4mm}
	\caption{Number of influenced users with varying $Q$}
	\vspace{-4mm}
	\label{fig:influenced_users_Q}
\end{figure*}

%\subsubsection{Varying Users Group Size $|\mathcal{U}|$ }

\iffalse
\begin{figure}
    \centering
    \includegraphics[scale=0.30]{figures/quality_us.png}
    \caption{Quality Users Group Size}
    \label{fig:qu_U}
\end{figure}
\fi

\begin{figure*}[htbp]
	\centering
	\subfigure[mathoverflow]{\label{R3:exp:influenced_users_U1}
		\includegraphics[scale=0.29]{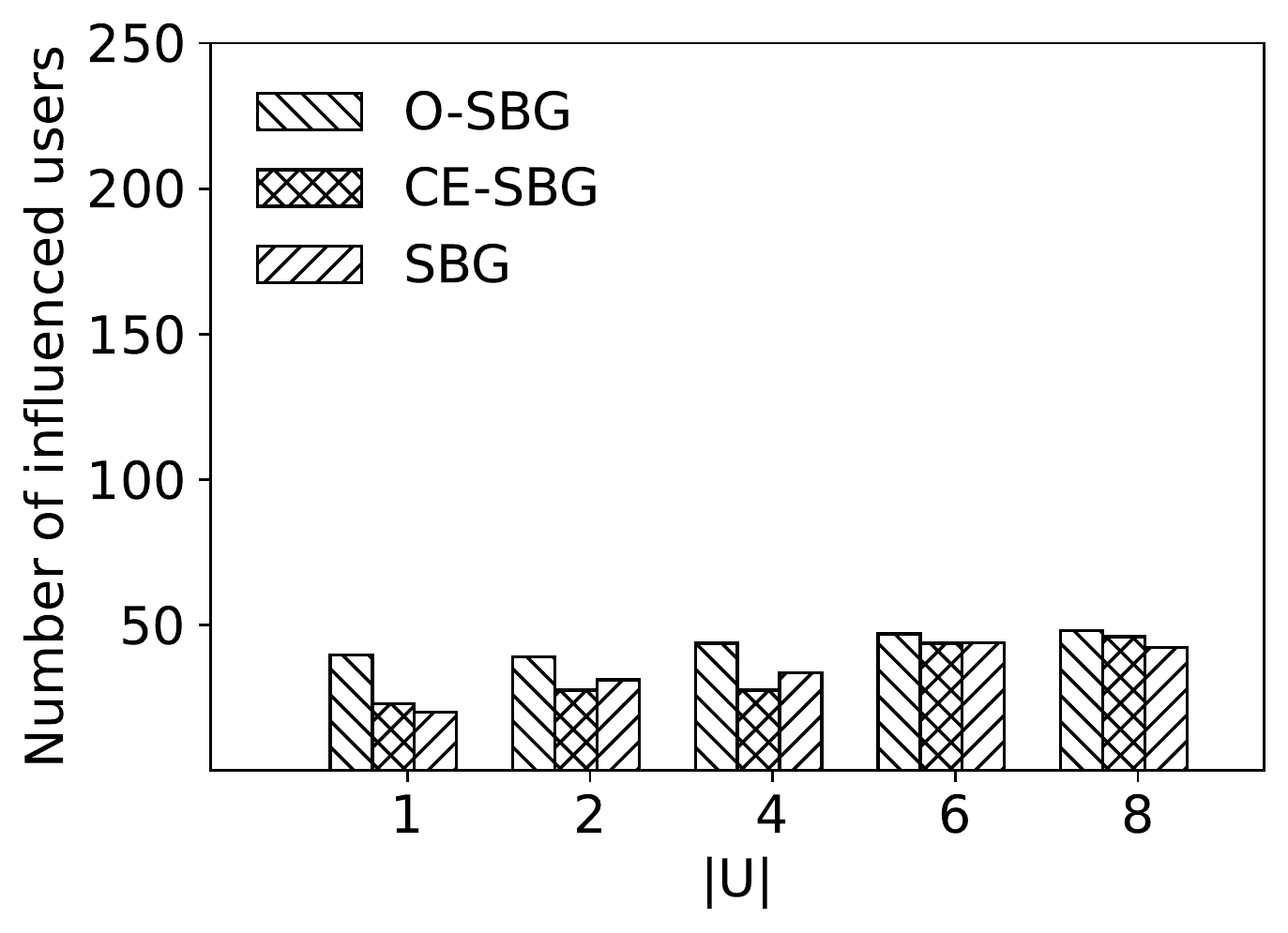}
	}
	\subfigure[ask-ubuntu]{\label{R3:exp:influenced_users_U2}		
		\includegraphics[scale=0.29]{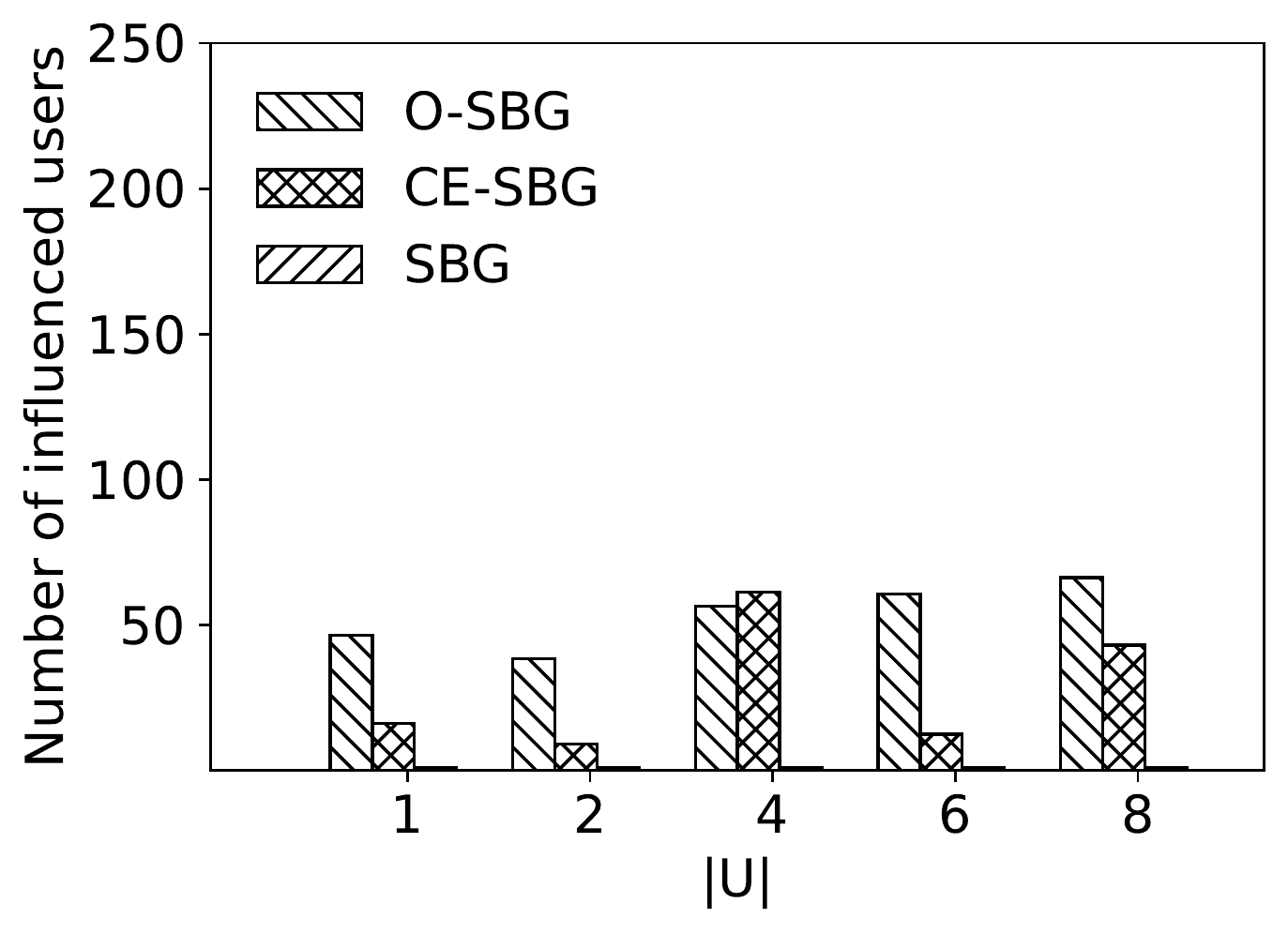}
	}
	\subfigure[CollegeMsg]{\label{R3:exp:influenced_users_U3}		
	\includegraphics[scale=0.29]{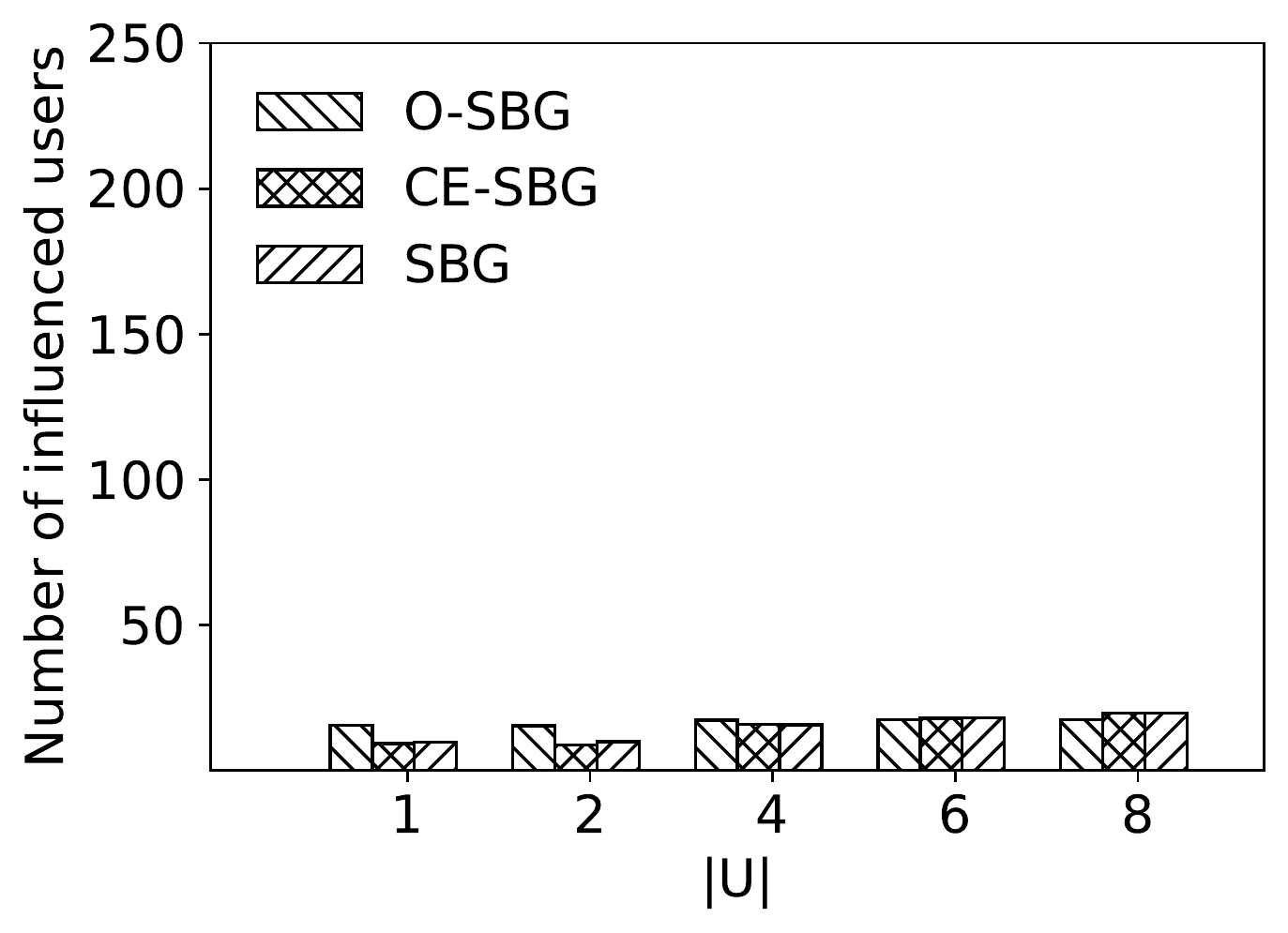}
	}
	\subfigure[eu-core]{\label{R3:exp:influenced_users_U4}		
	\includegraphics[scale=0.29]{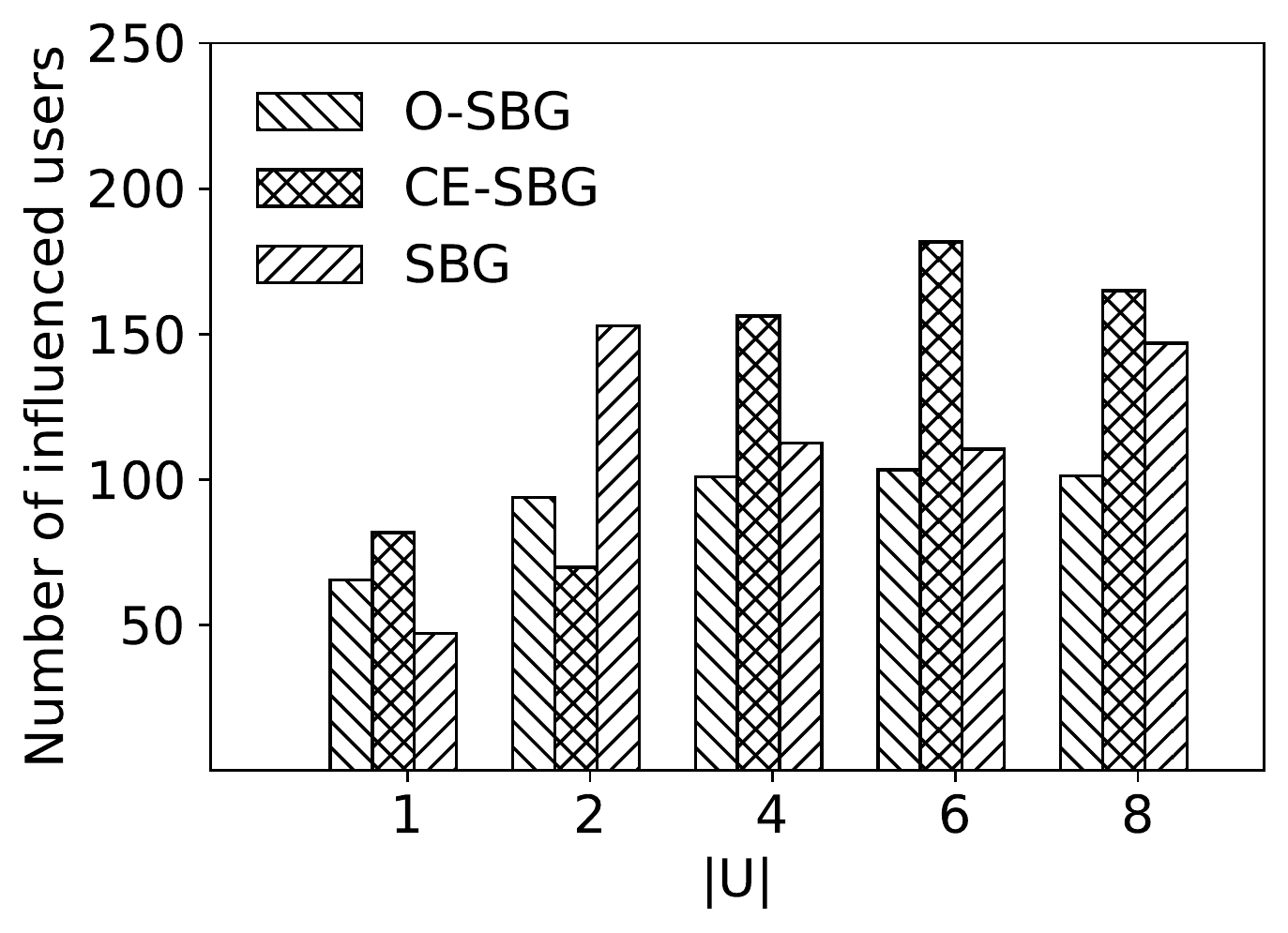}
	}
	\vspace{-4mm}
	\caption{Number of influenced users with varying $|\mathcal{U}|$}
	\vspace{-3mm}
	\label{fig:influenced_users_U}
\end{figure*}

%\subsubsection{Varying Reconnecting edges set size $l$ }
\iffalse
\begin{figure}
    \centering
    \includegraphics[scale=0.30]{figures/quality_re.png}
    \caption{Quality Reconnecting Edges Set Size}
    \label{fig:qu_L}
\end{figure}
\fi

\begin{figure*}[htbp]
	\centering
	\subfigure[mathoverflow]{\label{R3:exp:influenced_users_l1}
		\includegraphics[scale=0.29]{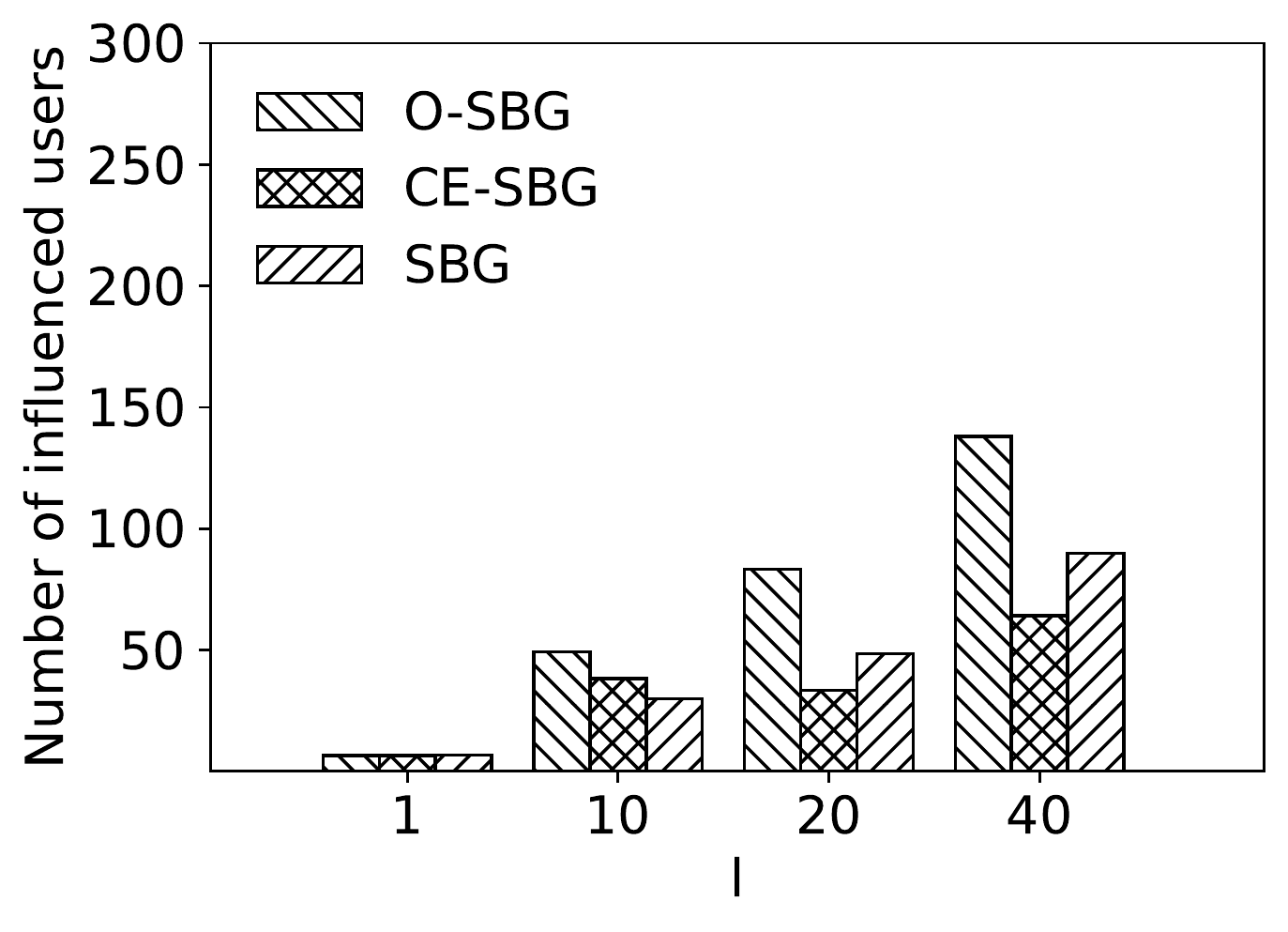}
	}
	\subfigure[ask-ubuntu]{\label{R3:exp:influenced_users_l2}		
		\includegraphics[scale=0.29]{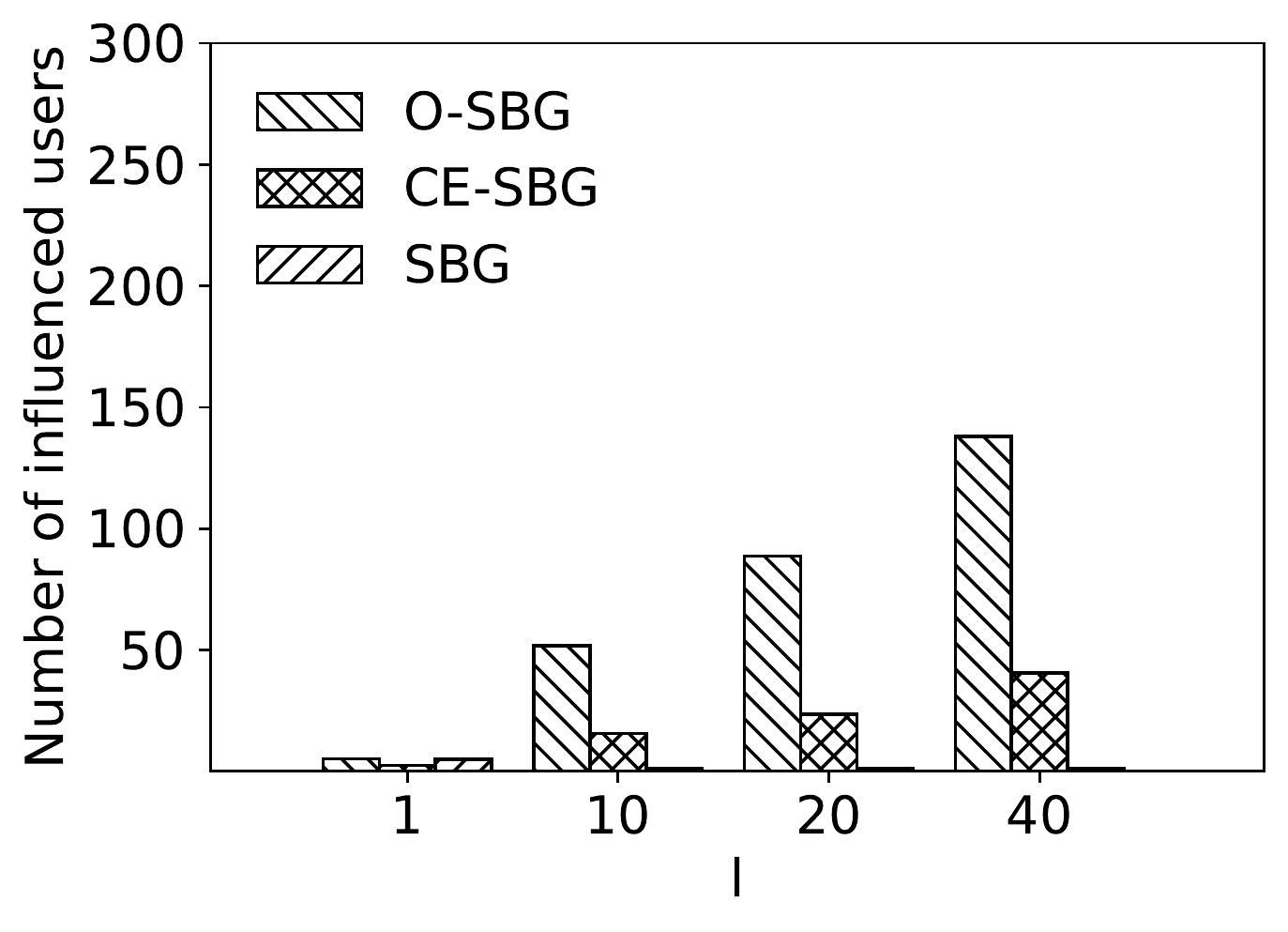}
	}
	\subfigure[CollegeMsg]{\label{R3:exp:influenced_users_l3}		
	\includegraphics[scale=0.29]{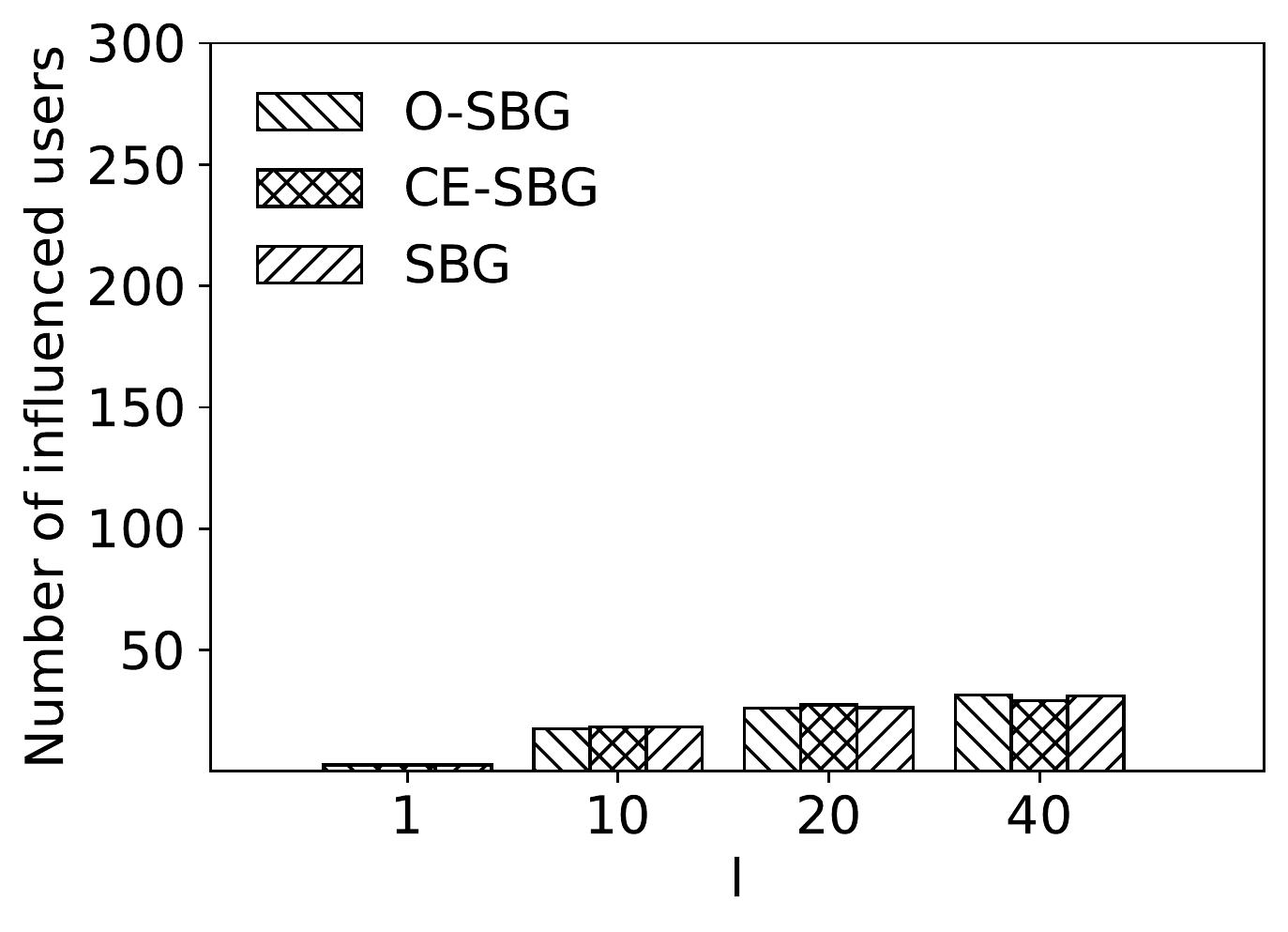}
	}
	\subfigure[eu-core]{\label{R3:exp:influenced_users_l4}		
	\includegraphics[scale=0.29]{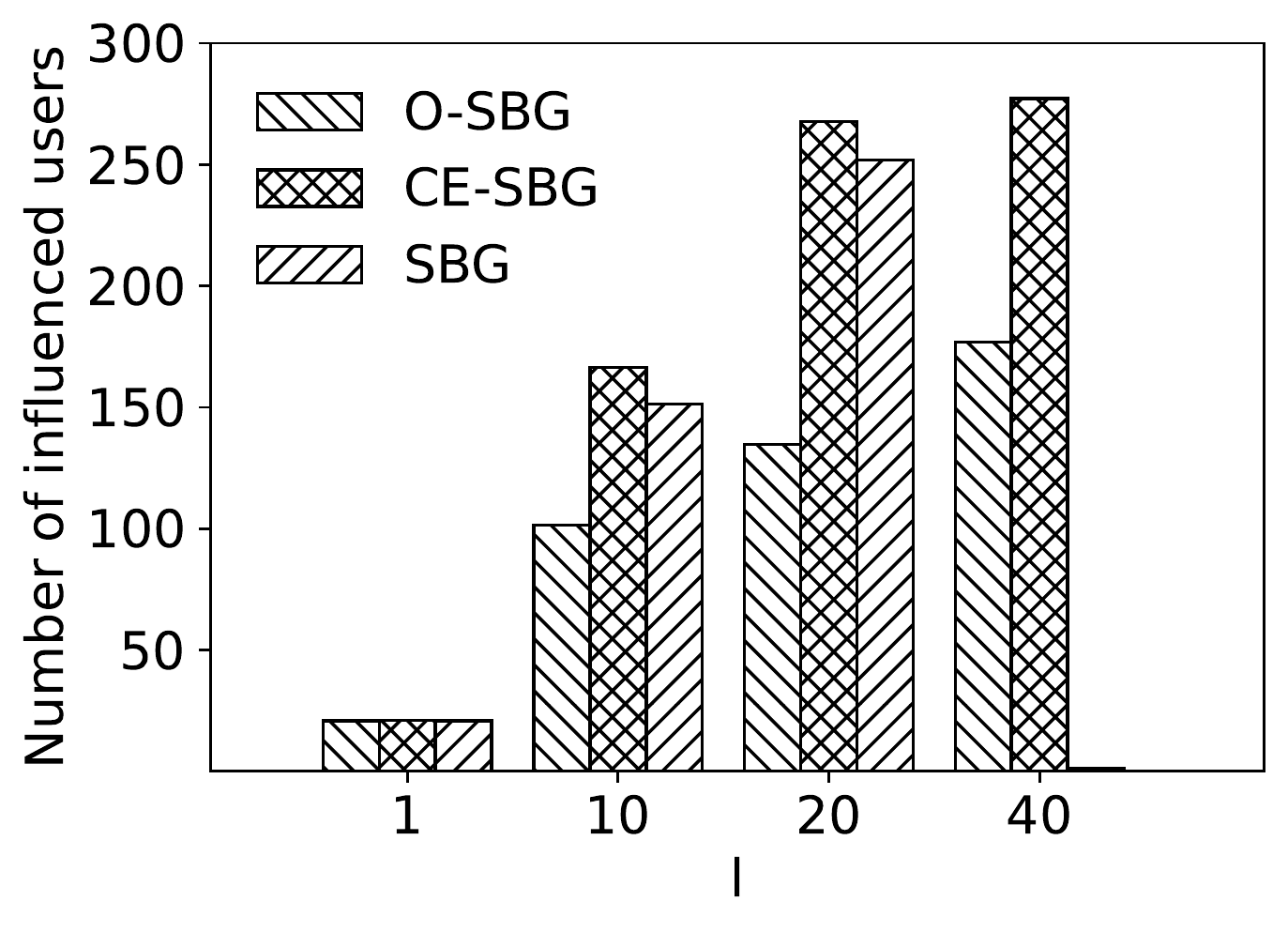}
	}
	\vspace{-4mm}
	\caption{Number of influenced users with varying $l$}
	\vspace{-3mm}
	\label{fig:influenced_users_l}
\end{figure*}

%\subsubsection{Varying Snapshot Size $T$}
\iffalse
\begin{figure}
    \centering
    \includegraphics[scale=0.30]{figures/quality_us.png}
    \caption{Quality Snapshot Size}
    \label{fig:qu_T}
\end{figure}
\fi

\iffalse
\subsection{Case Study}
In Figure~\ref{}, we visualize the selection results and its expanding influenced users of Top-$5$ RT$l$L query in \textit{eu-core} dataset.
\fi

%% file: Related_work.tex
\section{Related Work} \label{sec:related}

\subsection{Influence Maximization}
Influence maximization (IM) was first formulated by Domingos et al.~\cite{domingos2001mining} as an algorithmic problem in probabilistic methods. Later on, Kempe et al.~\cite{IM2003} modeled IM as an algorithmic problem in 2003. As the IM problem is NP-hard, all existing methods focus on approximate solutions, and a keystone of these algorithmic IM studies is the greedy framework. The existing IM algorithms can be categorized into three categories: {\em simulation-based}, {\em proxy-based}, and {\em sketch-based}.

\vspace{1mm}
\noindent
\textbf{Simulation-based approaches.} The key idea of these approaches is to estimate the influence spread $I(S)$ of given users set $S$ by using the \textit{Monte Carlo} (MC) simulations of the diffusion process~\cite{IM2003,leskovec2007cost,zhou2015upper}. 
Specifically, for a given users set $S$, the simulation-based approaches simulate the randomized diffusion process with $S$ for $R$ times. Each time they count the number of active users after the diffusion ends, and then take the average of these counts over the $R$ times. The accuracy of 
%this approach 
these approaches is positively associated with the number of $R$. 
The simulation-based approaches have the advantage of diffusion model generality, and these approaches can be incorporated into any classical influence diffusion model. However, the time complexity of these approaches are cost-prohibitive, which would hardly be used for dealing with sizeable networks.

\vspace{1mm}
\noindent
\textbf{Proxy-based approaches.} Instead of running heavy MC simulation, the proxy-based approaches estimate the influence spread of given users by using the proxy models. %Precisely, the proxy-based approaches are aligned with a specific diffusion model (\textit{e.g., IC model}), and the influence spread estimation process of given users is highly accelerated by taking advantage of the properties of the corresponding models. %
Intuitively, there are two branches of the proxy-based approaches, including (1) Estimate the influence spread of given users by transforming it to easier problems (\textit{e.g., Degree and PageRank})~\cite{chen2010scalable,galhotra2016holistic}; and (2) Simplify the typical diffusion model (\textit{e.g., IC model}) to a deterministic model (\textit{e.g., MIA model})~\cite{chen2010scalable} or restrict the influence propagation range of given users under the typical diffusion model to the local subgraph~\cite{goyal2011simpath}, to precisely compute the influence spread of given users. 
Compared with the simulation-based approach, 
%the 
a 
proxy-based approach offers significant performance improvements but lacks theoretical guarantees.

\vspace{1mm}
\noindent
\textbf{Sketch-based approaches.} To avoid running heavy MC simulations and reserve the theoretical guarantee, the sketch-based approaches~\cite{RIS14,tang2014influence,nguyen2017importance,10.1145/2661829.2662077,10.1145/2505515.2505541,ohsaka2014fast} pre-compute a number of sketches under a specific diffusion model, and then speed up the influence evaluation based on the constructed sketches. Compared with the simulation-based approaches, the sketch-based approaches have a lower time complexity under a theoretical guarantee. Unfortunately, the sketch-based approaches are not 
%general 
generic 
to all diffusion models because the generated sketches of the sketch-based approaches are relay on the underlying diffusion models.  
%and then evaluate the influence spread of given users set $S$ by exploiting the portion of RR-sketches which $ S$ can reach. In there, the RR-sketch is constructed by randomly selecting a node $u\in V$, and conducting a reverse MC sampling start from $u$. By repeatedly executing $\theta$ times of the above operation, we can get the RR sketches.  

\vspace{-2mm}
\subsection{Link Prediction}
Link prediction (LP) is an important network-related problem, first proposed by Liben-Nowell et al.~\cite{DBLP:conf/cikm/Liben-NowellK03}. The LP problem aims to infer the existence of new links or still unknown interactions between pairs of nodes based on the currently observed links. After decades study, a series of LP methods were proposed, including: similarity approaches~\cite{zhou2021experimental,he2015owa}, probabilistic approaches~\cite{das2017probabilistic,wang2017link}, hybrid approaches~\cite{wang2018fusion,zhang2020hybrid}, and deep learning approaches~\cite{rahman2018dylink2vec,zhang2018link,zhang2021labeling}. 

In this paper, we use the SEAL method~\cite{zhang2018link,zhang2021labeling} to predict the structure of the near future (\textit{i.e., time point $t$}) snapshot graph (\textit{i.e., $G_t$}) for a given evolving graph. Furthermore, for each given users group $\mathcal{U}$, our RT$l$R query problem aims to reconnect a set of edges in $G_t$ to maximize the number of influenced users of $\mathcal{U}$ in $G_t$, which is quite distinct from all existing IM works.

%% file: Conclusion.tex
\section{Conclusion} \label{sec:conclusion}
In this paper, we studied the problem of \textit{Reconnecting Top-$l$ Relationships} (RT$l$R), which aims to find $l$ previous existing relationships but being estranged 
subsequently, 
%in the near future, 
such that reconnecting these relationships 
%will 
would 
maximize the influence spread of given users group. We have shown that the RT$l$L query problem is NP-hard. We developed a FI-Sketch based greedy (SBG) algorithm  to solve this problem. We further devised an edge reducing method to prune the candidate edges that the given users' group cannot reach.
Moreover, an order-based SBG method has been designed by utilizing the submodular characteristic of the RT$l$L query and two well-designed upper bounds.    
Lastly, the extensive performance evaluations on real datasets also revealed the practical efficiency and effectiveness of our proposed method. 
In the future, we will focus on developing more efficient approaches to deal with the RT$l$R queries in hyper scale networks.